%% file: BiDirectMUJourn.tex
\documentclass[journal]{IEEEtran}
\usepackage{graphicx, amsmath, amsfonts, epsfig, ntheorem, hyperref, subfigure, color}
  \usepackage{nicefrac}
  \usepackage{psfrag}
  \usepackage{manfnt}

\usepackage{amssymb,eucal,graphicx}
\usepackage{epsfig}
\usepackage{exscale}
\usepackage{latexsym}
\usepackage{verbatim}
\usepackage{amstext}
\usepackage{latexsym}
\usepackage{ifthen}
\usepackage{multirow}
\usepackage{flushend}

\usepackage{dsfont}

\newif\iftodo   
\todotrue  
\newif\iftodoshort  
\todoshortfalse

\theoremstyle{plain}
\theoremheaderfont{\normalfont\bfseries}\theorembodyfont{\normalfont}
\theoremsymbol{}
\theoremseparator{:}
\newtheorem{theorem}{Theorem}
\theoremheaderfont{\normalfont}\theorembodyfont{\normalfont}
\theoremstyle{nonumberplain}
\theoremseparator{}
\theoremsymbol{}

\theoremstyle{plain}
\theoremseparator{:}
\theoremheaderfont{\normalfont}\theorembodyfont{\normalfont}

\theoremstyle{plain}
\theoremheaderfont{\normalfont\bfseries}\theorembodyfont{\normalfont}
\theoremsymbol{}
\theoremseparator{:}
\newtheorem{coro}{Corollary}
\theoremstyle{plain}
\theoremheaderfont{\normalfont\bfseries}\theorembodyfont{\normalfont}
\theoremsymbol{}
\theoremseparator{:}
\newtheorem{lemma}{Lemma}
\theoremstyle{plain}
\theoremheaderfont{\normalfont\bfseries}\theorembodyfont{\normalfont}
\theoremsymbol{}
\theoremseparator{:}
\newtheorem{defn}{Definition}[section]

\newcommand{\beq}{\begin{equation}}
\newcommand{\eeq}{\end{equation}}

\newcommand{\lp}{ \left(}
\newcommand{\rp}{ \right)}

\long\def\comment#1{}

\newfont{\bbb}{msbm10 scaled 700}

\newfont{\bb}{msbm10 scaled 1100}

\begin{document}

\title{Divide-and-conquer: Approaching the capacity of the two-pair
bidirectional Gaussian relay network
\thanks{Manuscript received January 18, 2010; revised April 5, 2011. Parts of this paper have been presented at ITW'09~\cite{AvestimehrITW09} and ISIT'09~\cite{SezKhaAveHasISIT09}, respectively.
        The associate editor coordinating the review of this
manuscript and approving it for publication was Dr. Suhas Diggavi.}
\thanks{
Aydin Sezgin has been with the Emmy-Noether-Research Group on Wireless Networks, TAIT, Ulm University, Germany. He is now with the Ruhr-University Bochum, 44801 Bochum, Germany, email: aydin.sezgin@rub.de. The research of A. Sezgin was supported by the DFG Grant 1697/3.}
\thanks{
A.~Salman Avestimehr is with the School of Electrical and Computer Engineering, Cornell University, USA. email: avestimehr@ece.cornell.edu. The research of A. S. Avestimehr was supported in part by the NSF CAREER award 0953117.}
\thanks{
M.~Amin Khajehnejad and Babak Hassibi are with the Department
of Electrical Engineering, California Institute of Technology, Pasadena, CA, USA. email:\{amin,hassibi\}@caltech.edu}}

\author{Aydin Sezgin,~\IEEEmembership{Member,~IEEE}, A.~Salman Avestimehr,~\IEEEmembership{Member,~IEEE}, M.~Amin Khajehnejad, and Babak Hassibi}
\markboth{IEEE Transactions on Information Theory,~Vol.~xx,
No.~xx,~month~year}{Sezgin et.al.: Divide-and-conquer: Approaching the capacity of the two-pair
bidirectional Gaussian relay network}

\maketitle

\begin{abstract}
The capacity region of multi-pair bidirectional relay networks, in which a relay node facilitates the communication between multiple pairs of users, is studied. This problem is first examined in the context of the linear shift deterministic channel model.
The capacity region of this network when the relay is operating at either full-duplex mode or half-duplex mode for arbitrary number of pairs is characterized. It is shown that the cut-set upper-bound is tight and the capacity region is achieved by a so called divide-and-conquer relaying strategy. The insights gained from the deterministic network are then used for the Gaussian bidirectional relay network. The strategy in the deterministic channel translates to a specific superposition of lattice codes and random Gaussian codes at the source nodes and successive interference cancelation at the receiving nodes for the Gaussian network. The achievable rate of this scheme with two pairs is analyzed and it is shown that for all channel gains it achieves to within 3 bits/sec/Hz per user of the cut-set upper-bound. Hence, the capacity region of the two-pair bidirectional Gaussian relay network to within 3 bits/sec/Hz per user is characterized.
\end{abstract}

\section{Introduction}
Cooperative communication and relaying is one of the important research topics in wireless network information theory. The basic model to study this problem is the 3-node relay channel which was first introduced in 1971 by van der Meulen \cite{Meulen} and several strategies for this network were developed by Cover and El Gamal~\cite{CoverGamal}.

While the main focus so far has been on the one-way-relay channel, bidirectional communication has also attracted attention. Bidirectional (or two-way) communication between two nodes was first studied by Shannon himself in~\cite{ShannonInt}. Nowadays the bidirectional communication where an additional node acting as a relay is supporting the exchange of information between the two nodes (or one pair) is gaining increased attention. Some relaying strategies for the one-pair bidirectional relay channel, such as decode-and-forward, compress-and-forward and amplify-and-forward, have been analyzed in~\cite{RankovISIT}. An interesting strategy referred to as noisy network coding was proposed in~\cite{LimKimGamalChung}, which generalizes the compress-and-forward strategy in~\cite{CoverGamal}.

Network coding type techniques have been proposed also by~\cite{KattiI,HauslHagenauer,BaikChung,NaryananLatticeBiDi} (and others) in order to improve the transmission rate.
{In~\cite{KattiI}, a network coding approach is used for the first time in a wireless network in order to reduce the number of transmissions needed to exchange the number of data packets between two nodes of bidirectional setup. While before $4$ transmissions were needed, the number of transmissions was to reduce to 3 in~\cite{KattiI} resulting in higher data rates. The transmit strategy in~\cite{HauslHagenauer} is similar to~\cite{KattiI} with the extension that a channel code is used by the nodes when communicating to the relay. Once the data is received at the nodes, they perform iterative network and channel decoding resulting in higher rates than without network coding. In~\cite{BaikChung,NaryananLatticeBiDi} the number of transmissions is further reduced by allowing the nodes to submit their data simultaneously to the relay resulting in a multiple-access setup.  Additionally, \cite{BaikChung,NaryananLatticeBiDi} utilize the idea of network coding for the binary case to extend it to the Gaussian case by using lattice coding, which is referred to as physical layer network coding. In~\cite{NaryananLatticeBiDi}, it is shown that the lattice based scheme outperforms other schemes at high $\mathsf{SNR}$. It turns out, however, that decoding the individual data streams in the multiple-access hop gives better performance at lower $\mathsf{SNR}$. In~\cite{BaikChung} decode-and-forward, amplify-and-forward, and modulo-and-forward relaying strategies are compared in terms of transmission rate. It turns out that depending on the scenario, one of schemes outperforms the other two, i.e., neither one is always outperforming the other. The tightest gap characterization on the capacity for the two-way relay channel is provided in~\cite{NamChungLeeJournal}, where it is shown that upper and lower bounds only differ by 1/2 bit. }

\subsection{System under investigation}
The bidirectional relay channel problem discussed above can be generalized to a multi-pair (or multiuser) setting in which the relay facilitates the communication between multiple pairs of users. The achievable degrees of freedom for a three user case with multiple antennas were determined in~\cite{NamLimChun}. In~\cite{RankovMU} authors analyzed the case that the relay orthogonalizes different bidirectional transmissions by a distributed zero forcing algorithm and then multiple pairs communicate with each other via several orthogonalize-and-forward relay terminals. In \cite{ChenYener,ChenYenerCISS} authors investigated this problem for interference limited systems in which each pair of users share a common spreading signature to distinguish themselves from the other pairs, and proposed a jointly demodulate-and-XOR forward strategy. However, so far no attempt has been done to characterize the capacity region of this network, and the optimal relaying strategy is unknown.

In this paper we study the information theoretic capacity of the multi-pair bidirectional wireless relay network. We first examine this problem in the context of the linear shift deterministic channel introduced by Avestimehr, Diggavi, and Tse~\cite{AvestimehrIT}. This model simplifies the wireless network interaction model by eliminating the noise and allows us to focus on the interaction between signals. This approach was successfully applied to the relay network in \cite{AvestimehrIT}, and resulted in insight in terms of transmission techniques which also led to an approximate characterization of the capacity of Gaussian relay networks. This approach has also been recently applied to the bidirectional relay channel problem~\cite{AvestimehrSezgin,AvestimehrSezginETT}, which again resulted in finding near optimal relaying strategies as well as approximating the capacity region of the noisy (Gaussian) bidirectional relay channel. The deterministic approach is not restricted to relay networks. For instance, an approximate characterization of the capacity for the Gaussian interference channel was obtained in~\cite{BreslerTseETT} using the deterministic approach. Transmission techniques in a deterministic relay-interference network were studied in~\cite{MohajerDiggaviAl08}.

\subsection{Main contributions}
Inspired by the results mentioned above, we apply the linear shift deterministic model to the multi-pair bidirectional relay network and analyze its capacity when the relay is operating at either full-duplex mode or half-duplex mode (with non adaptive listen-transmit scheduling). In both cases we exactly characterize the capacity region and show that the cut-set upper-bound is tight.
{We show that the capacity region is achieved by dividing the
signal level space elegantly between the multiple pairs, i.e., different pairs are orthogonalized on the signal level space. Each
pair is then operating on the portion of the signal level space assigned
to it.  The relay uses a similar \emph{functional-forwarding} scheme as
in~\cite{AvestimehrSezgin}, in which the relay re-orders the received superposed
signals on the different levels and forwards them without decoding everything explicitly. The strategy is therefore referred to as divide-and-conquer-strategy.}

Later on, we use these insights to find a near optimal transmission technique for the Gaussian case. More specifically, we propose a superposition of lattice codes and random Gaussian codes at the source nodes. {However, orthogonalization as in the deterministic setup is not possible in the Gaussian setup as all signals arriving at the relay interact with each other. Thus the relay attempts to decode the Gaussian codewords of the respective nodes and the superposition of the lattice codewords of each pair by using successive interference cancelation. The relay then forwards this information to the intended destinations}. We analyze the achievable rate region of this scheme and show that for all channel gains it achieves to within 3 bits/sec/Hz per user of the cut-set upper-bound on the capacity region of the two-pair bidirectional relay network.

The paper is organized as follows. In Section~\ref{sec:Main} we investigate the full-duplex and half-duplex multi-pair bidirectional linear shift deterministic relay network and characterize the exact capacity region of this network. In Section~\ref{sec:Insights}, we discuss the insights gained from the linear shift deterministic model and how these insights can be used in the Gaussian setup in the subsequent Section~\ref{sec:2Pair2WayGaussRelay}. In the Gaussian two-pair bidirectional relay network, we present upper bounds, derive our achievability strategy and characterize the constant gap between the upper bounds and our proposed scheme. We finally conclude the paper in Section~\ref{sec:Conc}.

\section{Multi-pair bidirectional linear shift deterministic relay network}\label{sec:Main}
In the following subsections, we state the precise definition of the problem and present the main result for the deterministic case.

\subsection{System model \label{sec:sysModel}}
The system model for the $M$-pair bidirectional relay network is shown in Figure \ref{fig:sysModel}. In this system $M$ pairs  $(A_1,B_1),\ldots,(A_M,B_M)$ aim to use the relay to communicate with each other (i.e., $A_1$ and $B_1$ want to communicate with each other, and so on). The relay can operate on either full-duplex or half-duplex mode. In the full-duplex mode it is able to listen and transmit at the same time, while in the half-duplex mode it can only listen or transmit at a particular time. In the half-duplex scenario, we only consider the case that the listen-transmit scheduling is non-adaptive and the relay listens a fixed $\Delta$ fraction of the time and transmits the rest. Although $\Delta$ can not change adaptively as a function of the channel gains, one can optimize over $\Delta$ beforehand.

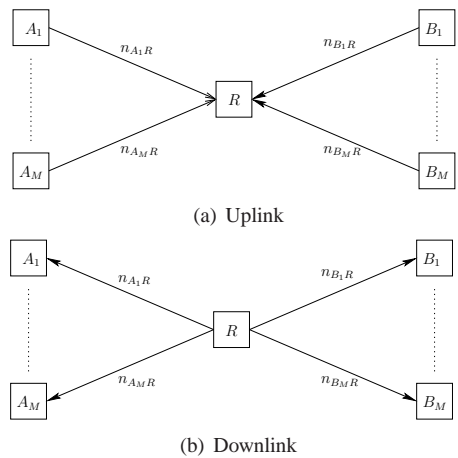
\begin{figure}
     \centering
     \subfigure[Uplink]{
  \scalebox{0.5}{   \input{fig1.pstex_t}}
}
\subfigure[Downlink]{
       \scalebox{0.5}{\input{fig2.pstex_t}}
}
     \caption{The system model for $M$ pair bidirectional linear shift deterministic relay network. \label{fig:sysModel}}
\end{figure}

We use the linear shift deterministic channel model to model the interaction between the transmitted signals. The linear shift deterministic channel model was introduced in~\cite{AvestimehrIT}. Here is a formal definition of this channel model.

\begin{defn} \textbf{(Definition of the linear shift deterministic model)}
Consider a wireless network as a set of nodes $V$, where $|V|=N$.
Communication from node $i$ to node $j$ has a non-negative integer
gain\footnote{Some channels may have zero gain.} $n_{(i,j)}$ associated
with it. This number models the channel gain in a corresponding Gaussian setting.
At each time $t$, node $i$ transmits a vector ${\mathbf{x}_i}[t] \in
\mathbb{F}_{2}^q$ and receives a vector ${\mathbf{y}_i}[t] \in \mathbb{F}_{2}^q$ where
$q=\max_{i,j}(n_{(i,j)})$. The received signal at each node is
a deterministic function of the transmitted signals at the other
nodes, with the following input-output relation: if the nodes in the
network transmit ${\mathbf{x}_1}[t], {\mathbf{x}_2}[t] , \ldots {\mathbf{x}_N}[t]$ then the received
signal at node j, $1 \leq j \leq N$ is:
\begin{equation}
\label{eq:channel_model}
{\mathbf{y}_{j}}[t]=\sum_{k=1}^N
 {\mathbf{S}^{q-n_{k,j}}}{\mathbf{x}_{k}}[t]
\end{equation}
for all $1 \leq k \leq N$,
where $\mathbf{S}$ is the $q \times q$ shift matrix and the summation and
multiplication is in $\mathbb{F}_{2}$.
\end{defn}

Now that we have defined the linear shift deterministic channel model we can apply it to the multi-pair bidirectional relay network. A pictorial representation of an example of such network with two pairs is shown in Figure \ref{fig:detExample}. In this figure each little circle represents a signal level and what is sent on it is a bit. The transmit and received signal levels are sorted from MSB to LSB from top to bottom. The channel gain between two nodes $i$ and $j$ indicates how many of the first MSB transmitted signal levels of node $i$ are received at destination node $j$. As described in the channel model (\ref{eq:channel_model}), at each received signal level, the receiver gets only the modulo two summation of the incoming bits.

\begin{figure}
     \centering
     \subfigure[Uplink]{
  \scalebox{0.33}{   \input{fig3.pstex_t}}
}
\subfigure[Downlink]{
       \scalebox{0.33}{\input{fig4.pstex_t}}
}
     \caption{The pictorial representation of a two-pair bidirectional linear shift deterministic relay network with channel gains $n_{A_1R}=3$, $n_{B_1R}=2$, $n_{A_2R}=2$, $n_{B_2R}=1$, $n_{RA_1}=2$, $n_{RB_1}=3$, $n_{RA_2}=1$ and $n_{RB_2}=2$. \label{fig:detExample}}
\end{figure}
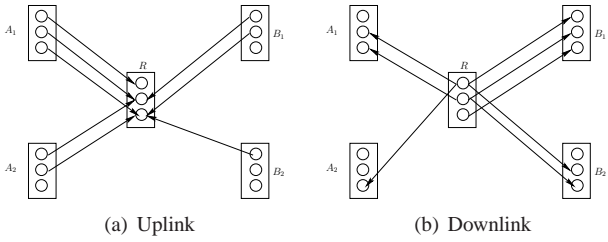

{\subsection{Cut-set upper-bound and a motivating example\label{sec:cutSet}}
The cut-set upper-bound \cite{CoverThomas} on the capacity region of the full-duplex $M$-pair bidirectional linear shift deterministic relay network (described in Section \ref{sec:sysModel}) is given by
\begin{align}
\label{eq:cutSetDet}
\sum_{i \in \mathcal{U}} [ \ell_i R_{A_i} + &   (1-\ell_i) R_{B_i} ]  \\
 \leq \min \Big( & \max_{i \in \mathcal{U} } ( \ell_i n_{A_iR} + (1-\ell_i) n_{B_iR} ), \nonumber \\
 & \max_{i \in \mathcal{U} } ( \ell_i n_{RB_i} + (1-\ell_i) n_{RA_i} ) \Big), \nonumber
\end{align}
for all $\mathcal{U}\subseteq \{1, \dots, M \}$ and $\ell_i \in \{0,1\}$, $i=1,\ldots,M$.
This bound is simply obtained by considering the pairs $(A_i,B_i)$, $i \in \mathcal{U}$, and creating a cut between them such that, if $\ell_i=1$, $A_i$ is on the left and $B_i$ is on the right side of the cut, and if $\ell_i=0$, $B_i$ is on the left and $A_i$ is on the right side of the cut. We then consider the sum-rate of communication from the nodes on the left side of the cut to the nodes on the right side of the cut. This is upper bounded by (\ref{eq:cutSetDet}), where the first term on the RHS of (\ref{eq:cutSetDet}) is the maximum number of bits that the relay can receive from the nodes on the left side of the cut, and  the second term on the RHS of (\ref{eq:cutSetDet}) is the maximum number of bits that the relay can broadcast to the nodes on the right side of the cut.
}

For example, in the case that we have only two pairs {($M=2$)} and the relay is operating on the full-duplex mode, the cut-set upper-bound on the capacity region is given by
\begin{align}
R_{A_1} &\leq \min \lp n_{A_1R}, n_{RB_1} \rp  \label{eq: cut-set bound 1} \\
R_{B_1} &\leq \min \lp n_{B_1R}, n_{RA_1} \rp  \label{eq: cut-set bound 2}\\
R_{A_2} &\leq \min \lp n_{A_2R}, n_{RB_2} \rp  \label{eq: cut-set bound 3}\\
R_{B_2} &\leq \min \lp n_{B_2R}, n_{RA_2} \rp  \label{eq: cut-set bound 4}\\
R_{A_1}+R_{A_2} & \leq  \min \lp \max \lp n_{A_1R},n_{A_2R} \rp, \max \lp n_{RB_1},n_{RB_2} \rp\rp  \label{eq: cut-set bound 5}\\
R_{B_1}+R_{B_2} & \leq  \min \lp \max \lp n_{B_1R},n_{B_2R} \rp, \max \lp n_{RA_1},n_{RA_2} \rp\rp  \label{eq: cut-set bound 6}\\
R_{A_1}+R_{B_2} & \leq  \min \lp \max \lp n_{A_1R},n_{B_2R} \rp, \max \lp n_{RB_1},n_{RA_2} \rp\rp  \label{eq: cut-set bound 7}\\
R_{B_1}+R_{A_2} & \leq  \min \lp \max \lp n_{B_1R},n_{A_2R} \rp, \max \lp n_{RA_1},n_{RB_2} \rp\rp  \label{eq: cut-set bound 8}.
\end{align}

As a motivating example, we now consider the network shown in Figure \ref{fig:detExample}. It is easy to check that the rate tuple
\begin{align}
(R_{A_1},R_{B_1},R_{A_2},R_{B_2})=(2,1,1,1)\nonumber
 \end{align}
 is inside its cut-set region. In Figure \ref{fig:detExampleAchi} we illustrate a simple scheme that achieves this rate point. With this strategy, the nodes in the uplink transmit
\begin{align*}
&x_{A_1}=\left[ a_{1,1},a_{1,2},0 \right]^t, \quad x_{B_1}=\left[ b_{1,1},0,0 \right]^t \\
&x_{A_2}=\left[ 0,a_{2,1},0 \right]^t ,\quad x_{B_1}=\left[ b_{2,1},0,0 \right]^t
\end{align*}
and the relay receives
\beq y_R=[  a_{1,1} ,~a_{1,2}\oplus b_{1,1},~a_{2,1} \oplus b_{2,1}]^t. \nonumber \eeq
Then the relay will re-order the received signal and transmit
\beq x_R=[a_{2,1} \oplus b_{2,1} ,~a_{1,2}\oplus b_{1,1},~a_{1,1} ]^t.\nonumber \eeq
Then node $A_1$ receives the {signals (i.e., XOR-combination)} $a_{1,2}\oplus b_{1,1}$ and since it knows $a_{1,2}$ can decode $b_{1,1}$. Similarly node $B_1$ can decode $a_{1,1}$ and $a_{1,2}$, node $A_2$ can decode $b_{2,1}$ and finally node $B_2$ can decode $a_{2,1}$. Therefore we achieve the rate point $(2,1,1,1)$.
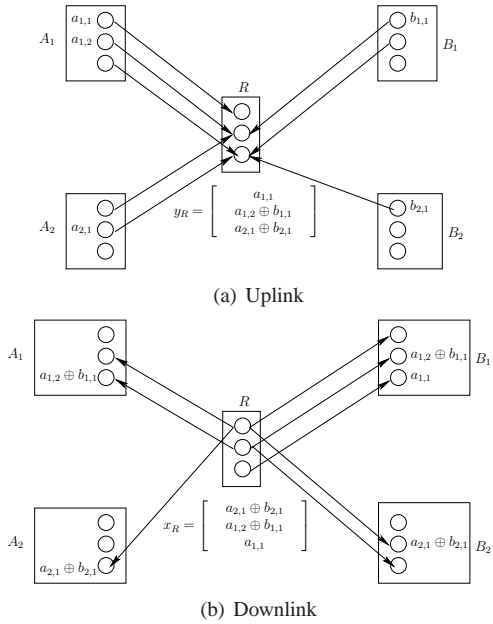
\begin{figure}
     \centering
     \subfigure[Uplink]{
  \scalebox{0.45}{   \input{fig5.pstex_t}}
}
\hspace{0.2in}
\subfigure[Downlink]{
       \scalebox{0.45}{\input{fig6.pstex_t}}
}
     \caption{The scheme that achieves rate point $(2,1,1,1)$. \label{fig:detExampleAchi}}
\end{figure}

There are some interesting points about this particular achievability strategy:
\begin{itemize}
\item There is no coding over time.
\item There is no interference between different pairs on the same received signal level at the relay.
\item The relay just re-orders the received XOR-combinations and forwards them.
\end{itemize}
{We call a strategy with these properties a {\emph{divide-and-conquer} relaying strategy}, which will be defined more formally in the next section. Quite interestingly, we next prove that any rate point in the cut-set bound region of the bidirectional linear shift deterministic relay network  can be achieved using such strategy. }

\subsection{Capacity region}

In this section we study the capacity region of the multi-pair bidirectional linear shift deterministic relay network. {We first give an overview of our achievability strategy. This strategy consists of three components, namely orthogonalization, reordering (or permutation) and forwarding. The first component (orthogonalization) divides the uplink signal levels at the relay between the pairs in such a way that no signal-level is assigned to more than one-pair. Hence, different pairs are orthogonalized in the uplink and do not interact with each other anymore. As a consequence, at each level, the relay either receives one bit from a single node or the XOR-combination of two bits coming from the pair of nodes that wish to communicate with each other. The relay then reorders its received signal by applying a permutation matrix $\mathbf{\Pi}$ (i.e., $\mathbf{x}_R=\Pi \mathbf{y}_R$) and forwards it in the downlink.} We name this strategy the {\emph{divide-and-conquer} relaying strategy}.

We now state our main result in this section.
\begin{theorem}
\label{thm:main}
The capacity region of the full-duplex multi-pair bidirectional linear shift deterministic relay network, described in Section \ref{sec:sysModel}, is equal to the cut-set upper-bound~(\ref{eq:cutSetDet}), and it is achieved by {the divide-and-conquer relaying strategy described above}.
\end{theorem}
\begin{proof}

We first prove the result for integral\footnote{i.e., with integer components.} rate-tuples. We use induction on the sum-rate
 \begin{align}
 R_{\text{sum}}=\sum_{i=1}^M (R_{A_i}+R_{B_i}) \nonumber
  \end{align}
  to show that every integral $2M$-tuple $(R_{A_1},R_{B_1},\cdots, R_{A_M},R_{B_M})$ satisfying the cut set bound is achievable by allocating subsets of the signal levels exclusively to users of different sessions\footnote{A session means the communication of one pair}, and using functional-forwarding at the relay.

The proof is obvious for $R_{\text{sum}}=1$. Assume it is true for all channel gains and all integral rate-tuples with sum-rate $R_{\text{sum}}\leq k$. We now prove this for $R_{\text{sum}}=k+1$. Consider a $2M$-tuple $\textbf{R}=(R_{A_1},R_{B_1},\cdots, R_{A_M},R_{B_M})$ satisfying the cut set bound (\ref{eq:cutSetDet}) and $R_{\text{sum}}=k+1$. We consider two separate cases.

\emph{Case $1$}: There is a pair where both nodes have nonzero transmission rates. Without loss of generality we may assume that $R_{A_1}$ and $R_{B_1}$ are both nonzero. Our goal is to choose one up-link signal level and one down-link signal level at the relay, and assign them to the ($A_1$,$B_1$) session.  $A_1$ and $B_1$ will then transmit one bit at the specified uplink level to the relay, and the relay will transmit (broadcast) the received XOR-combination at the specified down-link level to both $A_1$ and $B_1$.  After doing so and removing the specified signal levels, the network will reduce  to a network with lower channel gains. We then show that the reduced rate-tuple $(R_{A_1}-1,R_{B_1}-1,R_{A_2},R_{B_2},\cdots, R_{A_M},R_{B_M})$ is in the cut-set region of the reduced network. Therefore, by induction, it will be achieved and the proof will be complete.

More specifically, for the up-link we choose the highest signal level connected to both $A_1$ and $B_1$ (denoted by $l_u=\min(n_{A_1R},n_{B_1R}$), and for the down-link, we chose the lowest signal level connected to both $A_1$ and $B_1$ (denoted by $l_d=\min(n_{RA_1},n_{RB_1}$). After removing signal levels $l_u$ and $l_d$ from the up-link and down-link of the relay, we obtain a linear shift deterministic network with channel gains
\begin{eqnarray}
\label{eq:reduced1}n'_{A_iR}&=& n_{A_iR}- \mathbf{1}(n_{A_iR} \geq l_u), \quad i=1,\ldots,M,\\
\label{eq:reduced2} n'_{B_iR}&=& n_{B_iR}- \mathbf{1}(n_{B_iR} \geq l_u), \quad i=1,\ldots,M,\\
\label{eq:reduced3} n'_{RA_i}&=& n_{RA_i}- \mathbf{1}(n_{RA_i} \geq l_d), \quad i=1,\ldots,M,\\
\label{eq:reduced4} n'_{RB_i}&=& n_{RB_i}- \mathbf{1}(n_{RB_i} \geq l_d), \quad i=1,\ldots,M,
\end{eqnarray}
where $\mathbf{1}(.)$ is the indicator function.

As we show in Appendix~\ref{app:proofDetCase1}, the reduced rate-tuple $\mathbf{R}'=(R_{A_1}-1,R_{B_1}-1,R_{A_2},R_{B_2},\cdots, R_{A_M},R_{B_M})$ is in the cut-set region of this network, i.e.,
\begin{align}
\label{eq:cutSetProof}
\sum_{i \in \mathcal{U}} [ \ell_i R'_{A_i} + &  (1-\ell_i) R'_{B_i} ] \\
\leq \min \Big( & \max_{i \in \mathcal{U} } ( \ell_i n'_{A_iR} + (1-\ell_i) n'_{B_iR} ), \nonumber\\
&  \max_{i \in \mathcal{U} } ( \ell_i n'_{RB_i} + (1-\ell_i) n'_{RA_i} ) \Big),\nonumber
\end{align}
for all $\mathcal{U}\subseteq \{1, \dots, M \}$ and $\ell_i \in \{0,1\}$, $i=1,\ldots,M$.

Moreover, $\mathbf{R}'$, has sum-rate $R'_{\text{sum}}=R_{\text{sum}}-2=k-1$. Hence, by our induction assumption, it can be achieved by using the remaining levels and the proof in this case is complete.

\emph{Case 2}: Every session has a node with zero rate. Without loss of generality, assume that $R_{B_1}=\cdots=R_{B_M}=0$ and $R_{A_1} \geq 1$. Again, we choose one up-link signal level and one down-link signal level at the relay, and assign them to the ($A_1$,$B_1$) session.  $A_1$ will then transmit one bit at the specified uplink level to the relay, and the relay will transmit (broadcast) the received bit at the specified down-link level to $B_1$.  After doing so and removing the specified signal levels, the network will reduce  to a network with lower channel gains. We then show that the reduced rate-tuple $(R_{A_1}-1,0,R_{A_2},0,\cdots, R_{A_M},0)$ is in the cut-set region of the reduced network. Therefore, by induction, it will be achieved and the proof  will be complete.

More specifically, we choose the highest signal level in the up-link that is connected $A_1$ (denoted by $l_u=n_{A_1R}$), and for the down-link, we chose the lowest signal level connected to $B_1$ (denoted by $l_d=n_{RB_1}$). After removing signal levels $l_u$ and $l_d$ from the up-link and down-link of the relay, we obtain a linear shift deterministic network with channel gains in (\ref{eq:reduced1})-(\ref{eq:reduced4}).

As we show in Appendix \ref{app:proofDetCase2}, the reduced rate-tuple $\mathbf{R}'=(R_{A_1}-1,0,R_{A_2},0,\cdots, R_{A_M},0)$ is in the cut-set region of the reduced network. Moreover, it has sum-rate $R'_{\text{sum}}=R_{\text{sum}}-1=k$. Hence, by our induction assumption, it can be achieved and the proof in this case is complete.

To complete the proof, we just need to show that all corner points of the cut-set bound region are achieved by the divide-and-conquer relaying strategy. Note that since all coefficients of the hyperplanes of the cut-set bound region are integers, then all corner points of the region must be fractional. If a corner point $\overrightarrow{R}$ is integral then we are done. Otherwise, we choose a large enough integer $Q$ such that $Q\overrightarrow{R}$ is integral. Now note that $Q$ instances of a {linear} shift deterministic network over time is the same as the original network with all channel gains are multiplied by $Q$.
To see this, let the transmit and received signal of node $k$ ($k \in V$) at $Q$ time instances to be $\mathbf{x}_k[i]=\left [x_k^{(1)}[i],x_k^{(2)}[i],\ldots,x_k^{(q)}[i] \right ]^T \in \mathbb{F}_2^q$ and $\mathbf{y}_k[i]=\left [y_k^{(1)}[i],y_k^{(2)}[i],\ldots,y_k^{(q)}[i] \right ]^T\in \mathbb{F}_2^q$, $i=0,\ldots,Q-1$, satisfying (\ref{eq:channel_model}). We now define $\tilde{\mathbf{x}}_k$ and $\tilde{\mathbf{y}}_k$ as given in~\eqref{eq:DefXYtilde} on the top of the next page.
\begin{table*}
\begin{align}
\tilde{\mathbf{x}}_k&=\left [x_k^{(1)}[0],\ldots,x_k^{(1)}[Q-1], x_k^{(2)}[0],\ldots,x_k^{(2)}[Q-1],\ldots,x_k^{(q)}[0],\ldots,x_k^{(q)}[Q-1] \right ]^T \in \mathbb{F}_2^{Qq} \nonumber\\
\tilde{\mathbf{y}}_k&=\left [y_k^{(1)}[0],\ldots,y_k^{(1)}[Q-1], y_k^{(2)}[0],\ldots,y_k^{(2)}[Q-1],\ldots,y_k^{(q)}[0],\ldots,y_k^{(q)}[Q-1] \right ]^T \in \mathbb{F}_2^{Qq}. \label{eq:DefXYtilde}
\end{align}
\hrule
\end{table*}
From (\ref{eq:channel_model}), it is easy to see that $\tilde{\mathbf{x}}_k$'s and $\tilde{\mathbf{y}}_k$'s satisfy
\beq \tilde{\mathbf{y}}_j=\sum_{k=1}^N
 {\mathbf{\tilde{S}}^{Qq-Qn_{k,j}}}{\mathbf{\tilde{x}}_{k}}, \nonumber \eeq
 where ${\mathbf{\tilde{S}}}$ is now the $Qq \times Qq$ shift matrix. Hence, we equivalently have a linear shift deterministic network with all channel gains multiplied by $Q$.
 Now since $Q\overrightarrow{R}$ is integral and is obviously inside the cut-set upper-bound of the enhanced network (where all channel gains are multiplied by $Q$), then it is achievable by the divide-and-conquer relaying strategy. This strategy can then be simply  translated to a divide-and-conquer relaying strategy on the original network over $Q$ time-steps. Therefore the corner point $\frac{QR}{Q}=R$ is achievable.
\end{proof}

For illustration, let's apply the inductive algorithm in the proof of Theorem \ref{thm:main} to achieve the rate-tuple ($3$,$1$,$2$,$2$) in the example network as shown in Fig.~\ref{fig:inducExm1} and Fig.~\ref{fig:inducExm1F}. We should first take the ($A_1$,$B_1$) pair and serve them through one signal level in UL and one level in DL, reducing the remaining rate-tuple to ($2$,$0$,$2$,$2$). This step is shown in Fig.~\ref{fig:inducExm1a}. Next, we take the ($A_2$,$B_2$) and similarly assign corresponding levels in UL and DL to them. This is done twice, reducing the remaining rate-tuple to ($2$,$0$,$0$,$0$). These two steps are shown in Figures~\ref{fig:inducExm1b}  and \ref{fig:inducExm1c}. For the sake of clarity, the removed signal levels are dotted in each step.  The remaining unserved rates are ($2$,$0$,$0$,$0$). We then apply the procedure in the case 2 of the inductive algorithm in Theorem~\ref{thm:main}. Fig.~\ref{fig:inducExm1d} shows how this idea is applied to our example network. The final configuration that achieves the rate-tuple for this example is shown in Fig.~\ref{fig:inducExm1F}.

{In the case that the relay is operating on the half-duplex mode (i.e., listening $\Delta$ fraction of the time and transmitting the rest), the cut-set upper-bound \cite{CoverThomas} on the capacity region of $M$-pair bidirectional linear shift deterministic relay network will be
\begin{align}
\nonumber \sum_{i \in \mathcal{U}} [ \ell_i R_{A_i} + & (1-\ell_i) R_{B_i} ]  \\ \label{eq:cutSetDetHD}
 \leq \min \Big( & \Delta \max_{i \in \mathcal{U} } ( \ell_i n_{A_iR} + (1-\ell_i) n_{B_iR} ), \\
& (1-\Delta) \max_{i \in \mathcal{U} } ( \ell_i n_{RB_i} + (1-\ell_i) n_{RA_i} ) \Big), \nonumber
\end{align}
for all $\mathcal{U}\subseteq \{1, \dots, M \}$ and $\ell_i \in \{0,1\}$, $i=1,\ldots,M$.}

{As a corollary of Theorem \ref{thm:main}, we can also show that in this case the cut-set upper-bound is achievable.
\begin{coro} Let $\Delta$ be the fraction of the time the relay listens and transmits the rest. The capacity region of the multi-pair bidirectional linear shift deterministic relay network with a half-duplex relay is equal to the cut-set upper-bound~(\ref{eq:cutSetDetHD}), and it is achieved by the divide-and-conquer relaying strategy.
\end{coro}
\begin{proof}
Without loss of generality assume $\Delta$ is a fractional number (otherwise consider the sequence of fractional numbers approaching it). Then choose a large enough integer $Q$ such that $Q\Delta$ is integer. Then consider $Q$ instances of the network over time, such that for $Q\Delta$ instances the relay is listening and in the other $Q(1-\Delta)$ instances it is transmitting. After concatenating these instances together, the resulting network can be thought of as a full-duplex multi-pair network where the uplink channel gains are multiplied by $Q\Delta$ and the downlink channel gains are multiplied by $(1-\Delta)Q$. It is easy to verify that the cut-set bound region of this network is just the cut-set bound region of the original half-duplex network expanded by $Q$. Now by {Theorem \ref{thm:main}} and the previous argument, we know that the capacity region of this full-duplex multi-pair bidirectional network is equal to its cut-set upper-bound and is achieved by the divide-and-conquer relaying strategy. Now note that any divide-and-conquer relaying strategy in this full-duplex network can be translated to a divide-and-conquer relaying strategy in $Q$ instances of the original half-duplex network; $Q\Delta$ instances the relay is in the listen mode to get the signals and $(1-\Delta)Q$ instances in the transmit mode to forward the signals. Therefore the cut-set upper-bound is achievable and the proof is complete.
\end{proof}}
\begin{figure*}
     \centering
     \subfigure[Step 1]{\label{fig:inducExm1a}
     \includegraphics[scale=0.85]{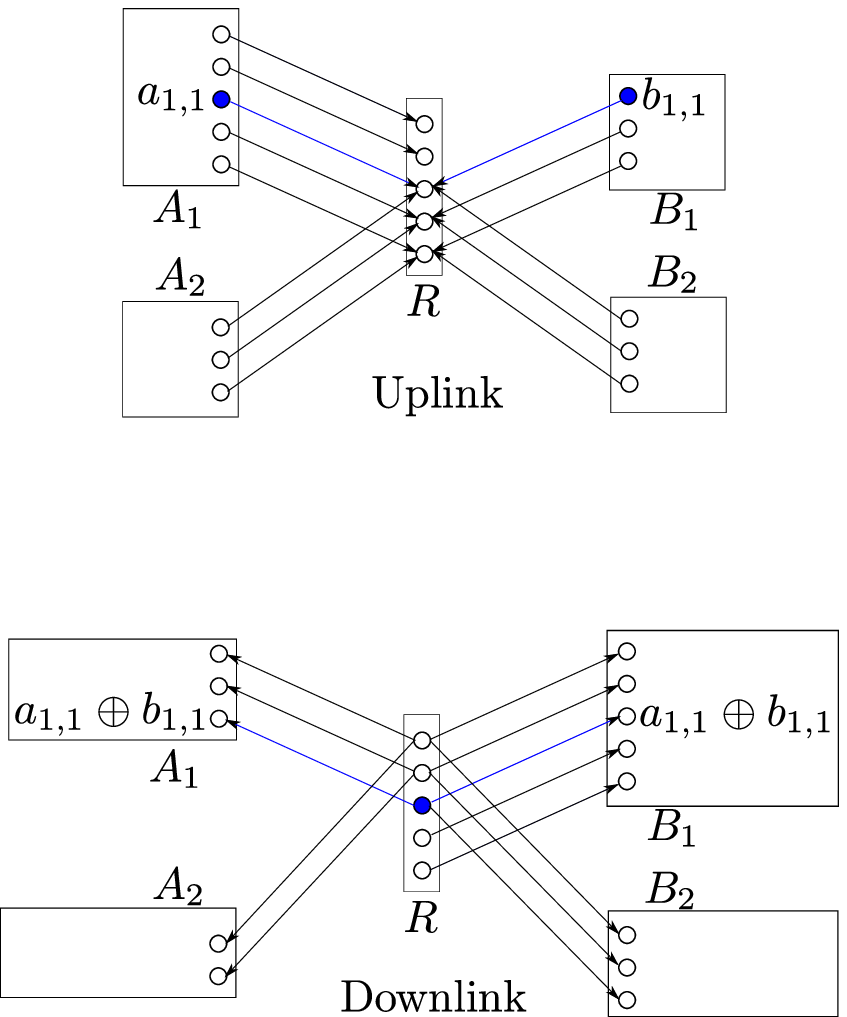}
     }
\subfigure[Step 2]{\label{fig:inducExm1b}
      \includegraphics[scale=0.85]{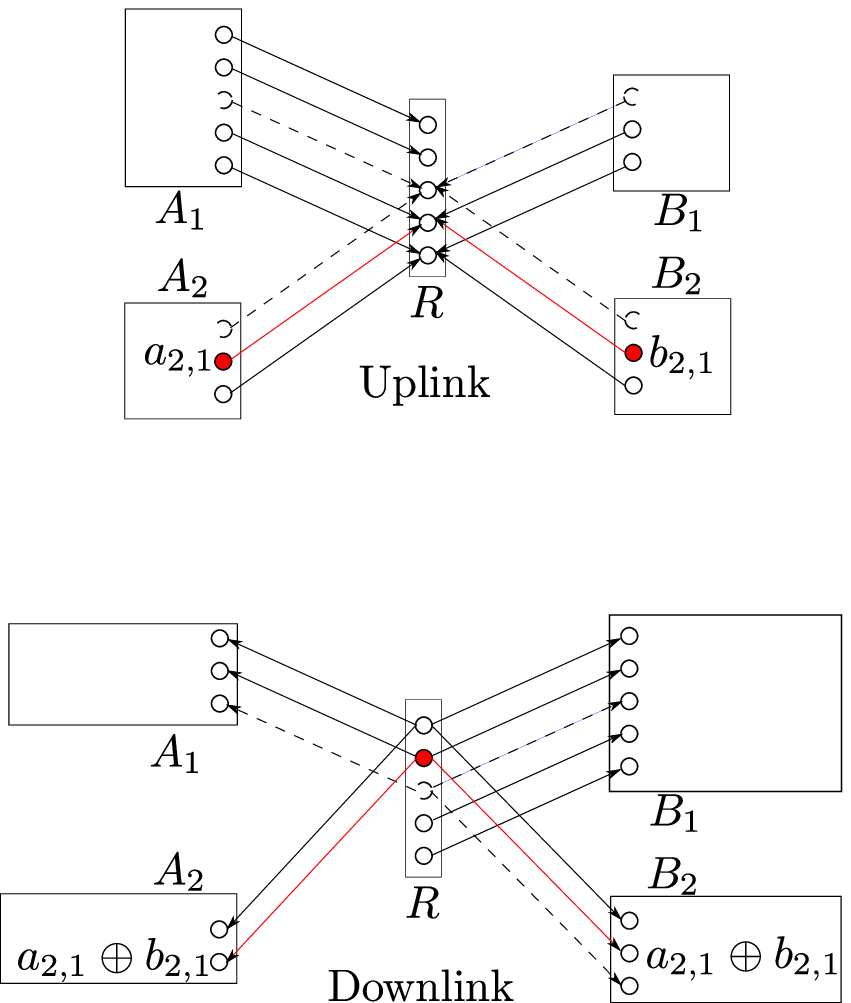}
}
\subfigure[Step 3]{\label{fig:inducExm1c}
       \includegraphics[scale=0.85]{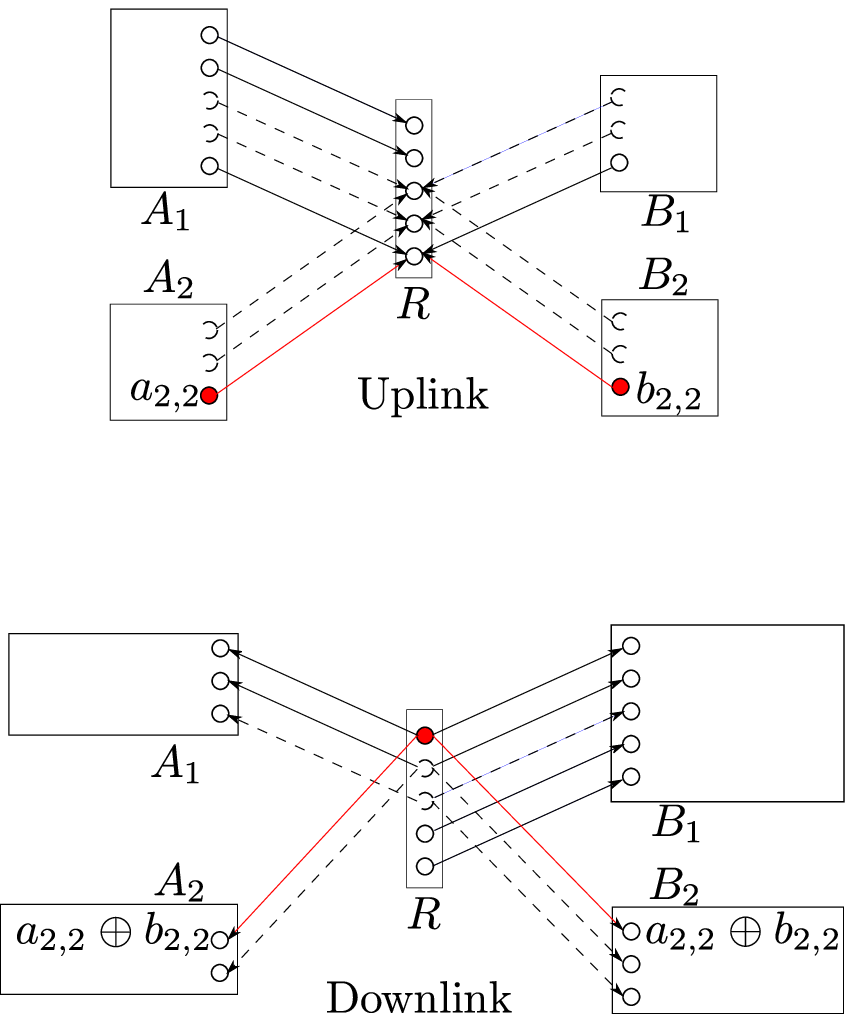}
}
\subfigure[Step 4]{\label{fig:inducExm1d}
       \includegraphics[scale=0.85]{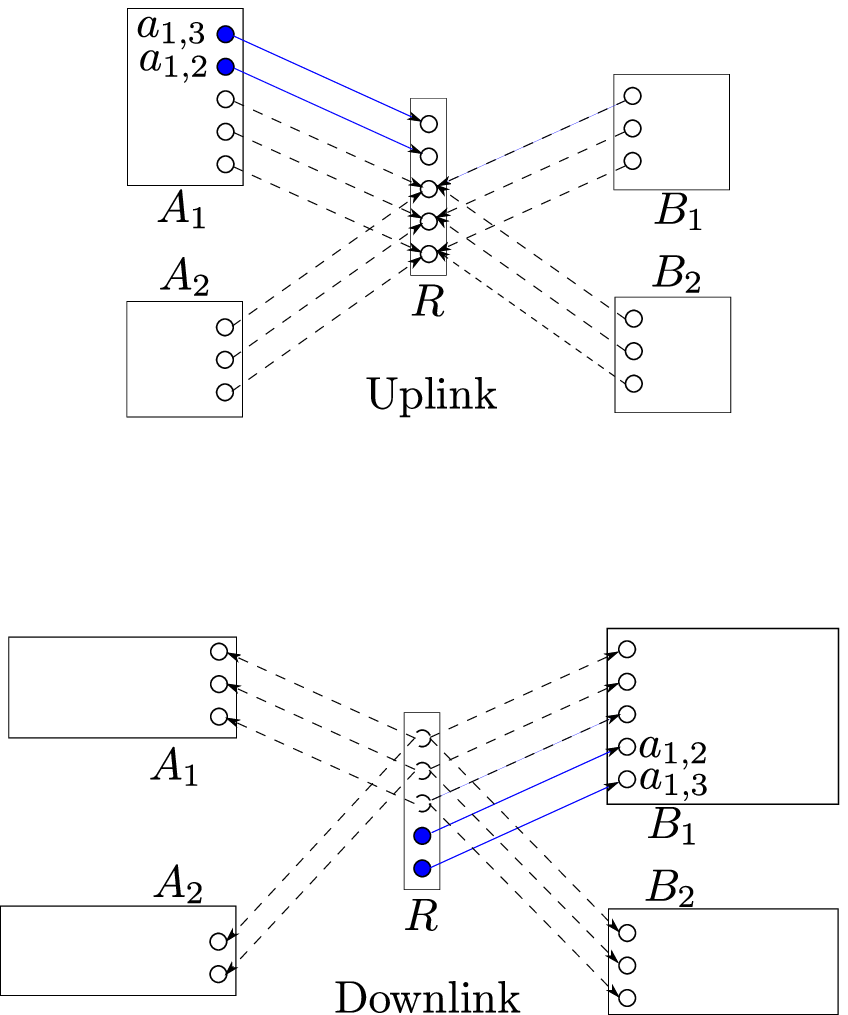}
}
     \caption{Illustration of the inductive algorithm introduced in Theorem \ref{thm:main}. \label{fig:inducExm1}}
\end{figure*}

\begin{figure}
     \centering
     \subfigure[Uplink]{
  \includegraphics[scale=0.3]{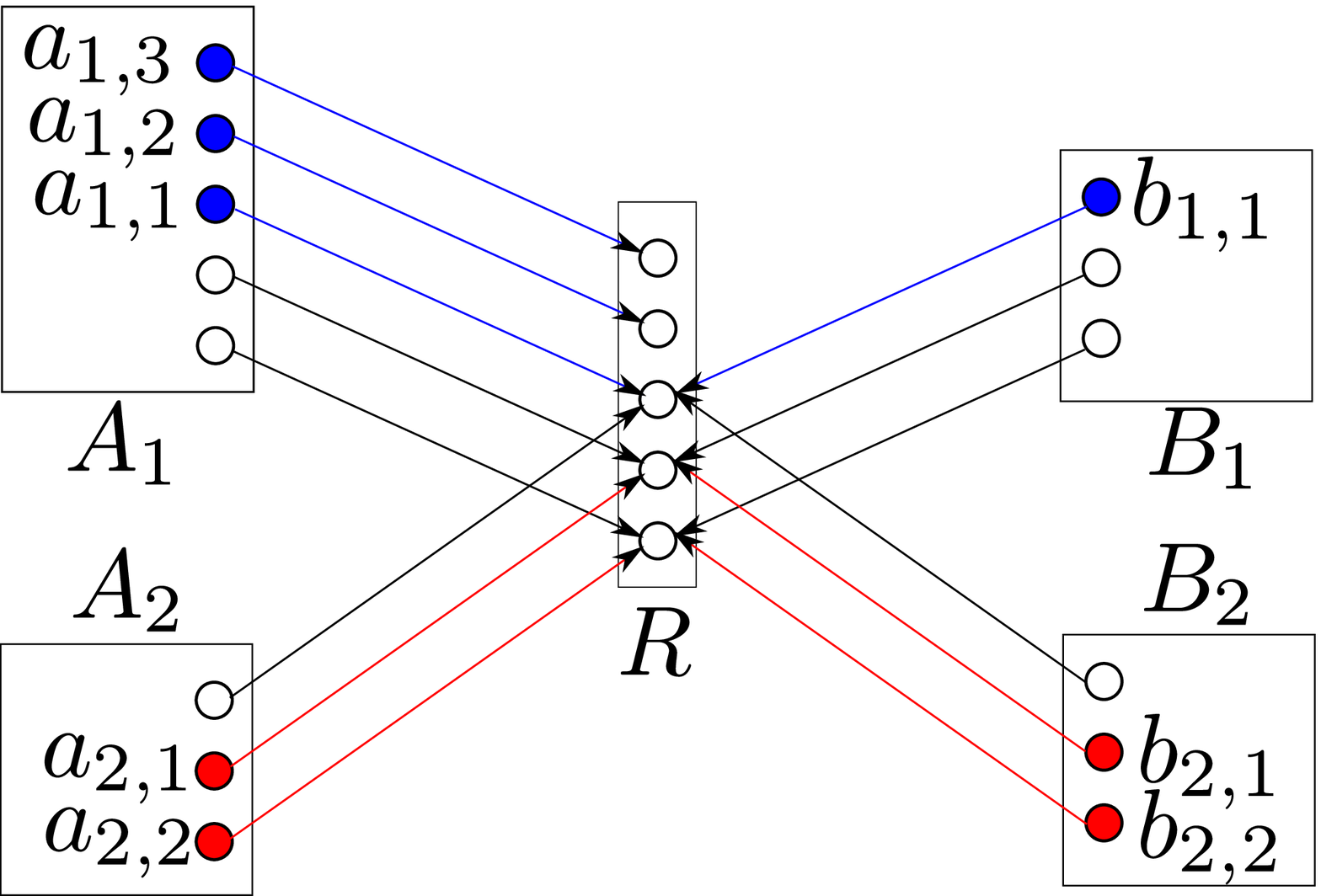}
}
\hfill
\subfigure[Downlink]{
       \includegraphics[scale=0.3]{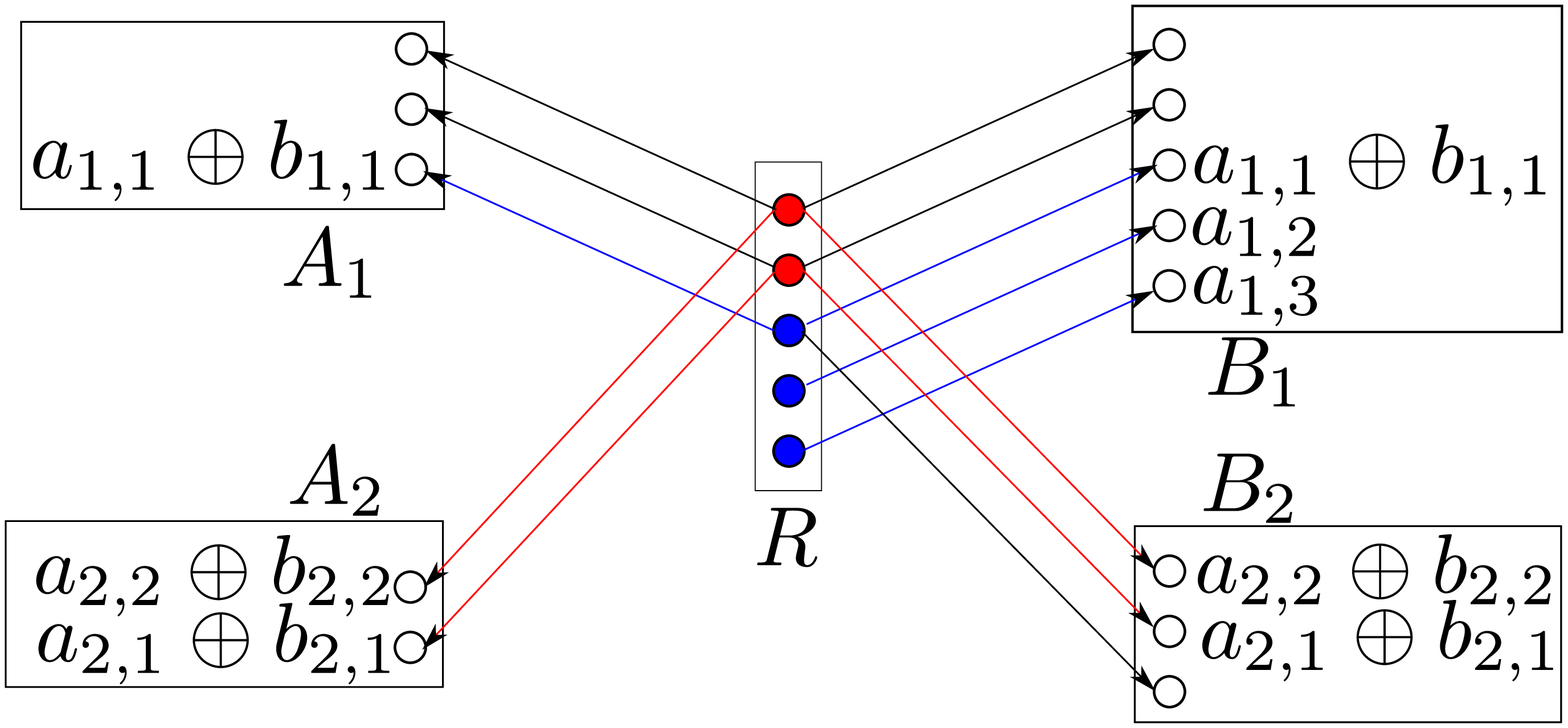}
}
     \caption{Illustration of the resulting divide-and-conquer strategy of the inductive algorithm.  \label{fig:inducExm1F}}
\end{figure}

\subsection{Remark}
\label{sec:Remark}
An interesting insight, which will prove to be useful in the transition to the Gaussian case is the following.
Although the scheme we provided  is an inductive way of level assignment and seems quite unstructured (in the sense that it assigns signal levels on a greedy  basis), one can actually say more about these assignments using certain observations. First of all, note that in this divide-and-conquer relaying strategy we have in general $2M$ types of signals that the relay might decode. Namely, $M$ types of signals that are made up of one bit from one user of a session, and $M$ types of signals that are the {XOR-combination} of bits from both users of the same pair. Each signal is received at the relay at some signal level, and is transmitted to one or  both of the end users at potentially another signal level in down-link.  Please refer to the example network of Figures \ref{fig:inducExm1} and \ref{fig:inducExm1F}, and observe that quite interestingly, in the final configuration of signal type-level assignments, all signals of the same type are concatenated together both in UL and DL. In other words they appear at concatenated signal levels. In general, one can serve all signals of the same type at once by choosing a pair with nonzero rates and serve them (one bit per user per signal level) until one of the rates is zero. For $M=2$ for example, assuming $R_{A_1}\geq R_{B_1}$ and $R_{A_2}\geq R_{B_2}$, instead of reducing ($R_{A_1}$,$R_{B_1}$,$R_{A_2}$,$R_{B_2}$) to ($R_{A_1}-1$,$R_{B_1}-1$,$R_{A_2}$,$R_{B_2}$) one can reduce it to ($R_{A_1}-R_{B_1}$,0,$R_{A_2}$,$R_{B_2}$) all at once and find a chunk of signal levels to afford them. Then the same thing can be done for the other pair. In the final configuration, all signals of the same type are in concatenation, which is also illustrated in Fig.~\ref{fig:OrthogSignLevelSpace}.
\begin{figure}
\begin{center}
  \includegraphics[scale=0.5]{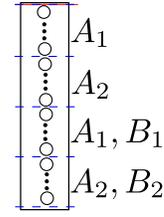}\\
  \caption{Allocating chunks of relay levels to signals of the same type with $R_{A_i}\geq R_{B_i}$.}\label{fig:OrthogSignLevelSpace}
  \end{center}
\end{figure}

In the following, we discuss more insights gained from the examination of the linear shift deterministic multi-pair bidirectional relay network that can be interpreted for the two-pair Gaussian relay network.

\section{Transition from the linear shift deterministic model to Gaussian model}~\label{sec:Insights}
The result of the deterministic network basically suggests that it is optimal to divide the signal-level space into subspaces and allocate these orthogonal subspaces to the different sessions, i.e., pairs. Furthermore, it suggests to split the message of the stronger user of each pair (the user with stronger uplink channel, cf. section~\ref{sec:Remark}) into two parts:
\begin{enumerate}
\item the first part has the same rate as the rate as the message from the weak user and it is transmitted such that at the relay it is received with the same power as that of the signal from the weak user,
\item the second part has the remaining rate and is transmitted at some higher signal levels.
\end{enumerate}

{Hence, for $M=2$, the relay receives four chunks of bits at different signal levels. Namely, the bits that are created from the XOR-combination of the signals of both users of each pair and, bits from the signals of the strong transmitter of each pair. The relay then forwards these signals at non-overlapping signal levels to the end users so that the XOR-combination of the signals is received by both users, whereas the other bits (from the strong transmitters) are received by the corresponding end users only. This way each user can easily decode its message having the received XOR-combinations, received bits and its own transmitted message.}

{To apply a similar strategy to Gaussian networks, one will face three immediate challenges. The first one is the effect of the additive noise which is inevitably  present in the Gaussian channels. The second issue is that the received signals at the relay can not be fully orthogonalized (i.e., we face interference between low power and high power signals). The third complication is in decoding the superposition of signals (and not the individual signals) which should take place at the relay.}

{We propose the following solutions to overcome these difficulties. The noise issue can be simply resolved by using an appropriate block symbol coding scheme. The orthogonalization problem is inevitable, however a compensation in the capacity region allows for interference tolerance. In other words, rather than showing the cut-set upper-bound is tight, we show that the cut-set upper-bound is achievable to within a constant. Finally, using an appropriate lattice code, the third challenge is resolvable, too. In a lattice structure, the superposition of every two codewords is also a lattice codeword and  therefore can be decoded at the relay~\cite{BaikChung,NaryananLatticeBiDi}. These will be addressed in the sections that follow.}

\section{Two-Pair Bidirectional Gaussian Relay Network}~\label{sec:2Pair2WayGaussRelay}
In this section we analyze the capacity region of the two-pair bidirectional Gaussian relay network shown in Figure~\ref{fig:System}.
In particular, we show that the transmission scheme which was developed in the previous section achieves within $3$ bits/sec/Hz per user of the cut-set upper-bound on the capacity region.

Thus, we consider two single-antenna transceiver pairs, $(A_1,B_1)$ and $(A_2,B_2)$,  communicating to each other by exploiting a relay $R$. The relay is operating in the full-duplex mode, i.e., it can listen and transmit at the same time. We use a complex AWGN channel model for all channels in this network.
\begin{figure}
\begin{center}
    \subfigure[Uplink]{\label{fig:FirstHop}
    \psfrag{A}{\hspace{-0.1cm}{\footnotesize{$A_1$}}}
\psfrag{R}{$\text{R}$}
\psfrag{h1}{\footnotesize{$h_{A_1R}$}}
\psfrag{h2}{\footnotesize{$h_{B_1R}$}}
\psfrag{h3}{\footnotesize{$h_{A_2R}$}}
\psfrag{h4}{\footnotesize{$h_{B_2R}$}}
\psfrag{B}{\hspace{-0.15cm}{\footnotesize{$B_1$}}}
\psfrag{C}{\hspace{-0.15cm}{\footnotesize{$A_2$}}}
\psfrag{D}{\hspace{-0.15cm}{\footnotesize{$B_2$}}}
    \includegraphics[width=4cm]{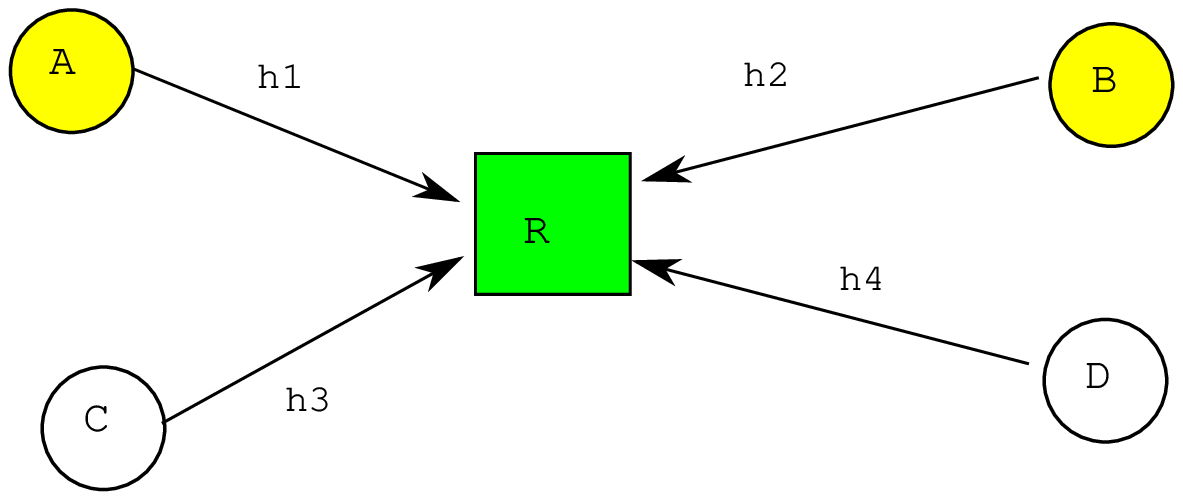}}
    \subfigure[Downlink]{\label{fig:SecHop}
    \psfrag{A}{\hspace{-0.1cm}{\footnotesize{$A_1$}}}
\psfrag{R}{$\text{R}$}
\psfrag{B}{\hspace{-0.15cm}{\footnotesize{$B_1$}}}
\psfrag{C}{\hspace{-0.15cm}{\footnotesize{$A_2$}}}
\psfrag{D}{\hspace{-0.15cm}{\footnotesize{$B_2$}}}
\psfrag{h4}{\footnotesize{$h_{RA_1}$}}
\psfrag{h5}{\footnotesize{$h_{RB_1}$}}
\psfrag{h6}{\footnotesize{$h_{RA_2}$}}
\psfrag{h7}{\footnotesize{$h_{RB_2}$}}
    \includegraphics[width=4cm]{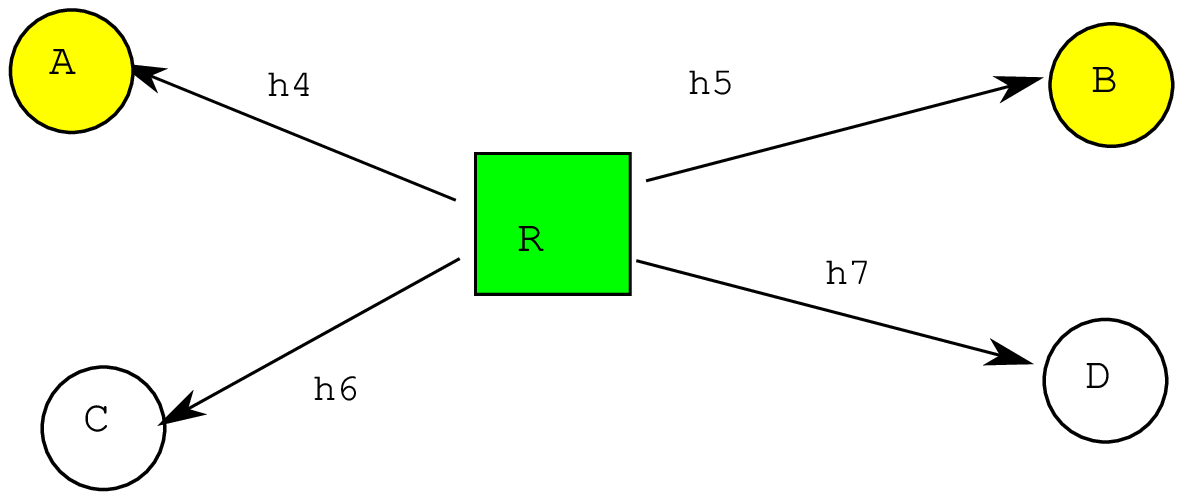}}
\end{center}
    \caption{Two-Pair bidirectional full-duplex relay network}
    \label{fig:System}
\end{figure}
Hence, the received signals at the nodes are given by
\begin{align}
y_R &=h_{A_1R}x_{A_1}+h_{B_1R}x_{B_1}+h_{A_2R}x_{A_2}+h_{B_2R}x_{B_2}+z_R,\nonumber\\
y_{A_i}& =h_{RA_i}x_R+z_{A_i}\;, \quad
y_{B_i} =h_{RB_i}x_R+z_{B_i}, \qquad i=1,2\nonumber
\end{align}
where $x_{A_1}$, $x_{B_1}$, $x_{A_2}$, $x_{B_2}$, and $x_R$ are the signals transmitted from nodes $A_1$, $B_1$, $A_2$, $B_2$, and $R$, respectively. The transmit power constraint is {$\mathbb{E}\left[|x_{A_i}|^2\right]=\mathbb{E}\left[|x_{B_i}|^2\right]=\mathbb{E}\left[|x_{R}|^2\right]\leq P$} and the noises $z_{A_1}$, $z_{B_1}$, $z_{A_2}$, $z_{B_2}$, and $z_R$ are all distributed as $\mathcal{CN}(0,1)$. Note that the uplink channels gains ($h_{A_iR}$ and $h_{B_iR}$) are not necessarily equal to the down-link channel gains ($h_{RA_i}$ and $h_{RB_i}$), i.e., channel reciprocity is not assumed. For each pair ($A_i$,$B_i$), $R_{A_i}$ is the rate at which $A_i$ transmits data to $B_i$ and $R_{B_i}$ is the transmission rate of $B_i$ to $A_i$.

We now begin by describing the cut-set upper-bound~\cite{CoverThomas},  denoted by $\bar{\mathcal{C}}_{\mathsf{cs}}$, on the capacity region of this network:
\begin{align}
\overline{\mathcal{C}}_{\mathsf{cs}}  = \Big\{ (R_{A_1},R_{B_1},& R_{A_2},R_{B_2}) \in \mathbb{R}_+^4 : \label{eq:CutSetRegion} \\
R_{A_i} \leq  \min & \lp  C\lp |h_{A_iR}|^2P\rp, C\lp |h_{RB_i}|^2P\rp \rp  \label{eq:GenGaussCut-set bound 1} \\
R_{B_i} \leq  \min & \lp  C\lp |h_{B_iR}|^2P\rp , C\lp|h_{RA_i}|^2P\rp \rp  \label{eq:GenGaussCut-set bound 2}\\
R_{A_1}+R_{A_2}  \leq  \min \Big( & C\lp \lp|h_{A_1R}|+|h_{A_2R}|\rp^2P\rp , \nonumber \\
&
C \lp \lp |h_{RB_1}|^2+|h_{RB_2}|^2\rp P \rp \Big)  \label{eq:GenGaussCut-set bound 5}  \\
R_{B_1}+R_{B_2}  \leq  \min \Big(& C\lp \lp|h_{B_1R}|+|h_{B_2R}|\rp^2P\rp , \nonumber \\
&
C \lp \lp |h_{RA_1}|^2+|h_{RA_2}|^2\rp P\rp \Big)  \label{eq:GenGaussCut-set bound 6} \\
R_{A_1}+R_{B_2}  \leq  \min \Big(& C\lp \lp|h_{A_1R}|+|h_{B_2R}|\rp^2P\rp , \nonumber \\
&
C \lp \lp |h_{RB_1}|^2+|h_{RA_2}|^2\rp P\rp   \Big)  \label{eq:GenGaussCut-set bound 7}
\end{align}
\begin{align}
R_{B_1}+R_{A_2}  \leq  \min \Big(& C\lp \lp|h_{B_1R}|+|h_{A_2R}|\rp^2P\rp , \nonumber\\
&
 C \lp \lp |h_{RA_1}|^2+|h_{RB_2}|^2\rp P\rp \Big)  \label{eq:GenGaussCut-set bound 8}\Big\},
\end{align}
where $C(x)=\log\left(1+x\right)$.
The terms in~\eqref{eq:GaussCut-set bound 1}-~\eqref{eq:GaussCut-set bound 8} correspond to the cuts labeled from $1$ to $8$ in Fig.~\ref{fig:Cuts}.

We also define a ``restricted cut-set bound'', denoted by $\bar{\mathcal{C}}$, to be:
\begin{align}
\overline{\mathcal{C}}  = \Big\{ (R_{A_1},R_{B_1}, & R_{A_2},R_{B_2}) \in \mathbb{R}_+^4 : \label{eq:RestCutSetRegion}  \\
R_{A_i} \leq  \min & \lp C\lp |h_{A_iR}|^2P\rp, C\lp |h_{RB_i}|^2P\rp \rp  \label{eq:GaussCut-set bound 1} \\
R_{B_i} \leq  \min & \lp C\lp |h_{B_iR}|^2P\rp , C\lp|h_{RA_i}|^2P\rp \rp  \label{eq:GaussCut-set bound 2}\\
R_{A_1}+R_{A_2}  \leq  \min \Big( & C\lp \lp|h_{A_1R}|^2+|h_{A_2R}|^2\rp P\rp , \nonumber\\
&
C \lp \max \lp |h_{RB_1}|^2,|h_{RB_2}|^2\rp P \rp \Big)  \label{eq:GaussCut-set bound 5}\\
R_{B_1}+R_{B_2}  \leq  \min \Big(& C\lp \lp|h_{B_1R}|^2+|h_{B_2R}|^2\rp P\rp , \nonumber \\
&
C \lp \max \lp |h_{RA_1}|^2,|h_{RA_2}|^2\rp P\rp \Big)  \label{eq:GaussCut-set bound 6}\\
R_{A_1}+R_{B_2}  \leq  \min \Big(& C\lp \lp|h_{A_1R}|^2+|h_{B_2R}|^2\rp P\rp , \nonumber\\
&
C \lp \max \lp |h_{RB_1}|^2,|h_{RA_2}|^2\rp P\rp   \Big)  \label{eq:GaussCut-set bound 7}\\
R_{B_1}+R_{A_2}  \leq  \min \Big( & C\lp \lp|h_{B_1R}|^2+|h_{A_2R}|^2\rp P\rp , \label{eq:GaussCut-set bound 8}\\
&
 C \lp \max \lp |h_{RA_1}|^2,|h_{RB_2}|^2\rp P\rp \Big)  \nonumber \Big\},
\end{align}

In the next lemma, we show that the gap between the cut-set bound and the restricted cut-set bound is at-most 1 bit/sec/Hz per user.
\begin{lemma}\label{lem:Extrabit}
The cut-set upper bound~ in~\eqref{eq:CutSetRegion} is within 1 bit/sec/Hz per user of the restricted cut-set upper bound in~\eqref{eq:RestCutSetRegion}.
\end{lemma}
\begin{proof}
Consider the first expressions in~\eqref{eq:GenGaussCut-set bound 5} and ~\eqref{eq:GaussCut-set bound 5}.
It holds that
\begin{align}
 C & \lp \lp|h_{A_1R}|+|h_{A_2R}|\rp^2P\rp \nonumber \\
& \leq   C\lp \lp|h_{A_1R}|^2+2|h_{A_1R}||h_{A_2R}|+|h_{A_2R}|^2\rp P \rp \nonumber \\
&   \stackrel{(a)}{\leq}  C\lp \lp 2|h_{A_1R}|^2+2|h_{A_2R}|^2\rp P \rp \nonumber \\
 & \leq  C\lp \lp|h_{A_1R}|^2+|h_{A_2R}|^2\rp P \rp +1, \nonumber
\end{align}
where $(a)$ follows since $\lp |h_{A_1R}|-|h_{A_2R}| \rp^2 \geq 0$. Thus, the gap between the first expressions in~\eqref{eq:GenGaussCut-set bound 5} and ~\eqref{eq:GaussCut-set bound 5} is at most $1$ bit/sec/Hz.
Similarly, for the second expressions in~\eqref{eq:GenGaussCut-set bound 5} and ~\eqref{eq:GaussCut-set bound 5}, it holds that
\begin{align}
C & \lp \lp |h_{RB_1}|^2+|h_{RB_2}|^2\rp P \rp \nonumber \\
& \leq  C \lp 2 \max \lp |h_{RB_1}|^2,|h_{RB_2}|^2\rp P \rp  \nonumber \\
& \leq C \lp \max \lp |h_{RB_1}|^2,|h_{RB_2}|^2\rp P \rp +1 \nonumber
\end{align}
and thus the gap between he second expressions in~\eqref{eq:GenGaussCut-set bound 5} and~\eqref{eq:GaussCut-set bound 5} is at most $1$ bit/sec/Hz.
Following the same procedure for the remaining sum rate terms in~\eqref{eq:CutSetRegion} and~\eqref{eq:RestCutSetRegion} completes the proof.
\end{proof}
In the remainder of the paper, we only consider the restricted cut-set upper bound. This is motivated as follows. The structure of the expressions in  in~\eqref{eq:RestCutSetRegion} resemble the rate expressions of the achievable scheme which is described in the following. Thus, the gap analysis becomes very convenient and by Lemma~\ref{lem:Extrabit} we are assured that we loose at most one additional bit/sec/Hz in the gap analysis to go from the restricted cut-set bound to the actual cut-set bound.

Next, we define the up-link and down-link cut-set regions. The up-link cut-set region, $\mathcal{C}_u$, is the set of rates satisfying equations (\ref{eq:GaussCut-set bound 1})-(\ref{eq:GaussCut-set bound 8}) when the down-link channel gains are assumed infinity. This means that the only restricting factors in determining the capacity regions are assumed to be the up-link channel gains. Likewise, the down-link cut-set region, $\mathcal{C}_d$,  is the set of rates satisfying (\ref{eq:GaussCut-set bound 1})-(\ref{eq:GaussCut-set bound 8}) in which the up-link channel gains are set to infinity. Note that
$\bar{\mathcal{C}}=\mathcal{C}_d \cap \mathcal{C}_u$.

We say that a 4-tuple $(R_{A_1},R_{B_1},R_{A_2},R_{B_2})$ is achievable if simultaneously $A_i$ can communicate to $B_i$ at rate $R_{A_i}$ and $B_i$ can communicate to $A_i$ at rate $R_{B_i}$ with arbitrary small error probability. The union of all achievable rate tuples is defined as the capacity region. We are now ready to state our main result.

\begin{theorem}\label{theo:Main}
The capacity region of the two pair full-duplex bidirectional relay network is within 2 bits/sec/Hz per user of its restricted cut-set upper-bound described in~\eqref{eq:GaussCut-set bound 1}-\eqref{eq:GaussCut-set bound 8}. Or, more precisely, if
\beq (R_{A_1},R_{B_1},R_{A_2},R_{B_2}) \in \overline{\mathcal{C}} \nonumber \eeq
and $R_{A_i}, R_{B_i} \geq 2$ for $i=1,2$, then the rate tuple $(R_{A_1}-2,R_{B_1}-2,R_{A_2}-2,R_{B_2}-2)$ is achievable.
\end{theorem}
The rest of this section is devoted to proving this Theorem. First, we state the following lemma which helps us by limiting the number of rate configurations that we have to consider.

\begin{lemma}~\label{lem:CaseReduction}
Let $\mathbf{R} = (R_{A_1} , R_{B_1} , R_{A_2} , R_{B_2} )$ be a rate tuple in  the cut-set region $\overline{\mathcal{C}}$. Assume $R_{A_i} \geq R_{B_i}$, $i=1,2$. Then it is always possible to sufficiently reduce the transmit powers at the uplink and add extra noise to the received signals at the downlink, such that  new effective channel gains satisfy $|\tilde{h}_{A_iR}| \geq |\tilde{h}_{B_iR}|$ and $|\tilde{h}_{RB_i}| \geq |\tilde{h}_{RA_i}|$ for $i=1,2$, and $\mathbf{R}$ is still in the shrunk cut-set region.
\end{lemma}
\begin{proof}
See Appendix~\ref{app:lemCaseReduction}.
\end{proof}
 This lemma basically reduces the number of relevant channel gain orderings that we have to consider in order to prove Theorem~\ref{theo:Main}. Assume that the rate tuple that we want to show to be achievable (within 2 bits/sec/Hz per user) satisfies $R_{A_i} \geq R_{B_i}$ for $i=1,2$. By Lemma~\ref{lem:CaseReduction}, we can without loss of generality (wlog) assume that $|h_{A_iR}| \geq |h_{B_iR}|$ for $i=1,2$. We can also wlog assume that $|h_{A_1R}| \geq |h_{A_2R}|$ (otherwise we can re-label pair 1 and pair 2). Therefore, we only need to consider three different channel gain orderings for the uplink. Those three cases are shown in Fig.~\ref{fig:CaseI},~\ref{fig:CaseII} and~\ref{fig:CaseIII}. Similarly, we only need to consider three cases for the downlink.
 To prove Theorem \ref{theo:Main}, first we describe the encoding strategy at the transmission nodes. As mentioned earlier, the idea is that strong transmitters of each pair split their signals into a Gaussian codeword and a lattice codeword, while the weak user only transmits a lattice codeword. While stating this encoding strategy we leave the power allocation parameters unspecified. In other words, the power level at which the user breaks up its message into the superposition of Gaussian and a lattice codeword remains as parameters. In the next step we mention the decoding at the relay where the superposition of lattice points and the Gaussian codewords are decoded. Afterwards, the relay maps each of the four decoded codewords into a random Gaussian codeword, and broadcasts their weighted superposition to all users. The last step is the decoding at the nodes, where every receiver first decodes the undesired codewords that have larger weights than the desired codewords.
Thus, those codewords are decoded and successively canceled from the received signal one by one. Afterwards,  both the weak and the strong receivers of each pair decode the Gaussian codeword corresponding to the lattice codeword belonging to that pair. In addition to that, the strong receivers decode one more codeword. This codeword corresponds to the Gaussian codeword, which was received by the relay from their transmitting strong counterpart.  Eventually as a result of this scheme the rates that the users will successfully transmit will be a function of the power parameters that we set at the beginning. We will finally show that by choosing these parameters appropriately any rate tuple within $2$ bits/sec/Hz per user of the cut set is achievable.

\begin{figure}
\begin{center}
\subfigure[Cuts (I)]{
\label{fig:Cuts16}
\psfrag{A}{\hspace{-0.1cm}{\footnotesize{$A_1$}}}
\psfrag{R}{$\text{R}$}
\psfrag{B}{\hspace{-0.15cm}{\footnotesize{$B_1$}}}
\psfrag{C}{\hspace{-0.15cm}{\footnotesize{$A_2$}}}
\psfrag{D}{\hspace{-0.15cm}{\footnotesize{$B_2$}}}
\psfrag{1}{\hspace{-0.05cm}{\footnotesize{$1$}}}
\psfrag{2}{\hspace{-0.05cm}{\footnotesize{$2$}}}
\psfrag{3}{\hspace{-0.05cm}{\footnotesize{$5$}}}
\psfrag{4}{\hspace{-0.05cm}{\footnotesize{$6$}}}
\psfrag{5}{\hspace{-0.05cm}{\footnotesize{$3$}}}
\psfrag{6}{\hspace{-0.05cm}{\footnotesize{$4$}}}
    \includegraphics[scale=0.35]{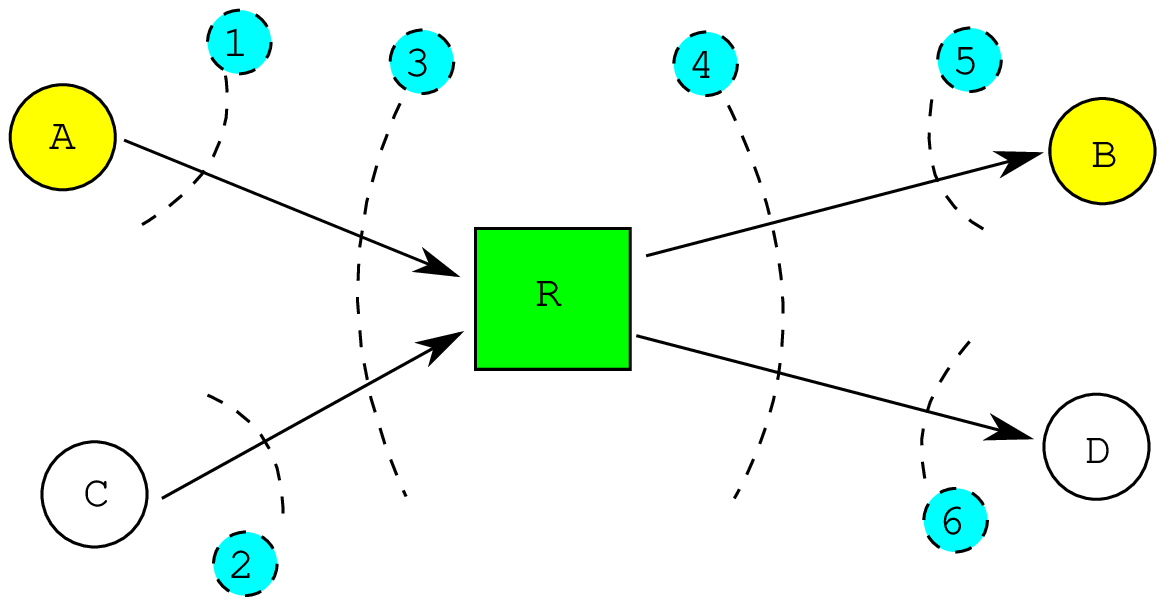}}
\subfigure[Cuts (II)]{
\psfrag{A}{\hspace{-0.1cm}{\footnotesize{$A_1$}}}
\psfrag{R}{$\text{R}$}
\psfrag{B}{\hspace{-0.15cm}{\footnotesize{$B_1$}}}
\psfrag{C}{\hspace{-0.15cm}{\footnotesize{$A_2$}}}
\psfrag{D}{\hspace{-0.15cm}{\footnotesize{$B_2$}}}
\psfrag{3}{\hspace{-0.05cm}{\footnotesize{$7$}}}
\psfrag{4}{\hspace{-0.05cm}{\footnotesize{$8$}}}
\label{fig:Cuts78}
    \includegraphics[scale=0.35]{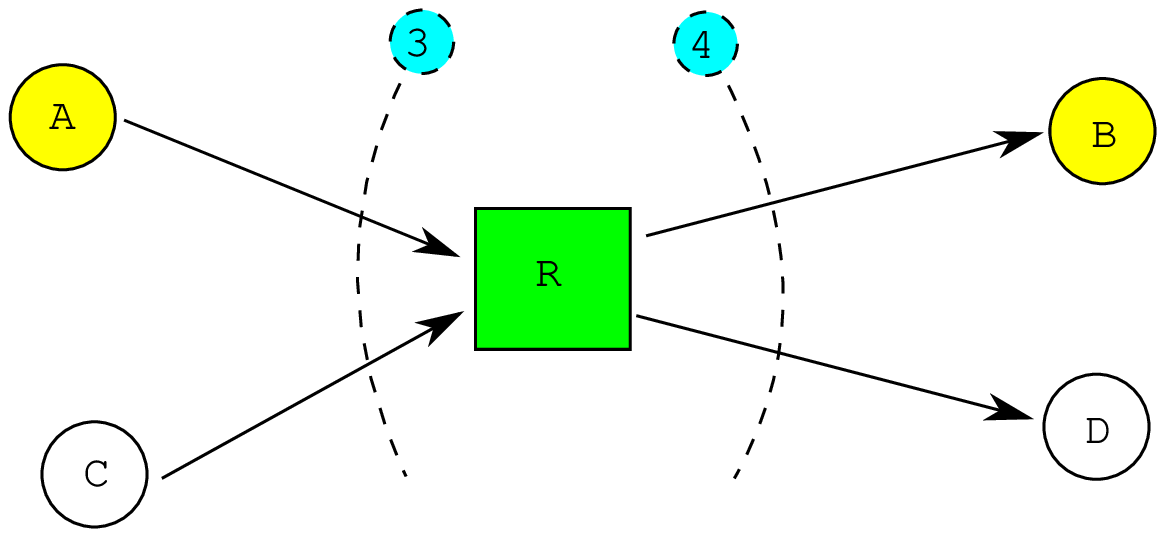}}
\end{center}
    \caption{Cuts for the upper-bound on the capacity region}
    \label{fig:Cuts}
\end{figure}
\begin{figure}
\begin{center}
\subfigure[Case I]{
\label{fig:CaseI}
\psfrag{R1}{{\tiny{$|h_{A_1R}|$}}}
\psfrag{R2}{{\tiny{$|h_{B_1R}|$}}}
\psfrag{R3}{{\tiny{$|h_{A_2R}|$}}}
\psfrag{R4}{{\tiny{$|h_{B_2R}|$}}}
\psfrag{Noise level}{{\tiny{$\text{Noise level}$}}}
\includegraphics[scale=0.75]{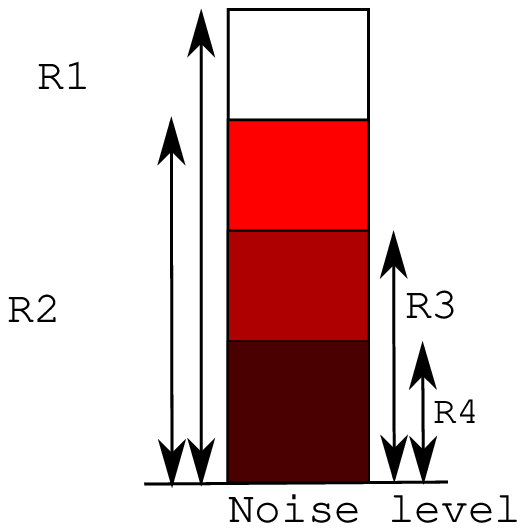}}
 \hspace*{0.2cm}   \subfigure[Case II]{
\label{fig:CaseII}
\psfrag{R1}{{\tiny{$|h_{A_1R}|$}}}
\psfrag{R2}{{\tiny{$|h_{B_1R}|$}}}
\psfrag{R3}{{\tiny{$|h_{A_2R}|$}}}
\psfrag{R4}{{\tiny{$|h_{B_2R}|$}}}
\psfrag{Noise level}{{\tiny{$\text{Noise level}$}}}
    \includegraphics[scale=0.75]{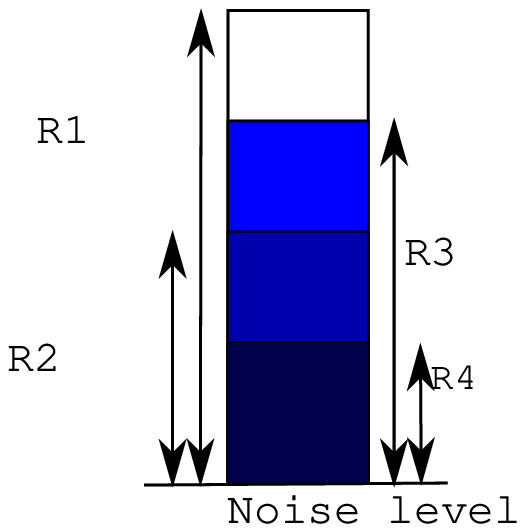}}
 \hspace*{0.2cm}    \subfigure[Case III ]{
\label{fig:CaseIII}
\psfrag{R1}{{\tiny{$|h_{A_1R}|$}}}
\psfrag{R2}{{\tiny{$|h_{B_1R}|$}}}
\psfrag{R3}{{\tiny{$|h_{A_2R}|$}}}
\psfrag{R4}{{\tiny{$|h_{B_2R}|$}}}
\psfrag{Noise level}{{\tiny{$\text{Noise level}$}}}
    \includegraphics[scale=0.75]{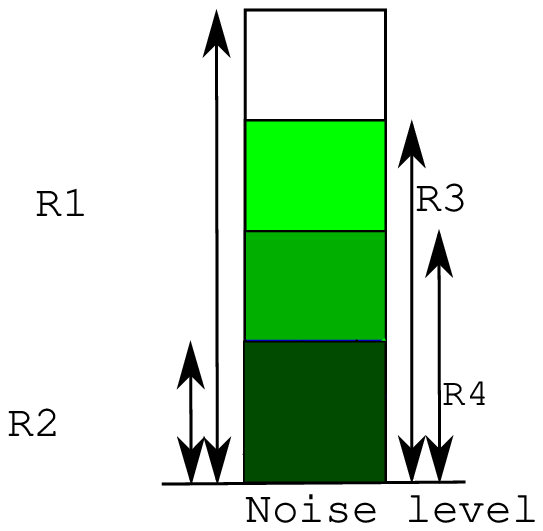}}
    \caption{Three relevant configurations for the uplink and their corresponding received signal at the relay. At the lowest level, all signals are superposed, while at the next level (medium shade), all but one signals are superposed. At the top level (white) only one signal remains.}
\label{fig:AllCases}
\end{center}
\end{figure}

\subsection{Lattice Coding}
In the following, some preliminaries and results on lattice coding are provided that we use in the remainder of the paper. We refer the interested reader to~\cite{Loelinger} for more details.

A lattice $\Lambda$ of dimension $n$ is described by
\begin{align}
    \mathbf{\Lambda}=\left\{ \mathbf{\lambda} =\mathbf{G}\mathbf{x}: \mathbf{x} \in \mathds{Z}^n\right\}, \nonumber
\end{align}
where $\mathbf{G}$ describes the lattice and is referred to as the generator matrix.
The fundamental Voronoi region of such a lattice $\mathbf{\Lambda}$ is denoted by $\Omega$.
Furthermore, the volume of $\Omega$, i.e., the reciprocical of the number of lattice point per unit volume, is denoted by $V$.
Now, let $p$ a positive integer  and $\mathds{Z}_p$ the set of integers modulo $p$.
Further, let $\bar{v}:\mathds{Z}^n \rightarrow \mathds{Z}_p^n$ be the componentwise modulo operation over integer vectors.
The lattices used in this paper are mod-$p$ lattices, i.e., of the form
\begin{align}
    \Lambda_c=\left\{v \in \mathds{Z}^n: \bar{v} \in C\right\},\nonumber
\end{align}
where $C$ be a linear $(n,k)$ code  over $\mathds{Z}_p$ and $p$ is prime\cite[Construction A]{Loelinger}.
Now, let $\mathcal{B}$ be a balanced set~\cite{Loelinger,SriramJafarShamai} of linear $(n,k)$ codes over $\mathds{Z}_p$ and let $\mathcal{L}_\mathcal{B}$ be the set of lattices denoted by
\begin{align}
    \mathcal{L}_\mathcal{B}=\left\{\Lambda_c: C \in \mathcal{B}\right\}.\nonumber
\end{align}
With this in mind, lets consider the following system model
\begin{align}
    y=x+z,\nonumber
\end{align}
where $y$ is the receive signal, $x$ is the transmit signal and $z$ is additive noise with zero mean and a variance $\sigma^2$.
It was shown in~\cite[Theorem 4]{Loelinger} that if the transmitted codeword is a lattice point, then there exists a lattice for that channel and the average probability of error with lattice decoding can be made arbitrarily small as the dimension of the lattice increases.
Similarly, it was shown in~\cite{Loelinger} that by using a codebook $(\mathbf{\Lambda} + \mathbf{s}) \cap \mathbf{S}$, where $\mathbf{s}$ is a shift and $\mathbf{S}$ describes the shaping gain, a rate $R$ with arbitrarily small probability of error can be achieved if
\begin{align}
    R\leq \log\left(\frac{P}{\sigma^2}\right).\nonumber
\end{align}
We will use this result in the remainder of the paper for the characterization of the rate region achievable with our proposed scheme.

\subsection{Encoding at the nodes}\label{subsec:encUL}
Wlog assume that $R_{A_iR}\geq R_{B_iR}$. By Lemma~\ref{lem:CaseReduction} this means that we can assume $|h_{A_iR}|\geq |h_{B_i R}|$ and $|h_{RB_i}|\geq |h_{RA_i}|$. Then, the transmit signals at the nodes are given by
\begin{align}
{x}_{A_i} &=\sqrt{\alpha_{A_i}^{(1)}}{x}_{A_i}^{(1)}+\sqrt{\alpha_{A_i}^{(2)}}{x}_{A_i}^{(2)}\;, \;
{x}_{B_i} =\sqrt{\alpha_{B_i}^{(2)}}{x}_{B_i}^{(2)} \quad i=1,2\nonumber\\
{x}_R &=\sum_{j=1}^{4}\sqrt{\alpha_R^{(j)}}{x}_R^{(j)} \label{eq:RelTransmitTwoPair}\qquad \text{with } \sum_{j}^{4} \alpha_R^{(j)}=1,
\end{align}
where ${x}_{A_i}^{(1)}$ and ${x}_R^{(j)}$ are codewords chosen from a random Gaussian codebook of size $2^{nR_{A_i}^{(1)}}$, $i=1,2$, and $2^{nR_R^{(j)}}$, for $j=1,\dots,4$,  respectively. $\mathbf{x}_{A_i}^{(2)}$ and  $\mathbf{x}_{B_i}^{(2)}$, $i=1,2$, are lattice coded~\cite{Loelinger} using lattice ensembles $\{\Lambda_{A_1^{(2)}},\Lambda_{A_1^{(2)}},\Lambda_{B_1^{(2)}},\Lambda_{B_2^{(2)}}\}$ giving a codebook of size $2^{nR_{A_i}^{(2)}}$ and $2^{nR_{B_i}^{(2)}}$ with $i=1,2$, respectively. We assume that the second moment per dimension of the fundamental Voronoi region~\cite{Loelinger} of each lattice is $\nicefrac{1}{2}$ which ensures satisfying the power constraint.
At nodes $A_i$ we have two messages ${m}_{A_i}^{(1)}$ and ${m}_{A_i}^{(2)}$ from dictionaries of size $2^{nR_{A_i}^{(1)}}$ and $2^{nR_{A_i}^{(2)}}$ that are mapped to $x_{A_i}^{(1)}$ and $x_{A_i}^{(2)}$, respectively.
In other words, the strong transmitter of each pair transmits a superposition of a lattice code and a random Gaussian code, while the weaker user only transmits a lattice code. Thus, the transmit signals of nodes $B_1$ and $B_2$ reduce to
\begin{align}
x_{B_1}& =\sqrt{\alpha_{B_1}^{(2)}}x_{B_1}^{(2)} \nonumber \\
x_{B_2} &= \sqrt{\alpha_{B_2}^{(2)}}x_{B_2}^{(2)} \nonumber.
\end{align}
For the nodes $A_1$ and $A_2$, we have a superposition code (cf.~\eqref{eq:RelTransmitTwoPair}). Note that
\begin{align}
t=x_{A_1}^{(2)}+x_{B_1}^{(2)}\; \text{ and }
f=x_{A_2}^{(2)}+x_{B_2}^{(2)} \nonumber,
\end{align}
where $t$ and $f$ are also lattice points due to the group structure of the lattice~\cite{NaryananLatticeBiDi}.

The power parameters (i.e., $\alpha_{A_i}$ and $\alpha_{B_i}$) are assigned such that the lattice codes of each pair arrive at the same power level, so that the relay can decode the sum codeword correctly. Thus we set,
\begin{align}\label{eq:HowtoChooseAlpha}
\alpha_{A_i}^{(2)}=\frac{|h_{B_iR}|^2}{|h_{A_iR}|^2}\alpha_{B_i}^{(2)}.
\end{align}
Furthermore, we should have $\alpha_{A_i}^{(1)}+\alpha_{A_i}^{(2)}\leq 1$ and $\alpha_{B_i}^{(2)}\leq 1$.

\subsection{Uplink: Decoding at the relay}
Recall that as discussed in Section~\ref{sec:Insights} and illustrated in Figure \ref{fig:AllCases} we have to analyze three cases only. Here, the analysis for the first case (cf. Fig.~\ref{fig:CaseI}) is given in detail. For the other cases, only the results are presented, since the other cases are similar and therefore omitted. However, along the presentation of the results, we also mention the differences should there be any.

\subsubsection{Case $|h_{A_1R}|\geq  |h_{B_1R}|\geq |h_{A_2R}|\geq |h_{B_2R}|$}~\label{sec:FirstCaseUplink}

The decoding order at the relay is as follows. First the relay decodes the Gaussian $x_{A_1}^{(1)}$, then the lattice point  $t$ from $A_1$ and $B_1$, followed by $x_{A_2}^{(1)}$ and finally the lattice point $f$ from $A_2$ and $B_2$.
We can show that for any choice of  $\alpha_{A_i}^{(j)}$ and $\alpha_{B_i}^{(2)}$, this can be done successfully as long as,

\begin{align}\label{eq:AchRatesUplink1}
& R_{A_1}^{(1)}\leq  \\ & C\left(\frac{|h_{A_1R}|^2\alpha_{A_1}^{(1)}P}{2\alpha_{B_1}^{(2)}|h_{B_1R}|^2P+\alpha_{A_2}^{(1)}|h_{A_2R}|^2P+2\alpha_{B_2}^{(2)}|h_{B_2R}|^2P+1}\right) \nonumber
\end{align}
\begin{align}
R_{A_1}^{(2)}, R_{B_1} &\leq \log\left(\frac{|h_{B_1R}|^2\alpha_{B_1}^{(2)}P}{\alpha_{A_2}^{(1)}|h_{A_2R}|^2P+2\alpha_{B_2}^{(2)} |h_{B_2R}|^2P+1}\right)^+\label{eq:AchRatesUplink2}\\
 R_{A_2}^{(2)}, R_{B_2} &\leq  \left(\log \left(\alpha_{B_2}^{(2)}|h_{B_2R}|^2P\right)\right)^+, \label{eq:AchRatesUplink3}\\
  R_{A_2}^{(1)}& \leq  C \left(\frac{|h_{A_2R}|^2\alpha_{A_2}^{(1)}P}
                  {2|h_{B_2R}|^2\alpha_{B_2}^{(2)}P+1}\right) \nonumber.
\end{align}
Details of the derivations are given in Appendix~\ref{sec:AppUplink}.

\subsubsection{Case $|h_{A_1R}|\geq |h_{A_2R}|\geq  |h_{B_1R}|\geq |h_{B_2R}|$}~\label{sec:SecCaseUplink}
The decoding order at the relay is as follows. First the relay decodes the Gaussian $x_{A_1}^{(1)}$ and $x_{A_2}^{(1)}$ simultaneously by treating the remaining signals as noise. Afterwards, the lattice point  $t$ from $A_1$ and $B_1$ is decoded, followed by the lattice point $f$ from $A_2$ and $B_2$.
We can show that for any choice of  $\alpha_{A_i}^{(j)}$ and $\alpha_{B_i}^{(2)}$, this can be done successfully as long as,

\begin{align}\label{eq:AchRatesUplink1Case2MAC1}
R_{A_1}^{(1)}\leq C\left(\frac{|h_{A_1R}|^2\alpha_{A_1}^{(1)}P}{2\alpha_{B_1}^{(2)}|h_{B_1R}|^2P+2\alpha_{B_2}^{(2)}|h_{B_2R}|^2P+1}\right)
\end{align}
\begin{align}\label{eq:AchRatesUplink1Case2MAC2}
R_{A_2}^{(1)}\leq C\left(\frac{\alpha_{A_2}^{(1)}|h_{A_2R}|^2P}{2\alpha_{B_1}^{(2)}|h_{B_1R}|^2P+2\alpha_{B_2}^{(2)}|h_{B_2R}|^2P+1}\right)
\end{align}
\begin{align}\label{eq:AchRatesUplink1Case2MACSum}
R_{A_1}^{(1)}+R_{A_2}^{(1)}\leq C\left(\frac{|h_{A_1R}|^2\alpha_{A_1}^{(1)}P+\alpha_{A_2}^{(1)}|h_{A_2R}|^2P}{2\alpha_{B_1}^{(2)}|h_{B_1R}|^2P+2\alpha_{B_2}^{(2)}|h_{B_2R}|^2P+1}\right)
\end{align}
\begin{align}
R_{A_1}^{(2)}, R_{B_1} &\leq \log\left(\frac{|h_{B_1R}|^2\alpha_{B_1}^{(2)}P}{2\alpha_{B_2}^{(2)} |h_{B_2R}|^2P+1}\right)^+\label{eq:AchRatesUplink2Case2}\\
 R_{A_2}^{(2)}, R_{B_2} &\leq  \left(\log \left(\alpha_{B_2}^{(2)}|h_{B_2R}|^2P\right)\right)^+\label{eq:AchRatesUplink3Case2}.
\end{align}

\subsubsection{Case $|h_{A_1R}|\geq |h_{A_2R}|\geq |h_{B_2R}|\geq  |h_{B_1R}|$}~\label{sec:ThirdCaseUplink}

The decoding is similar to the above case, except that the lattice point $f$ from $A_2$ and $B_2$ is decoded before decoding the lattice point  $t$ from $A_1$ and $B_1$ .
Again, we can show that for any choice of  $\alpha_{A_i}^{(j)}$ and $\alpha_{B_i}^{(2)}$, this can be done successfully as long as,
\begin{align}\label{eq:AchRatesUplink1Case3MAC1}
R_{A_1}^{(1)}\leq C\left(\frac{|h_{A_1R}|^2\alpha_{A_1}^{(1)}P}{2\alpha_{B_1}^{(2)}|h_{B_1R}|^2P+2\alpha_{B_2}^{(2)}|h_{B_2R}|^2P+1}\right)
\end{align}
\begin{align}\label{eq:AchRatesUplink1Case3MAC2}
R_{A_2}^{(1)}\leq C\left(\frac{\alpha_{A_2}^{(1)}|h_{A_2R}|^2P}{2\alpha_{B_1}^{(2)}|h_{B_1R}|^2P+2\alpha_{B_2}^{(2)}|h_{B_2R}|^2P+1}\right)
\end{align}
\begin{align}\label{eq:AchRatesUplink1Case3MACSum}
R_{A_1}^{(1)}+R_{A_2}^{(1)}\leq C\left(\frac{|h_{A_1R}|^2\alpha_{A_1}^{(1)}P+\alpha_{A_2}^{(1)}|h_{A_2R}|^2P}{2\alpha_{B_1}^{(2)}|h_{B_1R}|^2P+2\alpha_{B_2}^{(2)}|h_{B_2R}|^2P+1}\right)
\end{align}
\begin{align}
R_{A_2}^{(2)}, R_{B_2} &\leq \log\left(\frac{|h_{B_2R}|^2\alpha_{B_2}^{(2)}P}{2\alpha_{B_1}^{(2)} |h_{B_1R}|^2P+1}\right)^+\label{eq:AchRatesUplink2Case3}\\
 R_{A_1}^{(2)}, R_{B_1} &\leq  \left(\log \left(\alpha_{B_1}^{(2)}|h_{B_1R}|^2P\right)\right)^+\label{eq:AchRatesUplink3Case3}.
\end{align}

Now we state the following lemma whose proof is given in Appendix~\ref{app:lemAchRegionUL}.
\begin{lemma}\label{lem:AchRegionUL}
Suppose that the nodes are using the transmit strategy described in Section \ref{subsec:encUL}. Then for any 4-tuple $(r_{A_1},r_{B_1},r_{A_2},r_{B_2})$ satisfying
\begin{align}
&r_{A_1}\leq C\left(|h_{A_1R}|^2P\right)-2\;, r_{B_1}\leq C\left(|h_{B_1R}|^2P\right)-1 \label{eq:RatesUplink1}\\
&  r_{A_2}\leq C\left(|h_{A_2R}|^2P\right)-2\;, r_{B_2}\leq C\left(|h_{B_2R}|^2P\right)-1\label{eq:RatesUplink2}\\
 & r_{A_1}+r_{A_2}\leq C\left(|h_{A_1R}|^2P+|h_{A_2R}|^2P\right)-4 \label{eq:RatesUplink3}\\
 &  r_{A_1}+r_{B_2}\leq  C\left(|h_{A_1R}|^2P +|h_{B_2R}|^2P\right) -4\label{eq:RatesUplink4}\\
 & r_{B_1}+r_{B_2}\leq  C\left(|h_{B_1R}|^2P +|h_{B_2R}|^2P\right)-4 \label{eq:RatesUplink5}\\
 &  r_{B_1}+r_{A_2}\leq C\left(|h_{B_1R}|^2P+|h_{A_2R}|^2P\right)-4\label{eq:RatesUplink6},
\end{align}
there exists a choice of power assignments ($\alpha_{A_i}^{(j)}$ and $\alpha_{B_i}^{(2)}$) such that the relay can use the decoding strategy described earlier to decode the Gaussian $x_{A_i}^{(1)}$ of rate $R_{A_i}^{(1)}=r_{A_i}-r_{B_i}$, the lattice point $t$  of rate $R_{A_1}^{(2)}=R_{B_1}=r_{B_1}$, and the lattice point $f$ of rate $R_{A_2}^{(2)}=R_{B_2}=r_{B_2}$, with arbitrary small error probability.
\end{lemma}

\subsection{Encoding at the relay}\label{subsec:encDL}
The relay maps the decoded $x_{A_1}^{(1)}$, $t$, $x_{A_2}^{(1)}$, and $f$ to a Gaussian codeword $x_R^{(1)}$ from a codebook of size $2^{nR_{A_1}^{(1)}}$, $x_R^{(2)}$ from a codebook of size $2^{nR_{B_1}}$, $x_R^{(3)}$ from a codebook of size $2^{nR_{A_2}^{(1)}}$, and $x_R^{(4)}$ from a codebook of size $2^{nR_{B_2}}$, respectively.

\subsection{Downlink: Decoding at the nodes}
 As in the uplink, we have to consider three cases only, from which we provide the detailed analysis for $|h_{RB_1}|\geq  |h_{RA_1}|\geq |h_{RB_2}|\geq |h_{RA_2}|$. The other cases follow similar lines of arguments and thus only the results are presented.

The relay uses a superposition of four messages. One message is decoded by all users. Another message is decoded by both users of the first pair and the strong receiver of the second pair. Yet another message is decoded by only the strong receiver of the first pair,  and finally the remaining message is decoded by both users of the first pair.

\subsubsection{Case $|h_{RB_1}|\geq  |h_{RA_1}|\geq |h_{RB_2}|\geq |h_{RA_2}|$}~\label{sec:FirstCaseDownlink}
We can show that for any choice of  $\alpha_{A_i}^{(j)}$ and $\alpha_{B_i}^{(2)}$, this can be done successfully as long as,

\begin{align}
R_{A_1}^{(2)}, R_{B_1}\leq \min \Big( &  C\left(\frac{|h_{RB_1}|^2\alpha_R^{(2)}P}{1+|h_{RB_1}|^2\alpha_R^{(1)}P}\right), \nonumber \\
& C\left(|h_{RA_1}|^2\alpha_R^{(2)}P\right)\Big),\label{eq:AchRatesDownlink1}
\end{align}
\begin{align}
&R_{A_2}^{(2)},R_{B_2}\leq \min \Big(  C\left(\frac{|h_{RB_2}|^2\alpha_R^{(4)}P}{1+P|h_{RB_2}|^2\sum_{j=1}^3\alpha_R^{(j)}}\right),\nonumber\\
&\quad\quad\quad\quad\quad\quad\quad\quad\quad C\left(\frac{|h_{RA_2}|^2\alpha_R^{(4)}P}{1+P|h_{RA_2}|^2\left(\alpha_R^{(1)}+\alpha_R^{(2)}\right)}\right)\Big),\label{eq:AchRatesDownlink2}
\end{align}
\begin{align} &R_{A_1}^{(1)}\leq C\left(|h_{RB_1}|^2\alpha_R^{(1)}P\right), \label{eq:AchRatesDownlink3}\\
 & R_{A_2}^{(1)}\leq C\left(\frac{|h_{RB_2}|^2\alpha_R^{(3)}P}{1+P|h_{RB_2}|^2\left(\alpha_R^{(1)}+\alpha_R^{(2)}\right)}\right).\nonumber
\end{align}

Details of the derivation are given in Appendix~\ref{sec:AppDownlink}.

\subsubsection{Case $|h_{RB_1}|\geq   |h_{RB_2}|\geq |h_{RA_1}|\geq|h_{RA_2}|$ }~\label{sec:SecCaseDownlink}
We can show that for any choice of  $\alpha_{A_i}^{(j)}$ and $\alpha_{B_i}^{(2)}$, this can be done successfully as long as,

\begin{align}
R_{A_1}^{(2)}, R_{B_1}\leq \min \Big(  & C\left(\frac{|h_{RA_1}|^2\alpha_R^{(2)}P}{1+|h_{RA_1}|^2\alpha_R^{(3)}P}\right), \nonumber \\ & C\left(\frac{|h_{RB_2}|^2\alpha_R^{(2)}P}{1+P|h_{RB_2}|^2\left(\alpha_R^{(1)}+\alpha_R^{(3)}\right)}\right)\Big),\label{eq:AchRatesDownlink1Case2}
\end{align}
\begin{align}
R_{A_2}^{(2)},R_{B_2}\leq \min \Big( &  C\left(\frac{|h_{RB_2}|^2\alpha_R^{(4)}P}{1+P|h_{RB_2}|^2\sum_{j=1}^3\alpha_R^{(j)}}\right),\nonumber\\
&C\left(\frac{|h_{RA_1}|^2\alpha_R^{(4)}P}{1+P|h_{RA_1}|^2\left(\alpha_R^{(2)}+\alpha_R^{(3)}\right)}\right), \nonumber \\ & C\left(\frac{|h_{RA_2}|^2\alpha_R^{(4)}P}{1+P|h_{RA_2}|^2\left(\alpha_R^{(1)}+\alpha_R^{(2)}\right)}\right)\Big),\label{eq:AchRatesDownlink2Case2}
\end{align}
\begin{align}
&R_{A_1}^{(1)}\leq C\left(|h_{RB_1}|^2\alpha_R^{(1)}P\right),\label{eq:AchRatesDownlink3Case2} \\
 & R_{A_2}^{(1)}\leq C\left(\frac{|h_{RB_2}|^2\alpha_R^{(3)}P}{1+P|h_{RB_2}|^2\alpha_R^{(1)}}\right).\nonumber
\end{align}

\subsubsection{Case $|h_{RB_1}|\geq   |h_{RB_2}|\geq|h_{RA_2}|\geq |h_{RA_1}|$ }~\label{sec:ThirdCaseDownlink}
We can show that for any choice of  $\alpha_{A_i}^{(j)}$ and $\alpha_{B_i}^{(2)}$, this can be done successfully as long as,

\begin{align}
R_{A_1}^{(2)},R_{B_1}\leq \min \Big( & C\left(\frac{|h_{RB_2}|^2P\alpha_R^{(2)}}{1+|h_{RB_2}|^2P\sum_{j=1,j\neq 2}^4\alpha_R^{(j)}}\right),\nonumber\\
&C\left(\frac{|h_{RA_1}|^2P\alpha_R^{(2)}}{1+|h_{RA_1}|^2P\left(\alpha_R^{(4)}+\alpha_R^{(3)}\right)}\right), \nonumber \\ & C\left(\frac{|h_{RA_2}|^2P\alpha_R^{(2)}}{1+|h_{RA_2}|^2P\left(\alpha_R^{(1)}+\alpha_R^{(4)}\right)}\right)\Big),\label{eq:AchRatesDownlink2Case3}
\end{align}
\begin{align}
R_{A_2}^{(2)}, R_{B_2}\leq \min \Big( &  C\left(\frac{|h_{RA_2}|^2P\alpha_R^{(4)}}{1+|h_{RA_2}|^2P\alpha_R^{(1)}}\right), \nonumber \\
& C\left(\frac{|h_{RB_2}|^2P\alpha_R^{(4)}}{1+|h_{RB_2}|^2P\left(\alpha_R^{(1)}+\alpha_R^{(3)}\right)}\right)\Big),\label{eq:AchRatesDownlink1Case3}
\end{align}
\begin{align}
&R_{A_1}^{(1)}\leq C\left(|h_{RB_1}|^2P\alpha_R^{(1)}\right),\quad R_{A_2}^{(1)}\leq C\left(\frac{|h_{RB_2}|^2P\alpha_R^{(3)}}{1+|h_{RB_2}|^2P\alpha_R^{(1)}}\right).\label{eq:AchRatesDownlink3Case3}
\end{align}

Now we state the following lemma whose proof is given in Appendix~\ref{app:lemAchRegionDL}.

\begin{lemma}\label{lem:AchRegionDL}
Suppose that the relay is using the transmit strategy described above. Then for any 4-tuple $(r_{A_1},r_{B_1},r_{A_2},r_{B_2})$ satisfying
\begin{align}
\label{eq:RatesDownlink1}&r_{A_1}\leq C\left(|h_{RB_1}|^2P\right)-2\;, r_{B_1}\leq C\left(|h_{RA_1}|^2P\right)-2\\
\label{eq:RatesDownlink2}&  r_{A_2}\leq C\left(|h_{RB_2}|^2P\right)-2\;, r_{B_2}\leq C\left(|h_{RA_2}|^2P\right)-2
\end{align}
\begin{align}
\label{eq:RatesDownlink3} & r_{A_1}+r_{A_2}\leq C\left(\max\left(|h_{RB_1}|^2P,|h_{RB_2}|^2P\right)\right)-3\\
\label{eq:RatesDownlink4} &  r_{A_1}+r_{B_2}\leq  C\left(\max\left(|h_{RB_1}|^2P ,|h_{RA_2}|^2P\right)\right) -3\\
\label{eq:RatesDownlink5} & r_{B_1}+r_{B_2}\leq  C\left(\max\left(|h_{RA_1}|^2P ,|h_{RA_2}|^2P\right)\right)-3 \\
\label{eq:RatesDownlink6} &  r_{B_1}+r_{A_2}\leq C\left(\max\left(|h_{RA_1}|^2P,|h_{RB_2}|^2P\right)\right)-3
\end{align}
there exists a choice of power assignments ($\alpha_{R}^{(j)}$'s) such that $B_1$ can decode the Gaussian codewords $x_{R}^{(1)}$ of rate $R_{A_1}^{(1)}=r_{A_1}-r_{B_1}$, $A_1$ and $B_1$ can both decode the Gaussian codeword $x_{R}^{(2)}$  of rate $R_{A_1}^{(2)}=R_{B_1}=r_{B_1}$,  $B_2$ can decode the Gaussian codeword $x_{R}^{(3)}$ of rate $R_{R}^{(3)}=r_{A_2}-r_{B_2}$, and $A_2$ and $B_2$ can both decode the Gaussian codeword $x_{R}^{(4)}$ of rate $R_{A_2}^{(2)}=R_{B_2}=r_{B_2}$, with arbitrary small error probability.
\end{lemma}

Now note that if
\begin{align}
 (R_{A_1},R_{B_1},R_{A_2},R_{B_2}) \in \overline{\mathcal{C}} \nonumber
 \end{align}
and  $R_{A_i}, R_{B_i} \geq 2$ for $i=1,2$, then the rate tuple
$$(r_{A_1},r_{B_1},r_{A_2},r_{B_2})=(R_{A_1}-2,R_{B_1}-2,R_{A_2}-2,R_{B_2}-2)$$ satisfies the conditions of both Lemma~\ref{lem:AchRegionUL} and~\ref{lem:AchRegionDL}. Therefore by the proposed strategy the rate tuple $(R_{A_1}-2,R_{B_1}-2,R_{A_2}-2,R_{B_2}-2)$ is achievable, and this completes the proof of Theorem \ref{theo:Main}.

\section{Conclusion}\label{sec:Conc}
In this paper we studied the multi-pair bidirectional relay network which is a generalization of the bidirectional relay channel. We examined this problem in the context of the linear shift deterministic channel model introduced in \cite{AvestimehrIT} and characterized its capacity region completely in both full-duplex and half-duplex cases. We also showed that the capacity can be achieved by a divide-and-conquer {relaying} strategy.
Based on insights gained from the linear shift deterministic channel model, we proposed a transmission strategy for the Gaussian two-pair bidirectional full-duplex relay network and found an approximate characterization of the capacity region. In fact, we proposed a specific superposition coding scheme that achieves to within $3$ bits/sec/Hz per user of the cut-set upper-bound on the capacity of the two-pair bidirectional relay network. Possible directions for future work is the extension to the half-duplex mode. Extension of the proposed transmission strategy to the case that there are more than two pairs is possible, however, analyzing the gap between the achievable rate of the corresponding scheme and the cut-set upper-bound is expected to be quite cumbersome.

\appendices

\section{\label{app:proofDetCase1}}

In this Appendix we prove that the reduced rate-tuple $\mathbf{R}'=(R_{A_1}-1,R_{B_1}-1,R_{A_2},R_{B_2},\cdots, R_{A_M},R_{B_M})$, with {$R_{A_1}\geq 1$ and $R_{B_1}\geq 1$}, created in case 1 of the proof of Theorem \ref{thm:main} is in the cut-set region of the reduced network (defined in (\ref{eq:reduced1})-(\ref{eq:reduced4})), i.e.,
\begin{align*}
\sum_{i \in \mathcal{U}} [ \ell_i R'_{A_i} + & (1-\ell_i) R'_{B_i} ] \\
\leq \min \Big( & \max_{i \in \mathcal{U} } ( \ell_i n'_{A_iR} + (1-\ell_i) n'_{B_iR} ),\\
&  \max_{i \in \mathcal{U} } ( \ell_i n'_{RB_i} + (1-\ell_i) n'_{RA_i} ) \Big),
\end{align*}
for all $\mathcal{U}\subseteq \{1, \dots, M \}$ and $\ell_i \in \{0,1\}$, $i=1,\ldots,M$.

If $1 \in \mathcal{U}$, we have
\begin{align*}
  \sum_{i \in \mathcal{U}} [ \ell_i R'_{A_i} + & (1-\ell_i) R'_{B_i} ] \\
 \stackrel{(1 \in \mathcal{U})}{=}\sum_{i \in \mathcal{U}} & [ \ell_i R_{A_i} +  (1-\ell_i) R_{B_i} ] -1 \\
 \stackrel{\eqref{eq:cutSetDet}}{\leq}   \min \Big( & \max_{i \in \mathcal{U} } ( \ell_i n_{A_iR} + (1-\ell_i) n_{B_iR} )\\
  & \max_{i \in \mathcal{U} } ( \ell_i n_{RB_i} + (1-\ell_i) n_{RA_i} ) \Big) -1 \\
  =  \min \Big( & \max_{i \in \mathcal{U} } ( \ell_i (n_{A_iR}-1) + (1-\ell_i) (n_{B_iR}-1) ), \\
  & \max_{i \in \mathcal{U} } ( \ell_i (n_{RB_i}-1) + (1-\ell_i) (n_{RA_i}-1) ) \Big) \\
 \stackrel{\eqref{eq:reduced1}-\eqref{eq:reduced4}}{\leq}   \min \Big( & \max_{i \in \mathcal{U} } ( \ell_i n'_{A_iR} + (1-\ell_i) n'_{B_iR} ), \\
 & \max_{i \in \mathcal{U} } ( \ell_i n'_{RB_i} + (1-\ell_i) n'_{RA_i} ) \Big).
\end{align*}

If $1 \notin \mathcal{U}$, we have
\begin{align*}  \sum_{i \in \mathcal{U}} [ \ell_i R'_{A_i} & + (1-\ell_i) R'_{B_i} ] \stackrel{(1 \notin \mathcal{U})}{=} \sum_{i \in \mathcal{U}} [ \ell_i R_{A_i} + (1-\ell_i) R_{B_i} ]  \\
\stackrel{\eqref{eq:cutSetDet}}{\leq}   \min \Big( &  \max_{i \in \mathcal{U} } ( \ell_i n_{A_iR} + (1-\ell_i) n_{B_iR} ), \\
&  \max_{i \in \mathcal{U} } ( \ell_i n_{RB_i} + (1-\ell_i) n_{RA_i} ) \Big) \\
  =  \min \Big( & \max_{i \in \mathcal{U} } ( \ell_i (n_{A_iR}-1) + (1-\ell_i) (n_{B_iR}-1) ), \\
  & \max_{i \in \mathcal{U} } ( \ell_i (n_{RB_i}-1) + (1-\ell_i) (n_{RA_i}-1) ) \Big) +1  \\
\stackrel{\eqref{eq:reduced1}-\eqref{eq:reduced4}}{\leq}   \min \Big( & \max_{i \in \mathcal{U} } ( \ell_i n'_{A_iR} + (1-\ell_i) n'_{B_iR} ), \\
& \max_{i \in \mathcal{U} } ( \ell_i n'_{RB_i} + (1-\ell_i) n'_{RA_i} ) \Big) +1.
\end{align*}
Therefore, if $1 \notin \mathcal{U}$, the only way to violate the cut-set bound is to have all above inequalities as equality, i.e.,
\begin{align} \label{eq:lemmaDet0}  \sum_{i \in \mathcal{U}} [ \ell_i R_{A_i} & + (1-\ell_i) R_{B_i} ] \\
 =   \min \Big( & \max_{i \in \mathcal{U} } ( \ell_i n_{A_iR} + (1-\ell_i) n_{B_iR} ), \nonumber \\
 & \max_{i \in \mathcal{U} } ( \ell_i n_{RB_i} + (1-\ell_i) n_{RA_i} ) \Big)
 \end{align}
 \begin{align}
\label{eq:lemmaDet1}  = \min \Big ( &  \max_{i \in \mathcal{U} } ( \ell_i n'_{A_iR} + (1-\ell_i) n'_{B_iR} ),  \\
& \max_{i \in \mathcal{U} } ( \ell_i n'_{RB_i} + (1-\ell_i) n'_{RA_i} ) \Big) +1  \nonumber.
\end{align}
However, we show that this is in contradiction to our assumption of $R_{A_1}\neq 0$ and $R_{B_1}\neq 0$. To see this, note that by \eqref{eq:reduced1}-\eqref{eq:reduced4}, the equality in \eqref{eq:lemmaDet1} happens only if we have one the following four cases.
\begin{enumerate}
\item $\exists j \in \{2,\ldots,M \}$ such that,  $j \in  \mathcal{U}$, $ \ell_j=1$, and
\begin{align}
\label{eq:lemmaDet2} n'_{A_jR} &=  \max_{i \in \mathcal{U} } ( \ell_i n'_{A_iR} + (1-\ell_i) n'_{B_iR} ), \\
\label{eq:lemmaDet3} n'_{A_jR} &=  \min \Big( \max_{i \in \mathcal{U} } ( \ell_i n'_{A_iR} + (1-\ell_i) n'_{B_iR} ), \\
& \qquad \qquad \max_{i \in \mathcal{U} } ( \ell_i n'_{RB_i} + (1-\ell_i) n'_{RA_i} ) \Big) \nonumber \\
\label{eq:lemmaDet4}  n'_{A_jR} &=  n_{A_jR}-1.
\end{align}
First, note that from (\ref{eq:lemmaDet2}), (\ref{eq:lemmaDet3}), $\ell_j=1$, and the relationship between the channel gains in the original network and the channel gains in the  reduced network (\ref{eq:reduced1})-(\ref{eq:reduced4}), we have
\begin{align}
\label{eq:lemmaDet7} n_{A_jR} &=  \max_{i \in \mathcal{U} } ( \ell_i n_{A_iR} + (1-\ell_i) n_{B_iR} ), \\
\label{eq:lemmaDet8} n_{A_jR} &\leq  \max_{i \in \mathcal{U} } ( \ell_i n_{RB_i} + (1-\ell_i) n_{RA_i} ), \\
\label{eq:lemmaDet9} n_{A_jR} &=  \min \Big( \max_{i \in \mathcal{U} } ( \ell_i n_{A_iR} + (1-\ell_i) n_{B_iR} ), \\
& \qquad \qquad \max_{i \in \mathcal{U} } ( \ell_i n_{RB_i} + (1-\ell_i) n_{RA_i} ) \Big) \nonumber
\end{align}
Since $n'_{A_jR}=n_{A_jR}-1$, we have $n_{A_jR} \geq l_u=\min(n_{A_1R},n_{B_1R})$. If $n_{A_jR} \geq n_{A_1R}$, we can write
\begin{align}
& R_{A_1} +  \sum_{i \in \mathcal{U}} [ \ell_i R_{A_i} + (1-\ell_i) R_{B_i} ]  \nonumber \\
& \leq \min \Big( \max \lp n_{A_1R}, \max_{i \in \mathcal{U} } ( \ell_i n_{A_iR} + (1-\ell_i) n_{B_iR} ) \rp , \nonumber \\
&   \qquad \qquad  \max \lp n_{RB_1}, \max_{i \in \mathcal{U} } ( \ell_i n_{RB_i} + (1-\ell_i) n_{RA_i} ) \rp \Big) \nonumber \\
& \stackrel{\eqref{eq:lemmaDet7}}{=} \min \Big( \max \lp n_{A_1R}, n_{A_jR} \rp , \nonumber \\
& \qquad \qquad \max \lp n_{RB_1}, \max_{i \in \mathcal{U} } ( \ell_i n_{RB_i} + (1-\ell_i) n_{RA_i} ) \rp \Big) \nonumber \\
& \stackrel{(n_{A_jR} \geq n_{A_1R})}{=} \min \Big(  n_{A_jR}  , \nonumber \\
& \qquad \qquad \qquad \qquad \max \Big( n_{RB_1}, \nonumber \\
& \qquad \qquad \qquad \qquad \quad \max_{i \in \mathcal{U} } ( \ell_i n_{RB_i} + (1-\ell_i) n_{RA_i} ) \Big) \Big) \nonumber \\
& \stackrel{\eqref{eq:lemmaDet8}}{=}  n_{A_jR},\nonumber  \\
\label{eq:lemmaDet10} & \stackrel{\eqref{eq:lemmaDet9}}{=} \min \Big(  \max_{i \in \mathcal{U} } ( \ell_i n_{A_iR} + (1-\ell_i) n_{B_iR} ), \\
 & \qquad \qquad \max_{i \in \mathcal{U} } ( \ell_i n_{RB_i} + (1-\ell_i) n_{RA_i} ) \Big)
\end{align}
where the first step is true since $\mathbf{R}$ satisfies the cut-set bound (\ref{eq:cutSetDet}) with $\tilde{\mathcal{U}}=\mathcal{U}\cup \{1\}$ and $\ell_1=1$. Combining (\ref{eq:lemmaDet0}) and (\ref{eq:lemmaDet10}), we get $R_{A_1} \leq 0$, which is a contradiction to our assumption of  $R_{A_1} \geq 1$.

Similarly, if $n_{A_jR} \geq n_{B_1R}$, we can write
\begin{align}
& R_{B_1} +  \sum_{i \in \mathcal{U}} [ \ell_i R_{A_i} + (1-\ell_i) R_{B_i} ] \nonumber \\
& \leq \min \lp \max \lp n_{B_1R}, \max_{i \in \mathcal{U} } ( \ell_i n_{A_iR} + (1-\ell_i) n_{B_iR} ) \rp \right. , \nonumber \\
&  \left. \qquad  \qquad  \max \lp n_{RA_1}, \max_{i \in \mathcal{U} } ( \ell_i n_{RB_i} + (1-\ell_i) n_{RA_i} ) \rp \rp \nonumber \\
& \stackrel{\eqref{eq:lemmaDet7}}{=} \min \Big( \max \lp n_{B_1R}, n_{A_jR} \rp , \nonumber \\
& \qquad \qquad \max \lp n_{RA_1}, \max_{i \in \mathcal{U} } ( \ell_i n_{RB_i} + (1-\ell_i) n_{RA_i} ) \rp \Big) \nonumber \\
& \stackrel{(n_{A_jR} \geq n_{B_1R})}{=} \min \Big(  n_{A_jR}  , \nonumber \\
& \qquad \qquad \max \lp n_{RA_1}, \max_{i \in \mathcal{U} } ( \ell_i n_{RB_i} + (1-\ell_i) n_{RA_i} ) \rp \Big) \nonumber \\
& \stackrel{\eqref{eq:lemmaDet8}}{=}  n_{A_jR},\nonumber \\
\label{eq:lemmaDet10b} & \stackrel{\eqref{eq:lemmaDet9}}{=} \min \lp \max_{i \in \mathcal{U} } ( \ell_i n_{A_iR} + (1-\ell_i) n_{B_iR} ), \right. \\
& \qquad \qquad \left. \max_{i \in \mathcal{U} } ( \ell_i n_{RB_i} + (1-\ell_i) n_{RA_i} ) \rp \nonumber
\end{align}
where the first step is true since $\mathbf{R}$ satisfies the cut-set bound (\ref{eq:cutSetDet}) with $\tilde{\mathcal{U}}=\mathcal{U}\cup \{1\}$ and $\ell_1=0$. Combining (\ref{eq:lemmaDet0}) and (\ref{eq:lemmaDet10b}), we get $R_{B_1} \leq 0$, which is a contradiction to our assumption of  $R_{B_1} \geq 1$. Therefore, this case can not happen.

\item $\exists j \in \{2,\ldots,M \}$ such that,  $j \in  \mathcal{U}$, $ \ell_j=0$, and
\begin{align*}
 n'_{B_jR} &=  \max_{i \in \mathcal{U} } ( \ell_i n'_{A_iR} + (1-\ell_i) n'_{B_iR} ), \\
 n'_{B_jR} &=  \min \lp \max_{i \in \mathcal{U} } ( \ell_i n'_{A_iR} + (1-\ell_i) n'_{B_iR} ), \right. \\
 & \qquad \qquad \left. \max_{i \in \mathcal{U} } ( \ell_i n'_{RB_i} + (1-\ell_i) n'_{RA_i} ) \rp\\
 n'_{B_jR} &=  n_{B_jR}-1.
\end{align*}
The proof that this case can not also happen is very similar to the previous case, hence we omit repetition.

\item $\exists j \in \{2,\ldots,M \}$ such that,  $j \in  \mathcal{U}$, $ \ell_j=0$, and
\begin{align}
\label{eq:lemmaDet11} n'_{RB_j} &=  \max_{i \in \mathcal{U} } ( \ell_i n'_{RB_i} + (1-\ell_i) n'_{RA_i} ), \\
\label{eq:lemmaDet12} n'_{RB_j} &=  \min \lp \max_{i \in \mathcal{U} } ( \ell_i n'_{A_iR} + (1-\ell_i) n'_{B_iR} ), \right. \\
& \qquad \qquad \left. \max_{i \in \mathcal{U} } ( \ell_i n'_{RB_i} + (1-\ell_i) n'_{RA_i} ) \rp \nonumber \\
\label{eq:lemmaDet13}  n'_{RB_j} &=  n_{RB_j}-1.
\end{align}
From (\ref{eq:lemmaDet11}), (\ref{eq:lemmaDet12}), $\ell_j=0$, and the relationship between the channel gains in the original network and the channel gains in the  reduced network  (\ref{eq:reduced1})-(\ref{eq:reduced4}), we have
\begin{align}
\label{eq:lemmaDet13a}n_{RB_j} &=  \max_{i \in \mathcal{U} } ( \ell_i n_{RB_i} + (1-\ell_i) n_{RA_i} ), \\
\label{eq:lemmaDet14} n_{RB_j} &\leq  \max_{i \in \mathcal{U} } ( \ell_i n_{A_iR} + (1-\ell_i) n_{B_iR} ), \\
\label{eq:lemmaDet15} n_{RB_j} &=  \min \lp \max_{i \in \mathcal{U} } ( \ell_i n_{A_iR} + (1-\ell_i) n_{B_iR} ), \right. \\
& \qquad \qquad \left. \max_{i \in \mathcal{U} } ( \ell_i n_{RB_i} + (1-\ell_i) n_{RA_i} ) \rp.  \nonumber
\end{align}
Since $n'_{RB_j}=n_{RB_j}-1$, we have $n_{RB_j} \geq l_u=\min(n_{RA_1},n_{RB_1})$. If $n_{RB_j} \geq n_{RB_1}$, we can write
\begin{align}
& R_{A_1} +  \sum_{i \in \mathcal{U}} [ \ell_i R_{A_i} + (1-\ell_i) R_{B_i} ]  \nonumber \\
& \leq \min \lp \max \lp n_{A_1R}, \max_{i \in \mathcal{U} } ( \ell_i n_{A_iR} + (1-\ell_i) n_{B_iR} ) \rp , \right. \nonumber\\&  \left. \quad \quad \quad \quad  \max \lp n_{RB_1}, \max_{i \in \mathcal{U} } ( \ell_i n_{RB_i} + (1-\ell_i) n_{RA_i} ) \rp \rp \nonumber\\
& \stackrel{\eqref{eq:lemmaDet13a}}{=} \min \lp  \max \lp n_{A_1R}, \max_{i \in \mathcal{U} } ( \ell_i n_{A_iR} + (1-\ell_i) n_{B_iR} ) \rp , \right. \nonumber \\
& \qquad \qquad \left. \max \lp n_{RB_1}, n_{RB_j} \rp \rp \nonumber\\
& \stackrel{(n_{RB_j} \geq n_{RB_1})}{=} \min \Big(  \max \Big( n_{A_1R},  \nonumber \\
& \qquad \qquad  \qquad \max_{i \in \mathcal{U} } ( \ell_i n_{A_iR} + (1-\ell_i) n_{B_iR} ) \Big) , n_{RB_j}  \Big) \nonumber\\
& \stackrel{\eqref{eq:lemmaDet14}}{=}  n_{RB_j},\nonumber\\
\label{eq:lemmaDet16} & \stackrel{\eqref{eq:lemmaDet15}}{=} \min \lp \max_{i \in \mathcal{U} } ( \ell_i n_{A_iR} + (1-\ell_i) n_{B_iR} ), \right. \\
 & \qquad \qquad \left. \max_{i \in \mathcal{U} } ( \ell_i n_{RB_i} + (1-\ell_i) n_{RA_i} ) \rp \nonumber
 \end{align}
where the first step is true since $\mathbf{R}$ satisfies the cut-set bound (\ref{eq:cutSetDet}) with $\tilde{\mathcal{U}}=\mathcal{U}\cup \{1\}$ and $\ell_1=1$. Combining (\ref{eq:lemmaDet0}) and (\ref{eq:lemmaDet16}), we get $R_{A_1} \leq 0$, which is a contradiction to our assumption of  $R_{A_1} \geq 1$.

Similarly, if $n_{RB_j} \geq n_{RA_1}$, we can write
\begin{align}
& R_{B_1} +  \sum_{i \in \mathcal{U}} [ \ell_i R_{A_i} + (1-\ell_i) R_{B_i} ] \nonumber \\
 & \leq \min \lp \max \lp n_{B_1R}, \max_{i \in \mathcal{U} } ( \ell_i n_{A_iR} + (1-\ell_i) n_{B_iR} ) \rp , \right. \nonumber\\&  \left. \qquad \qquad \max \lp n_{RA_1}, \max_{i \in \mathcal{U} } ( \ell_i n_{RB_i} + (1-\ell_i) n_{RA_i} ) \rp \rp \nonumber\\
& \stackrel{\eqref{eq:lemmaDet13}}{=} \min \lp  \max \lp n_{B_1R}, \max_{i \in \mathcal{U} } ( \ell_i n_{A_iR} + (1-\ell_i) n_{B_iR} ) \rp , \right. \nonumber \\
 & \qquad \qquad \left. \max \lp n_{RA_1}, n_{RB_j} \rp \rp \nonumber\\
& \stackrel{(n_{RB_j} \geq n_{RA_1})}{=} \min \Big(  \max \Big( n_{B_1R}, \nonumber \\
 & \qquad \qquad \max_{i \in \mathcal{U} } ( \ell_i n_{A_iR} + (1-\ell_i) n_{B_iR} ) \Big) , n_{RB_j}  \Big) \nonumber\\
& \stackrel{\eqref{eq:lemmaDet14}}{=}  n_{RB_j},\nonumber\\
\label{eq:lemmaDet17} & \stackrel{\eqref{eq:lemmaDet15}}{=} \min \lp \max_{i \in \mathcal{U} } ( \ell_i n_{A_iR} + (1-\ell_i) n_{B_iR} ), \right. \\
& \qquad \qquad  \left. \max_{i \in \mathcal{U} } ( \ell_i n_{RB_i} + (1-\ell_i) n_{RA_i} ) \rp \nonumber
\end{align}
where the first step is true since $\mathbf{R}$ satisfies the cut-set bound (\ref{eq:cutSetDet}) with $\tilde{\mathcal{U}}=\mathcal{U}\cup \{1\}$ and $\ell_1=0$. Combining (\ref{eq:lemmaDet0}) and (\ref{eq:lemmaDet17}), we get $R_{B_1} \leq 0$, which is a contradiction to our assumption of  $R_{A_1} \geq 1$.

\item $\exists j \in \{2,\ldots,M \}$ such that,  $j \in  \mathcal{U}$, $ \ell_j=0$, and
\begin{align*}
 n'_{RA_j} &=  \max_{i \in \mathcal{U} } ( \ell_i n'_{RB_i} + (1-\ell_i) n'_{RA_i} ), \\
 n'_{RA_j} &=  \min \lp \max_{i \in \mathcal{U} } ( \ell_i n'_{A_iR} + (1-\ell_i) n'_{B_iR} ), \right. \\
 & \qquad \qquad \left. \max_{i \in \mathcal{U} } ( \ell_i n'_{RB_i} + (1-\ell_i) n'_{RA_i} ) \rp\\
  n'_{RA_j} &=  n_{RA_j}-1.
\end{align*}
The proof that this case can not also happen is very similar to the previous case, hence we omit repetition.
\end{enumerate}

\section{\label{app:proofDetCase2}}

In this Appendix we prove that the reduced rate-tuple $\mathbf{R}'=(R_{A_1}-1,0,R_{A_2},0,\cdots, R_{A_M},0)$, with {$R_{A_1}\geq 1$}, created in case 2 of the proof of Theorem \ref{thm:main} is in the cut-set region of the reduced network (defined in (\ref{eq:reduced1})-(\ref{eq:reduced4})). Since $R_{B_1}=\cdots=R_{B_M}=0$, we just need to show that
\begin{align*}
\sum_{i \in \mathcal{U}} R'_{A_i}  \leq & \min \lp \max_{i \in \mathcal{U} }n'_{A_iR}, \max_{i \in \mathcal{U} }n'_{RB_i}\rp, \\
 & \qquad \forall \mathcal{U}\subseteq \{1, \dots, M \}.
\end{align*}

If $1 \in \mathcal{U}$, we have
\begin{align*}
 & \sum_{i \in \mathcal{U}}  R'_{A_i}  \stackrel{(1 \in \mathcal{U})}{=}\sum_{i \in \mathcal{U}} R_{A_i} -1 \\
& \quad \stackrel{\eqref{eq:cutSetDet}}{\leq}   \min \lp \max_{i \in \mathcal{U} } n_{A_iR} , \max_{i \in \mathcal{U} } n_{RB_i} \rp -1 \\
&\quad  =  \min \lp \max_{i \in \mathcal{U} } (n_{A_iR}-1), \max_{i \in \mathcal{U} } ( n_{RB_i}-1) \rp \\
& \quad \stackrel{\eqref{eq:reduced1}, \eqref{eq:reduced4}}{\leq}   \min \lp \max_{i \in \mathcal{U} }  n'_{A_iR} , \max_{i \in \mathcal{U} } n'_{RB_i}  \rp.
\end{align*}

If $1 \notin \mathcal{U}$, we have
\begin{align*} & \sum_{i \in \mathcal{U}} R'_{A_i} \stackrel{(1 \notin \mathcal{U})}{=} \sum_{i \in \mathcal{U}} R_{A_i}  \\
& \quad \stackrel{\eqref{eq:cutSetDet}}{\leq}   \min \lp \max_{i \in \mathcal{U} } n_{A_iR} , \max_{i \in \mathcal{U} } n_{RB_i} \rp \\
&\quad  =  \min \lp \max_{i \in \mathcal{U} } (n_{A_iR}-1) , \max_{i \in \mathcal{U} } (n_{RB_i}-1)  \rp +1  \\
& \quad \stackrel{\eqref{eq:reduced1},\eqref{eq:reduced4}}{\leq}   \min \lp \max_{i \in \mathcal{U} } n'_{A_iR}, \max_{i \in \mathcal{U} } n'_{RB_i}  \rp +1.
\end{align*}
Therefore, if $1 \notin \mathcal{U}$, the only way to violate the cut-set bound is to have all above inequalities as equality, i.e.,
\begin{align} \label{eq:lemmaDet0b}  \sum_{i \in \mathcal{U}} R_{A_i}   &=   \min \lp \max_{i \in \mathcal{U} } n_{A_iR} , \max_{i \in \mathcal{U} }  n_{RB_i} \rp \\
\label{eq:lemmaDet1b} & = \min \lp \max_{i \in \mathcal{U} } n'_{A_iR}, \max_{i \in \mathcal{U} } n'_{RB_i}  \rp +1 .
\end{align}
However, we show that this is in contradiction to our assumption of $R_{A_1}\geq 1$. To see this, note that by (\ref{eq:reduced1}-\ref{eq:reduced4}), the equality in (\ref{eq:lemmaDet0b}) and (\ref{eq:lemmaDet1b}) happens only if we have one the following two cases.
\begin{enumerate}
\item $\exists j \in \{2,\ldots,M \}$ such that,  $j \in  \mathcal{U}$ and
\begin{align}
\label{eq:lemmaDet2b} n'_{A_jR} &=  \max_{i \in \mathcal{U} } n'_{A_iR} , \\
\label{eq:lemmaDet3b} n'_{A_jR} &=  \min \lp \max_{i \in \mathcal{U} } n'_{A_iR} , \max_{i \in \mathcal{U} }  n'_{RB_i}  \rp,\\
\label{eq:lemmaDet4b} n_{A_jR} &=  \max_{i \in \mathcal{U} } n_{A_iR} , \\
\label{eq:lemmaDet5b} n_{A_jR} &=  \min \lp \max_{i \in \mathcal{U} } n_{A_iR} , \max_{i \in \mathcal{U} }  n_{RB_i}  \rp,\\
\label{eq:lemmaDet6b}  n'_{A_jR} &=  n_{A_jR}-1.
\end{align}

Since $n'_{A_jR}=n_{A_jR}-1$, we have $n_{A_jR} \geq l_u=n_{A_1R}$. Hence, we can write
\begin{align}
& R_{A_1} +  \sum_{i \in \mathcal{U}} R_{A_i}  \leq \min \lp \max \lp n_{A_1R}, \max_{i \in \mathcal{U} } n_{A_iR} \rp, \right. \nonumber \\
& \qquad \qquad \qquad  \qquad \qquad \left. \max \lp n_{RB_1}, \max_{i \in \mathcal{U} } n_{RB_i} \rp \rp \nonumber\\
& \stackrel{\eqref{eq:lemmaDet4b}}{=} \min \lp \max \lp n_{A_1R}, n_{A_jR} \rp , \right. \nonumber \\
 & \qquad \qquad \left. \max \lp n_{RB_1}, \max_{i \in \mathcal{U} } n_{RB_i} \rp \rp \nonumber\\
& \stackrel{(n_{A_jR} \geq n_{A_1R})}{=} \min \lp  n_{A_jR}  , \max \lp n_{RB_1}, \max_{i \in \mathcal{U} } n_{RB_i}  \rp \rp \nonumber\\
\label{eq:lemmaDet7b} & \stackrel{\eqref{eq:lemmaDet5b}}{=}  n_{A_jR}= \min \lp \max_{i \in \mathcal{U} } n_{A_iR} , \max_{i \in \mathcal{U} }  n_{RB_i}  \rp,
\end{align}
where the first step is true since $\mathbf{R}$ satisfies the cut-set bound (\ref{eq:cutSetDet}) with $\tilde{\mathcal{U}}=\mathcal{U}\cup \{1\}$. Combining (\ref{eq:lemmaDet0b}) and (\ref{eq:lemmaDet7b}), we get $R_{A_1} \leq 0$, which is a contradiction to our assumption of  $R_{A_1} \geq 1$. Therefore, this case can not happen.

\item $\exists j \in \{2,\ldots,M \}$ such that,  $j \in  \mathcal{U}$ and
\begin{align*}
 n'_{RB_j} &=  \max_{i \in \mathcal{U} } n'_{RB_i} , \\
n'_{RB_j} &=  \min \lp \max_{i \in \mathcal{U} } n'_{A_iR} , \max_{i \in \mathcal{U} }  n'_{RB_i}  \rp,\\
 n_{RB_j} &=  \max_{i \in \mathcal{U} } n_{RB_i} , \\
 n_{A_jR} &=  \min \lp \max_{i \in \mathcal{U} } n_{A_iR} , \max_{i \in \mathcal{U} }  n_{RB_i}  \rp,\\
 n'_{RB_j} &=  n_{RB_j}-1.
\end{align*}
The proof that this case can not also happen is very similar to the previous case, hence we omit repetition.
\end{enumerate}

\section{Proof of Lemma \ref{lem:CaseReduction}}~\label{app:lemCaseReduction}
Since the proof for both pairs are similar, we only bring the proof for pair $i=1$.
We claim that if $|h_{B_1R}| > |h_{A_1R}|$ and $\mathbf{R}\in \mathcal{C}_u$, then  $\mathbf{R}\in \mathcal{\tilde{C}}_u$, where $ \mathcal{\tilde{C}}_u$ is the up-link  cut-set region of the network resulted by weakening  $|h_{B_1R}|$ and setting it equal to $|h_{A_1R}|$. We call the new (undermined) uplink channel gains $(\tilde{h}_{A_1R},\tilde{h}_{B_1R},\tilde{h}_{A_2R},\tilde{h}_{B_2R})$. The claim is justified by check marking equations (\ref{eq:GaussCut-set bound 1}) to (\ref{eq:GaussCut-set bound 8}) for new capacities (with infinite down-link channel gains). The only non-obvious inequalities are the ones in which $\tilde{h}_{B_1R}$ appears. By symmetry we only have to verify that (\ref{eq:GaussCut-set bound 2}) and (\ref{eq:GaussCut-set bound 8}) hold.  Start with the original equations for $(h_{A_1R},h_{B_1R},h_{A_2R},h_{B_2R})$  and note that the LHS of equations (\ref{eq:GaussCut-set bound 2}) and (\ref{eq:GaussCut-set bound 8}) are less than or equal to the LHS of (\ref{eq:GaussCut-set bound 1}) and (\ref{eq:GaussCut-set bound 5}) respectively and thus less than their RHS. Now replace $h_{A_1R}$ with $\tilde{h}_{B_1R}$ and $h_{A_2R}$ with $\tilde{h}_{A_2R}$ to get the desired inequalities. A similar argument on the down-link cut-set region shows that we can make the down-link channel gains of each pair consistent (in ordering) with the transmission rate and this completes the proof.

\section{Decoding at the relay}~\label{sec:AppUplink}
We receive the following signal at the relay
\begin{align}
y_R & = h_{A_1R}\sqrt{\alpha_{A_1}^{(1)}}x_{A_1}^{(1)}+h_{A_1R}\sqrt{\alpha_{A_1}^{(2)}}x_{A_1}^{(2)}+h_{B_1R}x_{B_1} \nonumber\\
& +h_{A_2R}\sqrt{\alpha_{A_2}^{(1)}}x_{A_2}^{(1)}+h_{A_2R}\sqrt{\alpha_{A_2}^{(2)}}x_{A_2}^{(2)}+h_{B_2R}x_{B_2}+z_R \nonumber.
\end{align}

For the case considered here ($|h_{A_1R}|\geq  |h_{B_1R}|\geq |h_{A_2R}|\geq |h_{B_2R}|$), we have the following decoding order at the relay: $x_{A_1}^{(1)}\rightarrow t \rightarrow x_{A_2}^{(1)}\rightarrow f$.
It follows the decoding of the signals from pair $(A_1,B_1)$.

Decoding of $x_{A_1}^{(1)}$ can be done with low error probability as long as
\begin{align}
& R_{A_1}^{(1)}\leq \nonumber \\ & C\left(\frac{|h_{A_1R}|^2P\alpha_{A_1}^{(1)}}{2\alpha_{B_1}^{(2)}|h_{B_1R}|^2P+\alpha_{A_2}^{(1)}|h_{A_2R}|^2P+2\alpha_{B_2}^{(2)}|h_{B_2R}|^2P+1}\right)\nonumber
\end{align}
Once $x_{A_1}^{(1)}$ is decoded, it can be subtracted successfully from the received signal.
Thus, we have
\begin{align}
\tilde{y}_R  = & h_{B_1R}\sqrt{\alpha_{B_1}^{(2)}}\underbrace{\left(x_{A_1}^{(2)}+x_{B_1}^{(2)}\right)}_{t} +h_{A_2R}\sqrt{\alpha_{A_2}^{(1)}}x_{A_2}^{(1)} \nonumber \\ & +h_{B_2R}\sqrt{\alpha_{B_2}^{(2)}}\underbrace{\left(x_{A_2}^{(2)}+ x_{B_2}^{(2)}\right)}_{f}+z_R\nonumber
\end{align}
Next, the sum codeword $t$ of the lattice codes from $x_{A_1}^{(2)}$ and $x_{B_1}$ is decoded.
The decoding of $t$ can be done with low error probability as long as
\begin{align}
R_{A_1}^{(2)}, R_{B_1}&\leq \log\left(\frac{|h_{B_1R}|^2P\alpha_{B_1}^{(2)}}{\alpha_{A_2}^{(1)}|h_{A_2R}|^2P+2\alpha_{B_2}^{(2)} |h_{B_2R}|^2P+1}\right)^+\nonumber.
\end{align}
Once $t$ is decoded, it can be subtracted successfully from the received signal. Thus, we have
\begin{align}
\hat{y}_R= &   h_{A_2R}\sqrt{\alpha_{A_2}^{(1)}}x_{A_2}^{(1)} + h_{B_2R}\sqrt{\alpha_{B_2}^{(2)}}f+  z \nonumber.
\end{align}
It follows the decoding of the signals from pair $(A_2,B_2)$.
beginning with the decoding of the Gaussian $x_{A_2}^{(1)}$.
This can be done with low probability as long as
\begin{align}
R_{A_2}^{(1)}\leq & C \left(\frac{|h_{A_2R}|^2P\alpha_{A_2}^{(1)}}
                  {2|h_{B_2R}|^2P\alpha_{B_2}^{(2)}+1}\right).\nonumber
\end{align}
Once $x_{A_2}^{(1)}$ is decoded, it can be subtracted successfully from the received signal.
Thus, we have
\begin{align}
\hat{y}_R= & \sqrt{\alpha_{B_2}^{(2)}}h_{B_2R}f+  z\nonumber
\end{align}
As a final step, we want to decode the lattice point $f$.
This can be done with low probability as long as
\begin{align}
R_{B_2}\leq & \left(\log \left(\alpha_{B_2}^{(2)}|h_{B_2R}|^2P\right)\right)^+.\nonumber
\end{align}

\section{Proof of Lemma \ref{lem:AchRegionUL}}~\label{app:lemAchRegionUL}
The three cases we have to consider are given in sections~\ref{sec:FirstCaseUplink}~to~\ref{sec:ThirdCaseUplink}. In the following we provide the proof for each case separately.
\subsection{Case $|h_{A_1R}|\geq  |h_{B_1R}|\geq |h_{A_2R}|\geq |h_{B_2R}|$}
Consider a 4-tuple $(r_{A_1},r_{B_1},r_{A_2},r_{B_2})$ satisfying~\eqref{eq:RatesUplink1}-\eqref{eq:RatesUplink6}. Starting with~\eqref{eq:AchRatesUplink3}, we equate
\begin{align}\label{eq:PowConstraintAlphaB22}
\left(\log \left(\alpha_{B_2}^{(2)}|h_{B_2R}|^2P\right)\right)^+=r_{B_2} \Rightarrow \alpha_{B_2}^{(2)}= \frac{2^{r_{B_2}}}{|h_{B_2R}|^2P}.
\end{align}
Now from~\eqref{eq:RatesUplink2} we know that
\begin{align}
\alpha_{B_2}^{(2)}\leq  \frac{1+|h_{B_2R}|^2P}{2|h_{B_2R}|^2P}& \stackrel{|h_{B_2R}|^2P\geq 1}{\leq} 1,\nonumber
\end{align}
which shows that this is a valid choice of $\alpha_{B_2}^{(2)}$.
Next we equate  $r_{A_2}-r_{B_2}=\mbox{RHS of}$~\eqref{eq:AchRatesUplink3} and use ~\eqref{eq:PowConstraintAlphaB22}. We get
\begin{align}\label{eq:PowConstraintAlphaA21a}
\alpha_{A_2}^{(1)}=\frac{\left(2^{r_{A_2}-r_{B_2}}-1\right)\left(2^{r_{B_2}}+1\right)}{|h_{A_2R}|^2P}.
\end{align}
Using~\eqref{eq:HowtoChooseAlpha} and adding this to~\eqref{eq:PowConstraintAlphaA21a} we get
\begin{align}
\alpha_{A_2}^{(1)} +\alpha_{A_2}^{(2)}=& \frac{2\cdot 2^{r_{A_2}}+2^{r_{A_2}-r_{B_2}}-2^{r_{B_2}}-1}{|h_{A_2R}|^2P} \nonumber \\
& \leq \frac{3\cdot 2^{r_{A_2}}-2}{|h_{A_2R}|^2P} \stackrel{\eqref{eq:RatesUplink2}}{\leq} 1,\nonumber
\end{align}
verifying that this is a valid choice of $\alpha_{A_2}^{(1)}$, $\alpha_{A_2}^{(2)}$.
Then we equate $r_{B_1}=\mbox{RHS of}$~\eqref{eq:AchRatesUplink2}, by setting
\begin{align}\label{eq:PowConstraintAB1a}
\alpha_{B_1}^{(2)}= & \frac{2^{r_{B_1}}2^{r_{A_2}-r_{B_2}}\left(2\cdot 2^{r_{B_2}}+1\right)}{|h_{B_1R}|^2P} \\
& \leq \frac{ 3\cdot 2^{r_{B_1}+r_{A_2}}}{|h_{B_1R}|^2P}\stackrel{\eqref{eq:RatesUplink6}, |h_{B_1R}|^2P\geq \frac{3}{2}}{\leq} 1, \nonumber
\end{align}
verifying that this is a valid choice of $\alpha_{B_1}^{(2)}$.
Finally we equate  $r_{A_1}-r_{B_1}=\mbox{RHS of}$~\eqref{eq:AchRatesUplink1}, by setting
\begin{align}\label{eq:PowConstraintAA11a}
& \alpha_{A_1}^{(2)}= \left(2^{r_{A_1}-r_{B_1}}-1\right) \times \\ & \frac{\left(2^{r_{A_2}+r_{B_1}-r_{B_2}}\left(1+22^{r_{B_2}}\right)+2^{r_{A_2}-r_{B_2}}\left(1+2^{r_{B_2}}\right)+2^{r_{B_2}}\right)}{|h_{A_1R}|^2P}. \nonumber
\end{align}
Using~\eqref{eq:HowtoChooseAlpha} and~\eqref{eq:PowConstraintAB1a} and adding this to~\eqref{eq:PowConstraintAA11a} we  get
\begin{align}
 \alpha_{A_1}^{(1)} +\alpha_{A_1}^{(2)} &\leq \frac{5\cdot 2^{r_{A_1}+r_{A_2}}+2^{r_{A_1}+r_{B_2}}-3}{|h_{A_1R}|^2P}\stackrel{\eqref{eq:RatesUplink3}}{\leq} 1.\nonumber
\end{align}
which shows that this is a valid choice of $\alpha_{A_1}^{(1)}$, $\alpha_{A_1}^{(2)}$.

\subsection{Case $|h_{A_1R}|\geq |h_{A_2R}|\geq  |h_{B_1R}|\geq |h_{B_2R}|$}
Consider a 4-tuple $(r_{A_1},r_{B_1},r_{A_2},r_{B_2})$ satisfying~\eqref{eq:RatesUplink1}-\eqref{eq:RatesUplink6}. Starting with~\eqref{eq:AchRatesUplink3Case2}, we equate
\begin{align}\label{eq:PowConstraintAlphaB22UpCase2}
\left(\log \left(\alpha_{B_2}^{(2)}|h_{B_2R}|^2P\right)\right)^+=r_{B_2} \Rightarrow \alpha_{B_2}^{(2)}= \frac{2^{r_{B_2}}}{|h_{B_2R}|^2P}.
\end{align}
Now from~\eqref{eq:RatesUplink2} we know that
\begin{align}
\alpha_{B_2}^{(2)}\leq  \frac{1+|h_{B_2R}|^2P}{2|h_{B_2R}|^2P}& \stackrel{|h_{B_2R}|^2P\geq 1}{\leq} 1,\nonumber
\end{align}
which shows that this is a valid choice of $\alpha_{B_2}^{(2)}$.
Next we equate $r_{B_1}=\mbox{RHS of}$~\eqref{eq:AchRatesUplink2Case2}, by setting
\begin{equation}\label{eq:PowConstraintAB1aUpCase2}
\alpha_{B_1}^{(2)}=\frac{2^{r_{B_1}}\left(2\cdot 2^{r_{B_2}}+1\right)}{|h_{B_1R}|^2P} \stackrel{\eqref{eq:RatesUplink1},\eqref{eq:RatesUplink5}, |h_{B_1R}|^2P\geq 2}{\leq} 1,
\end{equation}
verifying that this is a valid choice of $\alpha_{B_1}^{(2)}$.
Then we equate  $r_{A_2}-r_{B_2}=\mbox{RHS of}$~\eqref{eq:AchRatesUplink1Case2MAC2}, by setting
\begin{align}\label{eq:PowConstraintAlphaA21aUpCase2}
\alpha_{A_2}^{(1)}=\frac{\left(2^{r_{A_2}-r_{B_2}}-1\right)\left(4\cdot 2^{r_{B_1}+r_{B_2}}+2\left(2^{r_{B_1}+1}+2^{r_{B_2}}\right)+1\right)}{|h_{A_2R}|^2P}.
\end{align}
Using~\eqref{eq:HowtoChooseAlpha}, $2^x+2^y\leq 2^{x+y}$ with $x,y\geq 1$, and~\eqref{eq:PowConstraintAlphaB22UpCase2} and adding this to~\eqref{eq:PowConstraintAlphaA21aUpCase2} we get
\begin{align}
\alpha_{A_2}^{(1)} +\alpha_{A_2}^{(2)}& \leq \frac{6\cdot 2^{r_{A_2}+r_{B_1}+r_{A_2}-8}}{|h_{A_2R}|^2P} \stackrel{\eqref{eq:RatesUplink3},\eqref{eq:RatesUplink6}}{\leq} 1,\nonumber
\end{align}
verifying that this is a valid choice of $\alpha_{A_2}^{(1)}$, $\alpha_{A_2}^{(2)}$.
Now we equate  $r_{A_1}-r_{B_1}=\mbox{RHS of}$~\eqref{eq:AchRatesUplink1Case2MAC1}, by setting
\begin{align}\label{eq:PowConstraintAA11aUpCase2}
\alpha_{A_1}^{(1)}= \frac{\left(2^{r_{A_1}-r_{B_1}}-1\right)\left(4\cdot 2^{r_{B_1}+r_{B_2}}+2\left(2^{r_{B_1}+1}+2^{r_{B_2}}\right)+1\right)}{|h_{A_1R}|^2P}.
\end{align}
Using~\eqref{eq:HowtoChooseAlpha}, $2^x+2^y\leq 2^{x+y}$ with $x,y\geq 1$, and~\eqref{eq:PowConstraintAB1aUpCase2} and adding this to~\eqref{eq:PowConstraintAA11aUpCase2} we  get
\begin{align}
\alpha_{A_1}^{(1)} +\alpha_{A_1}^{(2)}  = & \frac{2^{r_{B_1}}\left(2\cdot 2^{r_{B_2}}+1\right)}{|h_{A_1R}|^2P} + \left(2^{r_{A_1}-r_{B_1}}-1\right) \times \nonumber \\
& \frac{\left(4\cdot 2^{r_{B_1}+r_{B_2}}+2\left(2^{r_{B_1}+1}+2^{r_{B_2}}\right)+1\right)}{|h_{A_1R}|^2P} \nonumber\\
 \leq & \frac{6\cdot 2^{r_{A_1}+r_{B_2}}+2^{r_{A_1}}-6}{|h_{A_1R}|^2P}\stackrel{\eqref{eq:RatesUplink1},\eqref{eq:RatesUplink4}}{\leq} 1.\nonumber
\end{align}
which shows that this is a valid choice of $\alpha_{A_1}^{(1)}$, $\alpha_{A_1}^{(2)}$.

Finally we equate  $r_{A_1}-r_{B_1}+r_{A_2}-r_{B_2}=\mbox{RHS of}$~\eqref{eq:AchRatesUplink1Case2MACSum}, by setting
\begin{align}\label{eq:PowConstraintAA11aUpCase2Sum}
\alpha_{A_1}^{(1)}= & \left(2^{r_{A_1}-r_{B_1}+r_{A_2}-r_{B_2}}-2^{r_{A_2}-r_{B_2}}\right) \times \\
& \frac{\left(4\cdot 2^{r_{B_2}+r_{B_1}}+2\left(2^{r_{B_2}}+2^{r_{B_1}}\right)+1\right)}{|h_{A_1R}|^2P}. \nonumber
\end{align}
Using~\eqref{eq:HowtoChooseAlpha}, $2^x+2^y\leq 2^{x+y}$ with $x,y\geq 1$, and~\eqref{eq:PowConstraintAB1aUpCase2} and adding this to~\eqref{eq:PowConstraintAA11aUpCase2Sum} we  get
\begin{align}
 \alpha_{A_1}^{(1)} +\alpha_{A_1}^{(2)} = & \left(2^{r_{A_1}-r_{B_1}+r_{A_2}-r_{B_2}}-2^{r_{A_2}-r_{B_2}}\right) \times \nonumber \\
 & \frac{4\cdot 2^{r_{B_2}+r_{B_1}}+2\left(2^{r_{B_2}}+2^{r_{B_1}}\right)+1}{|h_{A_1R}|^2P}\nonumber\\
& +\frac{2^{r_{B_1}}\left(2\cdot 2^{r_{B_2}}+1\right)}{|h_{A_1R}|^2P}  \nonumber \\
 \stackrel{r_{A_2}\geq r_{B_2}}{\leq} & \frac{7\cdot 2^{r_{A_1}+r_{A_2}}-2}{|h_{A_1R}|^2P}\stackrel{\eqref{eq:RatesUplink3}}{\leq} 1.\nonumber
\end{align}
which shows that this is a valid choice of $\alpha_{A_1}^{(1)}$, $\alpha_{A_1}^{(2)}$.

\subsection{Case $|h_{A_1R}|\geq  |h_{A_2R}|\geq |h_{B_2R}|\geq |h_{B_1R}|$}
Consider a 4-tuple $(r_{A_1},r_{B_1},r_{A_2},r_{B_2})$ satisfying~\eqref{eq:RatesUplink1}-\eqref{eq:RatesUplink6}. Starting with~\eqref{eq:AchRatesUplink3Case3}, we equate
\begin{align}\label{eq:PowConstraintAlphaB22UpCase3}
\left(\log \left(\alpha_{B_1}^{(2)}|h_{B_1R}|^2P\right)\right)^+=r_{B_1} \Rightarrow \alpha_{B_1}^{(2)}= \frac{2^{r_{B_1}}}{|h_{B_1R}|^2P}.
\end{align}
Now from~\eqref{eq:RatesUplink1} we know that
\begin{align}
\alpha_{B_1}^{(2)}\leq  \frac{1+|h_{B_1R}|^2P}{2|h_{B_1R}|^2P}& \stackrel{|h_{B_1R}|^2P\geq 1}{\leq} 1,\nonumber
\end{align}
which shows that this is a valid choice of $\alpha_{B_1}^{(2)}$.
Next we equate $r_{B_2}=\mbox{RHS of}$~\eqref{eq:AchRatesUplink2Case3}, by setting
\begin{equation}\label{eq:PowConstraintAB1aUpCase3}
\alpha_{B_2}^{(2)}=\frac{2^{r_{B_2}}\left(2^{r_{B_1}}+1\right)}{|h_{B_2R}|^2P} \stackrel{\eqref{eq:RatesUplink2},\eqref{eq:RatesUplink5}, |h_{B_2R}|^2P\geq\frac{5}{2}}{\leq} 1,
\end{equation}
verifying that this is a valid choice of $\alpha_{B_2}^{(2)}$.
Then we equate  $r_{A_2}-r_{B_2}=\mbox{RHS of}$~\eqref{eq:AchRatesUplink1Case3MAC2}, by setting
\begin{align}\label{eq:PowConstraintAlphaA21aUpCase3}
\alpha_{A_2}^{(1)}=\frac{\left(2^{r_{A_2}-r_{B_2}}-1\right)\left(4\cdot 2^{r_{B_1}+r_{B_2}}+2\left(2^{r_{B_2}}+2^{r_{B_1}}\right)+1\right)}{|h_{A_2R}|^2P}.
\end{align}
Using~\eqref{eq:HowtoChooseAlpha}, $2^x+2^y\leq 2^{x+y}$ with $x,y\geq 1$, and~\eqref{eq:PowConstraintAB1aUpCase3} and adding this to~\eqref{eq:PowConstraintAlphaA21aUpCase3} we get
\begin{align}
\alpha_{A_2}^{(1)} +\alpha_{A_2}^{(2)}  = & \left(2^{r_{A_2}-r_{B_2}}-1\right) \times \nonumber \\
& \frac{\left(4\cdot 2^{r_{B_1}+r_{B_2}}+2\left(2^{r_{B_2}}+2^{r_{B_1}}\right)+1\right)}{|h_{A_2R}|^2P}+\nonumber\\
& \frac{\left(2\cdot 2^{r_{B_2}+r_{B_1}}+2^{r_{B_2}}\right)}{|h_{A_2R}|^2P}  \nonumber \\
\leq  & \frac{6\cdot 2^{r_{A_2}+r_{B_1}}+2^{r_{A_2}}-6}{|h_{A_2R}|^2P} \stackrel{\eqref{eq:RatesUplink2},\eqref{eq:RatesUplink6}}{\leq} 1,\nonumber
\end{align}
verifying that this is a valid choice of $\alpha_{A_2}^{(1)}$, $\alpha_{A_2}^{(2)}$.
Now we equate  $r_{A_1}-r_{B_1}=\mbox{RHS of}$~\eqref{eq:AchRatesUplink1Case3MAC1}, by setting
\begin{align}\label{eq:PowConstraintAA11aUpCase3}
\alpha_{A_1}^{(1)}= \frac{\left(2^{r_{A_1}-r_{B_1}}-1\right)\left(4\cdot 2^{r_{B_2}+r_{B_1}}+2\left(2^{r_{B_1}}+2^{r_{B_2}}\right)+1\right)}{|h_{A_1R}|^2P}.
\end{align}
Using~\eqref{eq:HowtoChooseAlpha}, $2^x+2^y\leq 2^{x+y}$ with $x,y\geq 1$, and~\eqref{eq:PowConstraintAlphaB22UpCase3} and adding this to~\eqref{eq:PowConstraintAA11aUpCase3} we  get
\begin{align}
 \alpha_{A_1}^{(1)} +\alpha_{A_1}^{(2)}   \leq & \left(2^{r_{A_1}-r_{B_1}}-1\right) \times \nonumber \\
 & \frac{\left(4\cdot2^{r_{B_2}+r_{B_1}}+2\left(2^{r_{B_1}}+2^{r_{B_2}}\right)+1\right)}{|h_{A_1R}|^2P} \nonumber \\
 & +\frac{2^{r_{B_1}}}{|h_{A_1R}|^2P} \nonumber\\
 \leq &  \frac{6\cdot2^{r_{A_1}+r_{B_2}}+2^{r_{A_1}}-8}{|h_{A_1R}|^2P}\stackrel{\eqref{eq:RatesUplink1},\eqref{eq:RatesUplink4}}{\leq} 1.\nonumber
\end{align}
which shows that this is a valid choice of $\alpha_{A_1}^{(1)}$, $\alpha_{A_1}^{(2)}$.

Finally we equate  $r_{A_1}-r_{B_1}+r_{A_2}-r_{B_2}=\mbox{RHS of}$~\eqref{eq:AchRatesUplink1Case3MACSum}, by setting
\begin{align}\label{eq:PowConstraintAA11aUpCase3Sum}
\alpha_{A_1}^{(1)}= & \left(2^{r_{A_1}-r_{B_1}+r_{A_2}-r_{B_2}}-2^{r_{A_2}-r_{B_2}}\right) \times \\
& \frac{\left(4\cdot2^{r_{B_2}+r_{B_1}}+2\left(2^{r_{B_1}}+2^{r_{B_2}}\right)+1\right)}{|h_{A_1R}|^2P}. \nonumber
\end{align}
Using~\eqref{eq:HowtoChooseAlpha}, $2^x+2^y\leq 2^{x+y}$ with $x,y\geq 1$, and~\eqref{eq:PowConstraintAB1aUpCase3} and adding this to~\eqref{eq:PowConstraintAA11aUpCase3Sum} we  get
\begin{align}
 \alpha_{A_1}^{(1)} +\alpha_{A_1}^{(2)} =& \left(2^{r_{A_1}-r_{B_1}+r_{A_2}-r_{B_2}}-2^{r_{A_2}-r_{B_2}}\right) \nonumber \\
 &
 \frac{\left(4\cdot 2^{r_{B_2}+r_{B_1}}+2\left(2^{r_{B_1}}+2^{r_{B_2}}\right)+1\right)}{|h_{A_1R}|^2P} \nonumber \\
 & + \frac{2^{r_{B_1}}}{|h_{A_1R}|^2P} \nonumber\\
  \leq & \frac{7\cdot 2^{r_{A_1}+r_{A_2}}-2}{|h_{A_1R}|^2P}\stackrel{\eqref{eq:RatesUplink3}}{\leq} 1.\nonumber
\end{align}
which shows that this is a valid choice of $\alpha_{A_1}^{(1)}$, $\alpha_{A_1}^{(2)}$.

\section{Decoding at the nodes}~\label{sec:AppDownlink}
With $R_R^{(1)}=R_{A_1}^{(1)}$, $R_R^{(2)}=R_{A_1}^{(2)}=R_{B_1}$, $R_R^{(3)}=R_{A_2}^{(1)}$, $R_R^{(4)}=R_{A_2}^{(2)}=R_{B_2}$, we describe the decoding strategies at the nodes and the achievable rates for the case $|h_{RB_1}|\geq  |h_{RA_1}|\geq |h_{RB_2}|\geq |h_{RA_2}|$.
\subsection{Decoding at node $B_1$}
The node $B_1$ first decodes $x_R^{(4)}$ (corresponds to $f$ from the uplink) by treating $x_R^{(1)}$ to  $x_R^{(3)}$ as noise.
This can be done with low probability of error as long as
\begin{align}
R_R^{(4)}\leq C\left(\frac{|h_{RB_1}|^2P\alpha_R^{(4)}}{1+|h_{B_1R}|^2P\left(\sum_{j=1}^3\alpha_R^{(j)}\right)}\right)\nonumber
\end{align}
Once decoded, the signal $x_R^{(4)}$ is canceled from the received signal and $x_R^{(3)}$ (corresponds to $x_{A_2}^{(1)}$ from the uplink) is decoded by treating $x_R^{(1)}$ and $x_R^{(2)}$ as noise.
This can be done successfully with low probability of error as long as
\begin{align}
R_R^{(3)}\leq C\left(\frac{|h_{RB_1}|^2P\alpha_R^{(3)}}{1+|h_{B_1R}|^2P\left(\alpha_R^{(1)}+\alpha_R^{(2)}\right)}\right)\nonumber
\end{align}
Once decoded, the signal $x_R^{(3)}$ is canceled from the received signal and $x_R^{(2)}$ (corresponds to $t$ from the uplink) is decoded by treating $x_R^{(1)}$ as noise.
This can be done successfully with low probability of error as long as
\begin{align}\label{eq:DecXR2atB1}
R_R^{(2)}\leq C\left(\frac{|h_{RB_1}|^2P\alpha_R^{(2)}}{1+|h_{RB_1}|^2P\alpha_R^{(1)}}\right)
\end{align}
Once decoded, $x_R^{(2)}$ is canceled from the received signal.
Finally, $x_R^{(1)}$ (corresponds to $x_{A_1}^{(1)}$ from the uplink) is decoded free of interference. This can be done with low probability of error as long as
\begin{align}
R_R^{(1)}\leq C\left(|h_{RB_1}|^2P\alpha_R^{(1)}\right).\nonumber
\end{align}
\subsection{Decoding at node $A_1$}
The node $A_1$ proceeds similarly with the exception that $x_R^{(1)}$ is known already and can be canceled from the received signal. After having decoded $x_R^{(3)}$ and $x_R^{(4)}$, $x_R^{(2)}$ is decoded free of interference. This can be done with low probability of error as long as
\begin{align}\label{eq:DecXR2atA1}
R_R^{(2)}\leq C\left(|h_{RA_1}|^2P\alpha_R^{(2)}\right).
\end{align}

\subsection{Decoding at node $B_2$}
The receivers of the second pair have the same order of detection.
Thus, the node $B_2$ can decode $R_R^{(4)}$ with low probability of error as long as
\begin{align}\label{eq:DecXR4atB2}
R_R^{(4)}\leq C\left(\frac{|h_{RB_2}|^2P\alpha_R^{(4)}}{1+|h_{RB_2}|^2P\left(\sum_{j=1}^3\alpha_R^{(j)}\right)}\right).
\end{align}
Once decoded, the signal $x_R^{(4)}$ is canceled from the received signal and $x_R^{(3)}$ is decoded by treating $x_R^{(1)}$ and $x_R^{(2)}$ as noise.
This can be done successfully with low probability of error as long as
\begin{align}
R_R^{(3)}\leq C\left(\frac{|h_{RB_2}|^2P\alpha_R^{(3)}}{1+|h_{RB_2}|^2P\left(\alpha_R^{(1)}+\alpha_R^{(2)}\right)}\right)\nonumber
\end{align}

\subsection{Decoding at node $A_2$}
Assuming that the node $A_2$ knows the strategy of the relay and the codebook it has used, it can reconstruct $x_R^{(3)}$ perfectly, since it contains only its own message. Thus, it cancels the effect of $x_R^{(3)}$ from the received signal. As a next and final step, it decodes $x_R^{(4)}$.
This can be done with low probability of error as long as
\begin{align}\label{eq:DecXR4atA2}
R_R^{(4)}\leq C\left(\frac{|h_{RA_2}|^2P\alpha_R^{(4)}}{1+|h_{RA_2}|^2P\left(\alpha_R^{(1)}+\alpha_R^{(2)}\right)}\right)
\end{align}
Thus, in summary we have
\begin{align}
R_R^{(4)}\leq \min\left(\text{RHS  of~\eqref{eq:DecXR4atB2}},\text{RHS  of~\eqref{eq:DecXR4atA2}}\right)\nonumber
\end{align}
and
\begin{align}
R_R^{(2)}\leq \min\left(\text{RHS  of~\eqref{eq:DecXR2atA1}},\text{RHS  of~\eqref{eq:DecXR2atB1}}\right)\nonumber
\end{align}

\section{Proof of Lemma~\ref{lem:AchRegionDL}}~\label{app:lemAchRegionDL}
The three cases we have to consider are given in sections~\ref{sec:FirstCaseDownlink}-\ref{sec:ThirdCaseDownlink}. In the following we provide the proof for each case separately.
\subsection{Case $|h_{RB_1}|\geq  |h_{RA_1}|\geq |h_{RB_2}|\geq |h_{RA_2}|$}
Consider a 4-tuple $(r_{A_1},r_{B_1},r_{A_2},r_{B_2})$ satisfying~\eqref{eq:RatesDownlink1}-\eqref{eq:RatesDownlink6}. Starting with the first equation in~\eqref{eq:AchRatesDownlink3}, we equate
\begin{align}\label{eq:PowConstraintAlphaR1}
& \log \left(1+\alpha_{R}^{(1)}|h_{RB_1}|^2P\right)=r_{A_1}-r_{B_1}  \\
& \Rightarrow \alpha_{R}^{(1)}= \frac{2^{r_{A_1}-r_{B_1}}-1}{|h_{RB_1}|^2P}. \nonumber
\end{align}
Now from~\eqref{eq:RatesDownlink1} we know that
\begin{align}
\alpha_{R}^{(1)}\leq  \frac{\frac{1+|h_{RB_1}|^2P}{4}-1}{|h_{RB_1}|^2P}& \leq 1,\nonumber
\end{align}
which shows that this is a valid choice of $\alpha_{R}^{(1)}$.

From~\eqref{eq:AchRatesDownlink1}, we have
\begin{align}~\label{eq:AchRatesDownlink1UpBound}
R_{A_1}^{(2)},R_{B_1}\leq  C\left(\frac{|h_{RB_1}|^2P\alpha_R^{(2)}}{1+|h_{RB_1}|^2P\alpha_R^{(1)}}\right).
\end{align}
Next we equate  $r_{B_1}=\mbox{RHS of}$~\eqref{eq:AchRatesDownlink1UpBound}, by setting
\begin{align}\label{eq:PowConstraintAlphaR2a}
\alpha_{R}^{(2)}=\frac{\left(2^{r_{B_1}}-1\right)\left(2^{r_{A_1}-r_{B_1}}\right)}{|h_{RB_1}|^2P}.
\end{align}
Using~\eqref{eq:RatesDownlink1} and~\eqref{eq:PowConstraintAlphaR1} and adding this to~\eqref{eq:PowConstraintAlphaR2a} we get
\begin{align}\label{eq:VerifyAlphaR2a}
\alpha_{R}^{(1)} + \alpha_{R}^{(2)}\leq 2\frac{\frac{1+|h_{RB_1}|^2P}{4}-1}{|h_{RB_1}|^2P} \leq 1,
\end{align}
verifying that this is a valid choice of $\alpha_{R}^{(1)}$, $\alpha_{R}^{(2)}$.
Then we equate $r_{A_2}-r_{B_2}=\mbox{RHS of}$~\eqref{eq:AchRatesDownlink3} (second equation), by setting
\begin{equation}\label{eq:PowConstraintAlphaR3}
\alpha_{R}^{(3)}=\frac{\left(2^{r_{A_2}-r_{B_2}}-1\right)\left(1+\frac{|h_{RB_2}|^2P}{|h_{RB_1}|^2P}\left(2^{r_{B_2}}-1\right)\right)}{|h_{RB_2}|^2P} .
\end{equation}
Using~\eqref{eq:RatesDownlink2},~\eqref{eq:RatesDownlink3} and~\eqref{eq:VerifyAlphaR2a} and adding this to~\eqref{eq:PowConstraintAlphaR3} we get
\begin{align}
\sum_{j=1}^3\alpha_{R}^{(j)} \leq \frac{\frac{1+|h_{RB_2}|^2P}{4}-1}{|h_{RB_2}|^2P}+ 3\frac{\frac{1+|h_{RB_1}|^2P}{4}-1}{|h_{RB_1}|^2P} \leq 1,\nonumber
\end{align}
verifying that this is a valid choice of $\alpha_{R}^{(j)}$, $j=1\dots3$.
Finally from~\eqref{eq:AchRatesDownlink2}, we have
\begin{align}~\label{eq:AchRatesDownlink2UpBound}
R_{A_2}^{(2)},R_{B_2}\leq  C\left(\frac{|h_{RA_2}|^2P\alpha_R^{(4)}}{1+|h_{RA_2}|^2P\left(\alpha_R^{(1)}+\alpha_R^{(2)}\right)}\right).
\end{align}
Thus, we equate  $r_{B_2}=\mbox{RHS of}$~\eqref{eq:AchRatesDownlink2UpBound}, by setting
\begin{align}\label{eq:PowConstraintAlphaR4}
\alpha_{R}^{(4)}& = \frac{\left(2^{r_{B_2}}-1\right)}{|h_{RA_2}|^2P}\Big(1+\frac{|h_{RA_2}|^2P}{|h_{RB_1}|^2P}\left(2^{r_{A_1}}-1\right)\Big).
\end{align}
Using~\eqref{eq:RatesDownlink2}, ~\eqref{eq:RatesDownlink4},~\eqref{eq:PowConstraintAlphaR1},~\eqref{eq:PowConstraintAlphaR2a}, ~\eqref{eq:PowConstraintAlphaR3} and adding this to~\eqref{eq:PowConstraintAlphaR4} we get
\begin{align}
\sum_{j=1}^4\alpha_{R}^{(j)} \leq & \frac{\frac{1+|h_{RA_2}|^2P}{4}-1}{|h_{RA_2}|^2P} +\frac{\frac{1+|h_{RB_1}|^2P}{4}-1}{|h_{RB_1}|^2P} \nonumber \\
& +\frac{\frac{1+|h_{RB_2}|^2P}{4}-1}{|h_{RB_2}|^2P} +\frac{\frac{1+|h_{RB_2}|^2P}{4}-1}{|h_{RB_1}|^2P} \leq 1 \nonumber
\end{align}
which shows that this is a valid choice of $\alpha_{R}^{(j)}$, $j=1\dots4$.

\subsection{Case $|h_{RB_1}|\geq |h_{RB_2}|\geq  |h_{RA_1}|\geq |h_{RA_2}|$}

Consider a 4-tuple $(r_{A_1},r_{B_1},r_{A_2},r_{B_2})$ satisfying~\eqref{eq:RatesDownlink1}-\eqref{eq:RatesDownlink6}. Starting with the first equation in~\eqref{eq:AchRatesDownlink3Case2}, we equate
\begin{align}\label{eq:PowConstraintAlphaR1DoCase2}
& \log \left(1+\alpha_{R}^{(1)}|h_{RB_1}|^2P\right)=r_{A_1}-r_{B_1} \\
& \Rightarrow \alpha_{R}^{(1)}= \frac{2^{r_{A_1}-r_{B_1}}-1}{|h_{RB_1}|^2P}. \nonumber
\end{align}
Now from~\eqref{eq:RatesDownlink1} we know that
\begin{align}
\alpha_{R}^{(1)}\leq  \frac{\frac{1+|h_{RB_1}|^2P}{4}-1}{|h_{RB_1}|^2P}& \leq 1,\nonumber
\end{align}
which shows that this is a valid choice of $\alpha_{R}^{(1)}$.

Next we equate $r_{A_2}-r_{B_2}=\mbox{RHS of}$~\eqref{eq:AchRatesDownlink3Case2} (second equation), by setting
\begin{equation}\label{eq:PowConstraintAlphaR3DoCase2}
\alpha_{R}^{(3)}=\frac{\left(2^{r_{A_2}-r_{B_2}}-1\right)\left(1+\frac{|h_{RB_2}|^2P}{|h_{RB_1}|^2P}\left(2^{r_{A_1}-r_{B_1}}-1\right)\right)}{|h_{RB_2}|^2P} .
\end{equation}
Using~\eqref{eq:RatesDownlink2},~\eqref{eq:RatesDownlink3} and~\eqref{eq:VerifyAlphaR2a} and adding this to~\eqref{eq:PowConstraintAlphaR3DoCase2} we get
\begin{align}
\alpha_{R}^{(1)}+\alpha_{R}^{(3)} \leq \frac{\frac{1+|h_{RB_2}|^2P}{4}-1}{|h_{RB_2}|^2P}+ \frac{\frac{1+|h_{RB_1}|^2P}{4}-1}{|h_{RB_1}|^2P} \leq 1,\nonumber
\end{align}
verifying that this is a valid choice of $\alpha_{R}^{(1)}$, $\alpha_{R}^{(3)}$.

From~\eqref{eq:AchRatesDownlink1Case2}, we have
\begin{align}~\label{eq:AchRatesDownlink1UpBoundCase2}
R_{A_1}^{(2)}, R_{B_1}\leq  C\left(\frac{|h_{RA_1}|^2P\alpha_R^{(2)}}{1+|h_{RA_1}|^2P\alpha_R^{(3)}}\right).
\end{align}
Next we equate  $r_{B_1}=\mbox{RHS of}$~\eqref{eq:AchRatesDownlink1UpBoundCase2}, by setting
\begin{align}\label{eq:PowConstraintAlphaR2aDoCase2}
\alpha_{R}^{(2)}=\frac{\left(2^{r_{B_1}}-1\right)\left(1+|h_{RA_1}|^2P\alpha_{R}^{(3)}\right)}{|h_{RA_1}|^2P}.
\end{align}
Using~\eqref{eq:RatesDownlink1} and~\eqref{eq:PowConstraintAlphaR1DoCase2} and adding this to~\eqref{eq:PowConstraintAlphaR2aDoCase2} we get
\begin{align}\label{eq:VerifyAlphaR2aDoCase2}
\sum_{j=1}^3\alpha_{R}^{(j)}\leq & \frac{2^{r_{B_1}}-1}{|h_{RA_1}|^2P} +\frac{2^{r_{B_1}+r_{A_2}}-1}{|h_{RB_2}|^2P} +\frac{2^{r_{A_1}+r_{A_2}}-1}{|h_{RB_1}|^2P} \\
 & \leq 1, \nonumber
\end{align}
verifying that this is a valid choice of $\alpha_{R}^{(j)}$, $j=1\dots3$.
Finally from~\eqref{eq:AchRatesDownlink2}, we have
\begin{align}~\label{eq:AchRatesDownlink2UpBoundCase2}
R_{A_2}^{(2)},R_{B_2}\leq  C\left(\frac{|h_{RA_1}|^2P\alpha_R^{(4)}}{1+|h_{RA_1}|^2P\left(\alpha_R^{(3)}+\alpha_R^{(2)}\right)}\right).
\end{align}
Thus, we equate  $r_{B_2}=\mbox{RHS of}$~\eqref{eq:AchRatesDownlink2UpBoundCase2}, by setting
\begin{align}\label{eq:PowConstraintAlphaR4DoCase2}
\alpha_{R}^{(4)}& = \frac{\left(2^{r_{B_2}}-1\right)}{|h_{RA_1}|^2P}\Big(2^{r_{A_1}}+ \left(2^{r_{A_1}}-1\right)|h_{RA_1}|^2P\alpha_{R}^{(3)}\Big).
\end{align}
Using~\eqref{eq:RatesDownlink2}, ~\eqref{eq:RatesDownlink4},~\eqref{eq:PowConstraintAlphaR1DoCase2},~\eqref{eq:PowConstraintAlphaR2aDoCase2}, ~\eqref{eq:PowConstraintAlphaR3DoCase2} and adding this to~\eqref{eq:PowConstraintAlphaR4DoCase2} we get
\begin{align}
\sum_{j=1}^4\alpha_{R}^{(j)} \leq & \frac{\frac{1+|h_{RA_1}|^2P}{4}-1}{|h_{RA_1}|^2P} +\frac{\frac{1+|h_{RB_1}|^2P}{4}-1}{|h_{RB_1}|^2P}  \nonumber \\
& +\frac{\frac{1+|h_{RB_2}|^2P}{4}-1}{|h_{RB_2}|^2P} +\frac{\frac{1+|h_{RB_2}|^2P}{4}-1}{|h_{RB_1}|^2P} \leq 1\nonumber
\end{align}
which shows that this is a valid choice of $\alpha_{R}^{(j)}$, $j=1\dots4$.

\subsection{Case $|h_{RB_1}|\geq |h_{RB_2}|\geq |h_{RA_2}|\geq  |h_{RA_1}|$}

Consider a 4-tuple $(r_{A_1},r_{B_1},r_{A_2},r_{B_2})$ satisfying~\eqref{eq:RatesDownlink1}-\eqref{eq:RatesDownlink6}. Starting with the first equation in~\eqref{eq:AchRatesDownlink3Case2}, we equate
\begin{align}\label{eq:PowConstraintAlphaR1DoCase3}
& \log \left(1+\alpha_{R}^{(1)}|h_{RB_1}|^2P\right)=r_{A_1}-r_{B_1} \\
& \Rightarrow \alpha_{R}^{(1)}= \frac{2^{r_{A_1}-r_{B_1}}-1}{|h_{RB_1}|^2P}. \nonumber
\end{align}
Now from~\eqref{eq:RatesDownlink1} we know that
\begin{align}
\alpha_{R}^{(1)}\leq  \frac{\frac{1+|h_{RB_1}|^2P}{4}-1}{|h_{RB_1}|^2P}& \leq 1,\nonumber
\end{align}
which shows that this is a valid choice of $\alpha_{R}^{(1)}$.

Next we equate $r_{A_2}-r_{B_2}=\mbox{RHS of}$~\eqref{eq:AchRatesDownlink3Case3} (second equation), by setting
\begin{equation}\label{eq:PowConstraintAlphaR3DoCase3}
\alpha_{R}^{(3)}=\frac{\left(2^{r_{A_2}-r_{B_2}}-1\right)\left(1+\frac{|h_{RB_2}|^2P}{|h_{RB_1}|^2P}\left(2^{r_{A_1}-r_{B_1}}-1\right)\right)}{|h_{RB_2}|^2P} .
\end{equation}
Using~\eqref{eq:RatesDownlink2},~\eqref{eq:RatesDownlink3} and~\eqref{eq:VerifyAlphaR2a} and adding this to~\eqref{eq:PowConstraintAlphaR3DoCase3} we get
\begin{align}
\alpha_{R}^{(1)}+\alpha_{R}^{(3)} \leq \frac{\frac{1+|h_{RB_2}|^2P}{4}-1}{|h_{RB_2}|^2P}+ \frac{\frac{1+|h_{RB_1}|^2P}{4}-1}{|h_{RB_1}|^2P} \leq 1,\nonumber
\end{align}
verifying that this is a valid choice of $\alpha_{R}^{(1)}$, $\alpha_{R}^{(3)}$.

From~\eqref{eq:AchRatesDownlink1Case3}, we have
\begin{align}~\label{eq:AchRatesDownlink1UpBoundCase3}
R_{A_2}^{(2)}, R_{B_2}\leq  C\left(\frac{|h_{RA_2}|^2P\alpha_R^{(4)}}{1+|h_{RA_2}|^2P\alpha_R^{(1)}}\right).
\end{align}
Next we equate  $r_{B_2}=\mbox{RHS of}$~\eqref{eq:AchRatesDownlink1UpBoundCase3}, by setting
\begin{align}\label{eq:PowConstraintAlphaR2aDoCase3}
\alpha_{R}^{(4)}=\frac{\left(2^{r_{B_2}}-1\right)\left(1+|h_{RA_2}|^2P\alpha_{R}^{(1)}\right)}{|h_{RA_2}|^2P}.
\end{align}
Using~\eqref{eq:RatesDownlink1} and~\eqref{eq:PowConstraintAlphaR1DoCase3} and adding this to~\eqref{eq:PowConstraintAlphaR2aDoCase3} we get
\begin{align}\label{eq:VerifyAlphaR2aDoCase3}
\alpha_{R}^{(1)}+\alpha_{R}^{(3)}+\alpha_{R}^{(4)} \leq  & \frac{2^{r_{B_2}}-1}{|h_{RA_2}|^2P} +\frac{2^{r_{A_2}-r_{B_2}}-1}{|h_{RB_2}|^2P} \\
& +\frac{2^{r_{A_1}+r_{B_1}}-1}{|h_{RB_1}|^2P} +\frac{2^{r_{A_1}+r_{B_2}}-1}{|h_{RB_1}|^2P}
 \leq 1, \nonumber
\end{align}
verifying that this is a valid choice of $\alpha_{R}^{(1)}$, $\alpha_{R}^{(3)}$, and $\alpha_{R}^{(4)}$.
Finally from~\eqref{eq:AchRatesDownlink2}, we have
\begin{align}~\label{eq:AchRatesDownlink2UpBoundCase3}
R_{A_1}^{(2)},R_{B_1}\leq  C\left(\frac{|h_{RA_1}|^2P\alpha_R^{(2)}}{1+|h_{RA_1}|^2P\left(\alpha_R^{(3)}+\alpha_R^{(4)}\right)}\right).
\end{align}
Thus, we equate  $r_{B_2}=\mbox{RHS of}$~\eqref{eq:AchRatesDownlink2UpBoundCase3}, by setting
\begin{align}\label{eq:PowConstraintAlphaR4DoCase3}
\alpha_{R}^{(2)}& = \frac{\left(2^{r_{B_1}}-1\right)}{|h_{RA_1}|^2P}\Big(1+ \left(\alpha_{R}^{(3)}+\alpha_{R}^{(4)}\right)|h_{RA_1}|^2P\Big).
\end{align}
Using~\eqref{eq:RatesDownlink2}, ~\eqref{eq:RatesDownlink4},~\eqref{eq:PowConstraintAlphaR1DoCase3},~\eqref{eq:PowConstraintAlphaR2aDoCase3}, ~\eqref{eq:PowConstraintAlphaR3DoCase3} and adding this to~\eqref{eq:PowConstraintAlphaR4DoCase3} we get
\begin{align}
\sum_{j=1}^4\alpha_{R}^{(j)} \leq &  \frac{2^{r_{B_1}}-1}{|h_{RA_1}|^2P} +\frac{2^{r_{B_1}+r_{B_2}}-1}{|h_{RA_2}|^2P} +\frac{2^{r_{B_1}+r_{A_2}}-1}{|h_{RB_2}|^2P} \nonumber \\
& +\frac{2^{r_{A_1}+r_{B_2}}-1}{|h_{RB_1}|^2P}+\frac{2^{r_{A_1}+r_{A_2}}-1}{|h_{RB_1}|^2P}
 \leq 1,\nonumber
\end{align}
which shows that this is a valid choice of $\alpha_{R}^{(j)}$, $j=1,\dots,4$.

\begin{IEEEbiographynophoto}{Aydin Sezgin}
(S'01 - M'05) received the Dipl.-Ing. (M.S.) degree in communications engineering and the Dr.-Ing. (Ph.D.) degree in electrical engineering from the TFH Berlin in 2000 and the  TU  Berlin,  in  2005, respectively.
From 2001 to 2006, he was with the Heinrich-Hertz-Institut (HHI), Berlin.
From 2006 to 2008, he was a Post-doc and Lecturer at the Information Systems Laboratory, Department of Electrical Engineering, Stanford University.
From 2008 to 2009, he was a Post-doc at the Department of Electrical Engineering and Computer Science at the University of California Irvine.
From 2009 2011, he was the Head of the Emmy-Noether-Research Group on Wireless Networks at the Ulm University.
In 2011, he was full professor at the Department of Electrical Engineering and Information Technology at TU Darmstadt, Germany.
He is currently a full professor at the Department of Electrical Engineering and Information Technology at Ruhr-University Bochum, Germany.
His current research interests are in the area of information theory, communication theory, and signal processing with focus on applications to wireless communication systems.
He is currently serving as Editor for IEEE Transactions on Wireless Communications and Area Editor for Elsevier Journal of Electronics and Communications.
\end{IEEEbiographynophoto}

\begin{IEEEbiographynophoto}{A. Salman Avestimehr}
(S'05 - M'08) received the B.S. degree in electrical engineering from Sharif University of Technology, Tehran, Iran, in 2003 and the M.S. degree and Ph.D. degree in electrical engineering and computer science, both from the University of California, Berkeley, in 2005 and 2008, respectively.

He is currently an Assistant Professor at the School of Electrical and Computer Engineering at Cornell University, Ithaca, NY. He was also a postdoctoral scholar at the Center for the Mathematics of Informa- tion (CMI) at the California Institute of Technology, Pasadena, in 2008. His research interests include information theory, communications, and networking.

Dr. Avestimehr has received a number of awards, including the Presidential Early Career Award for Scientists and Engineers (PECASE) in 2011, the Young Investigator Program (YIP) award from the U. S. Air Force Office of Scientific Research (2011), the National Science Foundation CAREER award (2010), the David J. Sakrison Memorial Prize from the U.C. Berkeley EECS Department (2008), and the Vodafone U.S. Foundation Fellows Initiative Research Merit Award (2005). He has been a Guest Editor for the IEEE Transactions on Information Theory Special Issue on Interference Networks.
\end{IEEEbiographynophoto}

\begin{IEEEbiographynophoto}{Amin Khajehnejad}
received the Bachelor's degree in electrical engineering from the University of Tehran, Iran, in 2007 and the Master's degree in electrical engineering from the California Institute of Technology (Caltech), Pasadena, in 2009. He is currently working towards the Ph.D. degree at Caltech.

His general interests are in signal processing, coding, information theory and optimization. He has been with Lyric semiconductors Inc., and NEC laboratories America, Inc., during summers 2009 and 2010, respectively.
\end{IEEEbiographynophoto}

\begin{IEEEbiographynophoto}{Babak Hassibi}
was born in Tehran, Iran, in 1967. He received the B.S. degree from the University of Tehran, Iran, in 1989 and the M.S. and Ph.D. degrees from Stanford University, Stanford, CA, in 1993 and 1996, respectively, all in electrical engineering.

From October 1996 to October 1998, he was a Research Associate at the Information Systems Laboratory, Stanford University, and from November 1998 to December 2000, he was a Member of the Technical Staff in the Mathematical Sciences Research Center at Bell Laboratories, Murray Hill, NJ. He has also held short-term appointments at Ricoh California Research Center, the Indian Institute of Science, and Linkoping University, Sweden. Since January 2001, he has been with the California Institute of Technology, Pasadena, where he is currently Professor and Executive Officer of Electrical Engineering. His research interests include wireless communications and networks, robust estimation and control, adaptive signal processing, and linear algebra. He is the coauthor (with A. H. Sayed and T. Kailath) of the books Indefinite Quadratic Estimation and Control: A Unified Approach to $H^2$ and $H^{\infty}$ Theories (New York: SIAM, 1999) and Linear Estimation (Englewood Cliffs, NJ: Prentice-Hall, 2000).

Dr. Hassibi is a recipient of an Alborz Foundation Fellowship, the 1999 O. Hugo Schuck best paper award of the American Automatic Control Council (with H. Hindi and S. P. Boyd), the 2002 National Science Foundation Career Award, the 2002 Okawa Foundation Research Grant for Information and Telecommunications, the 2003 David and Lucille Packard Fellowship for Science and Engineering, and the 2003 Presidential Early Career Award for Sci- entists and Engineers (PECASE), and was a participant in the 2004 National Academy of Engineering ``Frontiers in Engineering'' program. He has been a Guest Editor for theIEEE Transactions on Information Theory Special Issue on Space-time Transmission, Reception, Coding and Signal Processing, was an Associate Editor for Communications of the IEEE Transactions on Information Theory from 2004 to 2006, and is currently an Editor for the journal Foundations and Trends in Information and Communication.
\end{IEEEbiographynophoto}

\end{document}

%% file: fig1.pstex_t
\begin{picture}(0,0)%
\epsfig{file=fig1.pstex}%
\end{picture}%
\setlength{\unitlength}{3947sp}%
\begingroup\makeatletter\ifx\SetFigFont\undefined%
\gdef\SetFigFont#1#2#3#4#5{%
  \reset@font\fontsize{#1}{#2pt}%
  \fontfamily{#3}\fontseries{#4}\fontshape{#5}%
  \selectfont}%
\fi\endgroup%
\begin{picture}(5756,2274)(1189,-3148)
\put(6376,-2986){\makebox(0,0)[lb]{\smash{{\SetFigFont{12}{14.4}{\rmdefault}{\mddefault}{\updefault}{\color[rgb]{0,0,0}$B_M$}%
}}}}
\put(3901,-2086){\makebox(0,0)[lb]{\smash{{\SetFigFont{12}{14.4}{\rmdefault}{\mddefault}{\updefault}{\color[rgb]{0,0,0}$R$}%
}}}}
\put(2551,-2686){\makebox(0,0)[lb]{\smash{{\SetFigFont{12}{14.4}{\rmdefault}{\mddefault}{\updefault}{\color[rgb]{0,0,0}$n_{A_MR}$}%
}}}}
\put(5101,-2686){\makebox(0,0)[lb]{\smash{{\SetFigFont{12}{14.4}{\rmdefault}{\mddefault}{\updefault}{\color[rgb]{0,0,0}$n_{B_MR}$}%
}}}}
\put(1351,-1186){\makebox(0,0)[lb]{\smash{{\SetFigFont{12}{14.4}{\rmdefault}{\mddefault}{\updefault}{\color[rgb]{0,0,0}$A_1$}%
}}}}
\put(6376,-1186){\makebox(0,0)[lb]{\smash{{\SetFigFont{12}{14.4}{\rmdefault}{\mddefault}{\updefault}{\color[rgb]{0,0,0}$B_1$}%
}}}}
\put(2551,-1411){\makebox(0,0)[lb]{\smash{{\SetFigFont{12}{14.4}{\rmdefault}{\mddefault}{\updefault}{\color[rgb]{0,0,0}$n_{A_1R}$}%
}}}}
\put(5101,-1336){\makebox(0,0)[lb]{\smash{{\SetFigFont{12}{14.4}{\rmdefault}{\mddefault}{\updefault}{\color[rgb]{0,0,0}$n_{B_1R}$}%
}}}}
\put(1276,-2986){\makebox(0,0)[lb]{\smash{{\SetFigFont{12}{14.4}{\rmdefault}{\mddefault}{\updefault}{\color[rgb]{0,0,0}$A_M$}%
}}}}
\end{picture}%

%% file: fig2.pstex_t
\begin{picture}(0,0)%
\epsfig{file=fig2.pstex}%
\end{picture}%
\setlength{\unitlength}{3947sp}%
\begingroup\makeatletter\ifx\SetFigFont\undefined%
\gdef\SetFigFont#1#2#3#4#5{%
  \reset@font\fontsize{#1}{#2pt}%
  \fontfamily{#3}\fontseries{#4}\fontshape{#5}%
  \selectfont}%
\fi\endgroup%
\begin{picture}(5756,2274)(1189,-3148)
\put(6376,-2986){\makebox(0,0)[lb]{\smash{{\SetFigFont{12}{14.4}{\rmdefault}{\mddefault}{\updefault}{\color[rgb]{0,0,0}$B_M$}%
}}}}
\put(3901,-2086){\makebox(0,0)[lb]{\smash{{\SetFigFont{12}{14.4}{\rmdefault}{\mddefault}{\updefault}{\color[rgb]{0,0,0}$R$}%
}}}}
\put(2551,-2686){\makebox(0,0)[lb]{\smash{{\SetFigFont{12}{14.4}{\rmdefault}{\mddefault}{\updefault}{\color[rgb]{0,0,0}$n_{A_MR}$}%
}}}}
\put(5101,-2686){\makebox(0,0)[lb]{\smash{{\SetFigFont{12}{14.4}{\rmdefault}{\mddefault}{\updefault}{\color[rgb]{0,0,0}$n_{B_MR}$}%
}}}}
\put(1351,-1186){\makebox(0,0)[lb]{\smash{{\SetFigFont{12}{14.4}{\rmdefault}{\mddefault}{\updefault}{\color[rgb]{0,0,0}$A_1$}%
}}}}
\put(6376,-1186){\makebox(0,0)[lb]{\smash{{\SetFigFont{12}{14.4}{\rmdefault}{\mddefault}{\updefault}{\color[rgb]{0,0,0}$B_1$}%
}}}}
\put(2551,-1411){\makebox(0,0)[lb]{\smash{{\SetFigFont{12}{14.4}{\rmdefault}{\mddefault}{\updefault}{\color[rgb]{0,0,0}$n_{A_1R}$}%
}}}}
\put(5101,-1336){\makebox(0,0)[lb]{\smash{{\SetFigFont{12}{14.4}{\rmdefault}{\mddefault}{\updefault}{\color[rgb]{0,0,0}$n_{B_1R}$}%
}}}}
\put(1276,-2986){\makebox(0,0)[lb]{\smash{{\SetFigFont{12}{14.4}{\rmdefault}{\mddefault}{\updefault}{\color[rgb]{0,0,0}$A_M$}%
}}}}
\end{picture}%

%% file: fig3.pstex_t
\begin{picture}(0,0)%
\epsfig{file=fig3.pstex}%
\end{picture}%
\setlength{\unitlength}{3947sp}%
\begingroup\makeatletter\ifx\SetFigFont\undefined%
\gdef\SetFigFont#1#2#3#4#5{%
  \reset@font\fontsize{#1}{#2pt}%
  \fontfamily{#3}\fontseries{#4}\fontshape{#5}%
  \selectfont}%
\fi\endgroup%
\begin{picture}(5592,3699)(1276,-3823)
\put(6376,-3361){\makebox(0,0)[lb]{\smash{{\SetFigFont{12}{14.4}{\rmdefault}{\mddefault}{\updefault}{\color[rgb]{0,0,0}$B_2$}%
}}}}
\put(3826,-1336){\makebox(0,0)[lb]{\smash{{\SetFigFont{12}{14.4}{\rmdefault}{\mddefault}{\updefault}{\color[rgb]{0,0,0}$R$}%
}}}}
\put(1276,-661){\makebox(0,0)[lb]{\smash{{\SetFigFont{12}{14.4}{\rmdefault}{\mddefault}{\updefault}{\color[rgb]{0,0,0}$A_1$}%
}}}}
\put(1276,-3286){\makebox(0,0)[lb]{\smash{{\SetFigFont{12}{14.4}{\rmdefault}{\mddefault}{\updefault}{\color[rgb]{0,0,0}$A_2$}%
}}}}
\put(6376,-736){\makebox(0,0)[lb]{\smash{{\SetFigFont{12}{14.4}{\rmdefault}{\mddefault}{\updefault}{\color[rgb]{0,0,0}$B_1$}%
}}}}
\end{picture}%

%% file: fig4.pstex_t
\begin{picture}(0,0)%
\epsfig{file=fig4.pstex}%
\end{picture}%
\setlength{\unitlength}{3947sp}%
\begingroup\makeatletter\ifx\SetFigFont\undefined%
\gdef\SetFigFont#1#2#3#4#5{%
  \reset@font\fontsize{#1}{#2pt}%
  \fontfamily{#3}\fontseries{#4}\fontshape{#5}%
  \selectfont}%
\fi\endgroup%
\begin{picture}(5592,3699)(1276,-3823)
\put(6376,-3361){\makebox(0,0)[lb]{\smash{{\SetFigFont{12}{14.4}{\rmdefault}{\mddefault}{\updefault}{\color[rgb]{0,0,0}$B_2$}%
}}}}
\put(3826,-1336){\makebox(0,0)[lb]{\smash{{\SetFigFont{12}{14.4}{\rmdefault}{\mddefault}{\updefault}{\color[rgb]{0,0,0}$R$}%
}}}}
\put(1276,-661){\makebox(0,0)[lb]{\smash{{\SetFigFont{12}{14.4}{\rmdefault}{\mddefault}{\updefault}{\color[rgb]{0,0,0}$A_1$}%
}}}}
\put(1276,-3286){\makebox(0,0)[lb]{\smash{{\SetFigFont{12}{14.4}{\rmdefault}{\mddefault}{\updefault}{\color[rgb]{0,0,0}$A_2$}%
}}}}
\put(6376,-736){\makebox(0,0)[lb]{\smash{{\SetFigFont{12}{14.4}{\rmdefault}{\mddefault}{\updefault}{\color[rgb]{0,0,0}$B_1$}%
}}}}
\end{picture}%

%% file: fig5.pstex_t
\begin{picture}(0,0)%
\epsfig{file=fig5.pstex}%
\end{picture}%
\setlength{\unitlength}{3947sp}%
\begingroup\makeatletter\ifx\SetFigFont\undefined%
\gdef\SetFigFont#1#2#3#4#5{%
  \reset@font\fontsize{#1}{#2pt}%
  \fontfamily{#3}\fontseries{#4}\fontshape{#5}%
  \selectfont}%
\fi\endgroup%
\begin{picture}(6192,3699)(1051,-3823)
\put(1051,-3286){\makebox(0,0)[lb]{\smash{{\SetFigFont{12}{14.4}{\rmdefault}{\mddefault}{\updefault}{\color[rgb]{0,0,0}$A_2$}%
}}}}
\put(3826,-1336){\makebox(0,0)[lb]{\smash{{\SetFigFont{12}{14.4}{\rmdefault}{\mddefault}{\updefault}{\color[rgb]{0,0,0}$R$}%
}}}}
\put(6676,-736){\makebox(0,0)[lb]{\smash{{\SetFigFont{12}{14.4}{\rmdefault}{\mddefault}{\updefault}{\color[rgb]{0,0,0}$B_1$}%
}}}}
\put(6751,-3361){\makebox(0,0)[lb]{\smash{{\SetFigFont{12}{14.4}{\rmdefault}{\mddefault}{\updefault}{\color[rgb]{0,0,0}$B_2$}%
}}}}
\put(6226,-2986){\makebox(0,0)[lb]{\smash{{\SetFigFont{12}{14.4}{\rmdefault}{\mddefault}{\updefault}{\color[rgb]{0,0,0}$b_{2,1}$}%
}}}}
\put(6226,-361){\makebox(0,0)[lb]{\smash{{\SetFigFont{12}{14.4}{\rmdefault}{\mddefault}{\updefault}{\color[rgb]{0,0,0}$b_{1,1}$}%
}}}}
\put(1501,-361){\makebox(0,0)[lb]{\smash{{\SetFigFont{12}{14.4}{\rmdefault}{\mddefault}{\updefault}{\color[rgb]{0,0,0}$a_{1,1}$}%
}}}}
\put(1501,-661){\makebox(0,0)[lb]{\smash{{\SetFigFont{12}{14.4}{\rmdefault}{\mddefault}{\updefault}{\color[rgb]{0,0,0}$a_{1,2}$}%
}}}}
\put(1501,-3286){\makebox(0,0)[lb]{\smash{{\SetFigFont{12}{14.4}{\rmdefault}{\mddefault}{\updefault}{\color[rgb]{0,0,0}$a_{2,1}$}%
}}}}
\put(2926,-3061){\makebox(0,0)[lb]{\smash{{\SetFigFont{12}{14.4}{\rmdefault}{\mddefault}{\updefault}{\color[rgb]{0,0,0}$y_R=\left[\begin{array}{ccc}  &  a_{1,1} &   \\  & a_{1,2}\oplus b_{1,1}  &   \\  & a_{2,1} \oplus b_{2,1}  &  \end{array}\right]$}%
}}}}
\put(1051,-661){\makebox(0,0)[lb]{\smash{{\SetFigFont{12}{14.4}{\rmdefault}{\mddefault}{\updefault}{\color[rgb]{0,0,0}$A_1$}%
}}}}
\end{picture}%

%% file: fig6.pstex_t
\begin{picture}(0,0)%
\epsfig{file=fig6.pstex}%
\end{picture}%
\setlength{\unitlength}{3947sp}%
\begingroup\makeatletter\ifx\SetFigFont\undefined%
\gdef\SetFigFont#1#2#3#4#5{%
  \reset@font\fontsize{#1}{#2pt}%
  \fontfamily{#3}\fontseries{#4}\fontshape{#5}%
  \selectfont}%
\fi\endgroup%
\begin{picture}(7017,3699)(601,-3823)
\put(2776,-3061){\makebox(0,0)[lb]{\smash{{\SetFigFont{12}{14.4}{\rmdefault}{\mddefault}{\updefault}{\color[rgb]{0,0,0}$x_R=\left[\begin{array}{ccc}  &  a_{2,1} \oplus b_{2,1} &   \\  & a_{1,2}\oplus b_{1,1}  &   \\  &  a_{1,1} &  \end{array}\right]$}%
}}}}
\put(3826,-1336){\makebox(0,0)[lb]{\smash{{\SetFigFont{12}{14.4}{\rmdefault}{\mddefault}{\updefault}{\color[rgb]{0,0,0}$R$}%
}}}}
\put(601,-661){\makebox(0,0)[lb]{\smash{{\SetFigFont{12}{14.4}{\rmdefault}{\mddefault}{\updefault}{\color[rgb]{0,0,0}$A_1$}%
}}}}
\put(601,-3286){\makebox(0,0)[lb]{\smash{{\SetFigFont{12}{14.4}{\rmdefault}{\mddefault}{\updefault}{\color[rgb]{0,0,0}$A_2$}%
}}}}
\put(1051,-961){\makebox(0,0)[lb]{\smash{{\SetFigFont{12}{14.4}{\rmdefault}{\mddefault}{\updefault}{\color[rgb]{0,0,0}$a_{1,2}\oplus b_{1,1}$}%
}}}}
\put(6226,-661){\makebox(0,0)[lb]{\smash{{\SetFigFont{12}{14.4}{\rmdefault}{\mddefault}{\updefault}{\color[rgb]{0,0,0}$a_{1,2}\oplus b_{1,1}$}%
}}}}
\put(6226,-961){\makebox(0,0)[lb]{\smash{{\SetFigFont{12}{14.4}{\rmdefault}{\mddefault}{\updefault}{\color[rgb]{0,0,0}$a_{1,1}$}%
}}}}
\put(1051,-3661){\makebox(0,0)[lb]{\smash{{\SetFigFont{12}{14.4}{\rmdefault}{\mddefault}{\updefault}{\color[rgb]{0,0,0}$a_{2,1}\oplus b_{2,1}$}%
}}}}
\put(7126,-736){\makebox(0,0)[lb]{\smash{{\SetFigFont{12}{14.4}{\rmdefault}{\mddefault}{\updefault}{\color[rgb]{0,0,0}$B_1$}%
}}}}
\put(6226,-3286){\makebox(0,0)[lb]{\smash{{\SetFigFont{12}{14.4}{\rmdefault}{\mddefault}{\updefault}{\color[rgb]{0,0,0}$a_{2,1}\oplus b_{2,1}$}%
}}}}
\put(7126,-3361){\makebox(0,0)[lb]{\smash{{\SetFigFont{12}{14.4}{\rmdefault}{\mddefault}{\updefault}{\color[rgb]{0,0,0}$B_2$}%
}}}}
\end{picture}%